\documentclass{article}
\usepackage{hyperref, graphicx}
\usepackage{amsmath}
\usepackage{tocloft}
\usepackage{amssymb}
\usepackage{slashed}
\usepackage{yfonts}
\usepackage[T1]{fontenc}
\usepackage[utf8]{inputenc}
\usepackage[english]{babel}
\usepackage{amsthm}
\usepackage{tikz-cd}
\usepackage{quiver}
\usepackage{mathabx}
\usepackage{stmaryrd}
\usepackage{braket}
\usepackage{physics}
\usepackage{makecell}
\usepackage{adjustbox}
\usepackage[most]{tcolorbox}
\usepackage{geometry}
\usepackage[
backend=biber,
style=numeric-comp,
sorting=none,
]{biblatex}
\usepackage{url}

\makeatletter
\newcommand{\ostar}{\mathbin{\mathpalette\make@circled\star}}
\newcommand{\make@circled}[2]{%
  \ooalign{$\m@th#1\smallbigcirc{#1}$\cr\hidewidth$\m@th#1#2$\hidewidth\cr}%
}
\newcommand{\smallbigcirc}[1]{%
  \vcenter{\hbox{\scalebox{0.77778}{$\m@th#1\bigcirc$}}}%
}
\makeatother

\newtheorem{theorem}{Theorem}[section]

\newtheorem{corollary}[theorem]{Corollary}
\newtheorem{proposition}[theorem]{Proposition}

\usepackage{authblk}
    \allowdisplaybreaks

\hypersetup{
pdfstartview = {FitH},
}
\hypersetup{
	colorlinks=true,         
	linkcolor=blue,          
	citecolor=red,        
	urlcolor=blue            
}

\DeclareMathAlphabet{\mathpzc}{OT1}{pzc}{m}{it}

\makeatletter
\def\@font@info#1{}
\makeatother

\theoremstyle{remark}
\newtheorem{remark}{Remark}[section]
\newenvironment{rmk}{\begin{remark}}{\hfill$\Diamond$\end{remark}}

\theoremstyle{definition}
\newtheorem{definition}[theorem]{Definition}

\usepackage{mathtools}

\usepackage{mathrsfs}

\numberwithin{equation}{section}

\newcommand{\beq}{\begin{eqnarray}}
\newcommand{\eeq}{\end{eqnarray}}

\def\g{\mathfrak g}

\def\h{\mathfrak h}

\def\G{\mathfrak{G}}

\def\bbZ{\mathbb{Z}}
\def\R{\mathbb{R}}
\def\bbC{\mathbb{C}}

\newcommand{\cA}{{\cal A}}
\newcommand{\cC}{{\cal C}}
\newcommand{\cD}{{\cal D}}
\newcommand{\cE}{{\cal E}}

\newcommand{\cI}{{\cal I}}

\newcommand{\cB}{{\cal B}}

\newcommand{\cV}{{\cal V}}

\makeatletter
\newsavebox{\@brx}
\newcommand{\llangle}[1][]{\savebox{\@brx}{\(\m@th{#1\langle}\)}%
  \mathopen{\copy\@brx\kern-0.5\wd\@brx\usebox{\@brx}}}
\newcommand{\rrangle}[1][]{\savebox{\@brx}{\(\m@th{#1\rangle}\)}%
  \mathclose{\copy\@brx\kern-0.5\wd\@brx\usebox{\@brx}}}
\makeatother

\makeatletter
\newcommand{\xRrightarrow}[2][]{\ext@arrow 0359\Rrightarrowfill@{#1}{#2}}
\newcommand{\Rrightarrowfill@}{\arrowfill@\equiv\equiv\Rrightarrow}
\newcommand{\xLleftarrow}[2][]{\ext@arrow 3095\Lleftarrowfill@{#1}{#2}}
\newcommand{\Lleftarrowfill@}{\arrowfill@\Lleftarrow\equiv\equiv}
\makeatother

\newcommand{\id}{\operatorname{id}}

\begin{document}

\title{Categorical quantum symmetries and ribbon tensor 2-categories}

\bigskip

\author[1]{{ \sf Hank Chen}\thanks{hank.chen@uwaterloo.ca \quad or\quad chunhaochen@bimsa.cn}}

\affil[1]{\small Beijing Institute of Mathematical Sciences and Applications, Beijing 101408, China}

\maketitle

\bigskip

\begin{abstract}
In a companion work on the combinatorial quantization of 4d 2-Chern-Simons theory, the author has constructed the Hopf category of quantum  2-gauge transformations $\tilde{C}=\mathbb{U}_q\G$ acting on the discrete surface-holonomy configurations on a lattice. We prove in this article that the 2-$\mathsf{Hilb}$-enriched 2-representation 2-category $\operatorname{2Rep}(\tilde C)$ of finite semisimple $\bbC$-linear $\tilde C$-module categories is braided, planar-pivotal, and \textit{lax }rigid, hence $\operatorname{2Rep}(\tilde C)$ provides an example of a \textit{ribbon tensor 2-category}. We explicitly construct the ribbon balancing functors, and exhibit their coherence conditions against the rigid dagger structures. This allows one to refine the various notions of \textit{framing} in a 2-category with duals that have been previously studied in the literature. Following the 2-tangle hypothesis of Baez-Langford, \textit{framed invariants} of 2-tangles can then be constructed from ribbon 2-functors into $\operatorname{2Rep}(\tilde C)$, analogous to the definition of decorated ribbon graphs in the Reshetikhin-Turaev construction. We will also prove that, in the classical limit $q\rightarrow 1$, the 2-category $\operatorname{2Rep}(\mathbb{U}_{q=1}\G)$ becomes \textit{strict} pivotal in the sense of Douglas-Reutter. 


\end{abstract}

\bigskip


\tableofcontents

\newpage

\section{Introduction}
Over the past century, the algebraic structures describing the symmetries of various physical systems of interest were found to be captured by \textit{quantum groups}. Examples include the XXX/XXY/XYZ family of integrable spin chains \cite{Jones:1990} and the Wilson loop observables of 3-dimensional Chern-Simons theory \cite{WITTEN1983422}. \textit{Compact} quantum groups \cite{Majid:1996kd,Woronowicz1988}, in particular, such as the quantum enveloping algebra $U_q\g$ of Drinfel'd-Jimbo \cite{Drinfeld:1986in,Jimbo:1985zk} associated a semisimple Lie algebra $\g$, play a very important role in these physical examples. These ideas were then extended substantially over the past few decades to construct various types of topological quantum field theories (TQFTs). Furthermore, quantum group symmetry can also be found \cite{Alekseev:1992wn,KazhdanLusztig:1994} hidden within the algebra of current operators of the conformal field theory living on the \textit{boundary} of such 3-dimensional TQFTs.

Subsequently, it was discovered that the representation theory of quantum group Hopf algebras was able to give rise to invariants of 3-manifolds \cite{Reshetikhin:1990pr,Turaev:1992,BalsamKirillov:2012KitaevTuraevViro,Barrett1993,reshetikhin2010lecturesintegrability6vertexmodel}. This jump-started the field of \textit{quantum topology}, in which methods of theoretical physics --- namely field theory and gauge theory --- was applied to study the topology of 3-manifolds. The celebrated \textit{Reshetikhin-Turaev TQFT} can be thought of as the cornerstone of the entire field. One of, if not \textit{the}, reason that quantum groups and their representations turned out to play such a crucial role in the topology of 3-manifolds is that (i) Hopf algebra representations encode algebraically the properties of {\it (1-)tangles} in 3-space, and (ii) the geometry of 1-tangles encode the topology of 3-manifolds. 

To be more precise, 1-tangles are collections of copies of the interval $[0,1]$ embedded into a slice $\R^2\times[0,1]$, such that their endpoints are located at $\R^2\times\{0,1\}$. The ambient isotopies of these 1-tangles fixing their endpoints are well-known to be generated by the so-called \textit{Reidemeister moves}. A link (ie. a collection fo embedded copies of $S^1$) can then be obtained by attaching the top and the bottom endpoints of a 1-tangle. It was in fact known that invariants of links and knots can be obtained from Hopf algebra quantum groups \cite{Jones:1990}, prior to the construction of quantum 3-manifold invariants.

A classic theorem of Lickorish-Wallace \cite{Lickorish1962ARO,Wallace_1960} states that any closed compact 3-manifold can be obtained from the 3-sphere $S^3$ by a procedure known as \textit{surgery theory}. Briefly, one can take an embedded framed link in $S^3$, excise its tubular neighborhood, perform a so-called "Dehn filling" on the excised torii, then glue it back to obtain another 3-manifold. The link equivalences under which the resulting 3-manifolds are diffeomorphic are known as the \textit{Kirby moves} \cite{Kirby1978-pe}. These two aspects of 3-manifold topology were combined in the seminal work of Reshetikhin and Turaev \cite{Reshetikhin:1990pr,Reshetikhin:1991tc}, where quantum invariants of 3-manifolds were obtained by decorating the surgery links in $S^3$ with the data of a ribbon tensor category, such as $\operatorname{Rep}(U_q\frak{sl}_2)$, the representation category of the quantum enveloping algebra.

\bigskip

Now over the past decade, significant efforts have been dedicated to investigating the higher-dimensional analogue of the above phenomenon --- a "4-dimensional categorified quantum topology" of sorts. The success of the cobordism hypothesis \cite{Baez:1995xq,lurie2008classification} of Baez-Dolan to classify higher-dimensional TQFTs led many to explore the homotopy properties of higher-dimensional analogues of tangles --- the so-called "2-tangles" --- through higher categorical algebras \cite{BAEZ2003705}. On the other hand, several 4-dimensional field/gauge theories were constructed throughout the late 20th to early 21st century, which gave rise to very interesting invariants that can detect exotic smooth structures. Examples include (but may not be limited to) the Donaldson invariant \cite{DONALDSON1990257}, the Seiberg-Witten invariant \cite{Moore:2001}, the Rozansky-Khovanov homology \cite{Morrison2019InvariantsO4}, and the Kontsevich integral \cite{Chmutov2001-gi}. However, aside from the Khovanov homology story (see \cite{Manolescu2022SkeinLM}), it is not yet clear if these are related to the homotopy theory of 2-tangles. 

In order to complete the analogy with the Chern-Simons/Reshetikhin-Turaev TQFT, the author has opted to begin from the perspective of higher-gauge theory. Based on the theory of derived $L_n$-algebras and their associated \textit{non-Abelian} bundle gerbes (with connection) \cite{Baez:2012,Wockel2008Principal2A,Nikolaus2011FOUREV,Schommer_Pries_2011,Kim:2019owc}, one can develop the so-called "homotopy Maurer-Cartan theories" \cite{Jurco:2018sby}, which can be understood as higher-dimensional generalizations of Chern-Simons theory in the derived context. At dimension 4, in particular, the homotopy "2-Chern-Simons" theory has relatively recently received attention in both the context of physics and mathematics \cite{Soncini:2014,Zucchini:2021bnn,Chen:2024axr,Schenkel:2024dcd}. Particularly in the companion work \cite{Chen1:2025?} by the author, a framework for the combinatorial quantization of this 4d 2-Chern-Simons theory was developed on a lattice, in which the analytic and Hopf categorical structures of the underlying surface holonomy degrees-of-freedom were unraveled. 

\medskip

This paper is dedicated to the detailed categorical study of the (finite semisimple linear) 2-representations of the categorical gauge symmetries $\tilde C$ in 2-Chern-Simons theory. Physically, they correspond to (a local algebraic description of) the Wilson loop and surface observables (see \cite{Kapustin:2013uxa,Chen2z:2023}), and mathematically they form a 2-category denoted by $\operatorname{2Rep}(\tilde{C};\tilde R)$. One of the results in \cite{Chen1:2025?} proves that $\tilde C$ is a \textit{Hopf monoidal category} (cf. \cite{Crane:1994ty,DAY199799,batista2016hopf}), and $\tilde R$ is a so-called "cobraiding" which behaves like a categorical version of a quantum 2-$R$-matrix. More details can be found in \S \ref{lattice2gau}.

In this paper, we shall focus entirely on the braiding, adjunctions, and duals of $\operatorname{2Rep}(\tilde C;\tilde R)$. We will describe all of the coherence conditions of the  \textit{ribbon balancings}, such that a notion of a "ribbon 2-functor" can be understood as a 2-functor between two rigid dagger braided tensor 2-categories preserving these ribbon balancings. By leveraging the 2-tangle 2-category of Baez-Langford \cite{BAEZ2003705} $\mathcal{T}$, the notion of {\bf decorated ribbon 2-tangles} can then be defined as a ribbon 2-functor
\begin{equation*}
    \mathcal{T}\rightarrow \operatorname{2Rep}(\tilde{C};\tilde R).
\end{equation*}
Indeed, the 2-tangle hypothesis then dictates that such a functor would determine the quantum 2-Chern-Simons theory as a functorial 4d TQFT, whose  quantum invariant on a closed 4-manifold can in principle be constructed through a 4-dimensional version of the Reshetikhin-Turaev functor.

\subsection{Summary of results}
We begin by mentioning some previous works in the literature which sought to capture the geometry of 2-tangles using 2-categorical notions. We will also refer to places in this paper where generalizations and refinements of these results can be found.
\begin{enumerate}
    \item The earliest work on this, to the best of the author's knowledge, is \cite{BAEZ2003705}. In this paper, a "braided 2-category with duals" is introduced, which serves to algebraically model the geometric and homotopical aspects of tangles embedded in 4-space. This 2-tangle 2-category also has a single self-dual generator, which is "unframed" in the sense that it has equipped a trivialization of the first Reidemeister move. As we will discuss in \S \ref{higherdagger}, the notion of duality suffices for unframed objects, but not in general. We will explain the situation of \cite{BAEZ2003705} in the context of our paper in \S \ref{baezlangfordassumption}.
    
    \item In the seminal work of Douglas-Reutter \cite{Douglas:2018}, they introduced the notion of semisimple and fusion 2-categories, as well as, subsequently, pivotality and sphericality. In defining pivotality, two distinct notions of duality and rigidity were introduced: planar-rigidity and object-level rigidity. The object-level dual was defined to be strictly involutive/reflexive, which allowed a certain pivotal condition to be imposed. In \S \ref{rigidagger}, we identify this particular pivotal condition as the main culprit for the drawback of their framework mentioned in Warning 2.2.5 in \cite{Douglas:2018}. We will demonstrate in \S \ref{ribbon} how the \textit{ribbon balancing} underlying our 2-category can resolve this issue, and "unstrictify" the pivotality of Douglas-Reutter to a notion that may be called "$SO(3)$-volutvity" \cite{ferrer2024daggerncategories}.
    
    \item Over the past decade, the properties of a $\mathsf{Gray}$-categories equipped with duals had been under study \cite{Barrett_2024}. The author believes that this framework is the closest one to this paper, due to the fact that $\mathsf{Gray}$-categories with duals are the natural algebraic description of the local part of non-extended 3d defect TQFTs \cite{Carqueville:2016kdq}. We shall see in \S \ref{secondreide} and \S \ref{reidemeisterI} how many of the structures appearing \cite{Barrett_2024}) also appears in $\operatorname{2Rep}(\tilde{C};\tilde R)$. 
\end{enumerate}
Much of the writing of this paper have taken significant inspiration from the above cited papers, and we will make references to them frequently whenever appropriate. The "main result" here, so to speak, is the definition of a more refined notion of "framing" for objects in a braided 2-category in \S \ref{reidemeisterI}. 

\begin{table}[h]\hspace*{-1cm}
    \centering
    \begin{tabular}{c|c|c|c|c}
         &  fully-framed  & half-framed & unframed & self-dual \\
         \hline
       \makecell{$\operatorname{2Rep}(\tilde{C};\tilde R)$ \\ (or its delooping)} & \makecell{ribbon tensor \\ 2-category} & \makecell{$\mathsf{Gray}$-category \\ with duals \\ \cite{Barrett_2024} \\ (with monoidal product)} & \makecell{pivotal 2-category \\ \cite{Douglas:2018}\\ (with braiding)} & \makecell{2-tangle 2-category \\ with one self-dual generator \\ \cite{BAEZ2003705}}
    \end{tabular}
    \caption{A schematic table displaying the various notions of "framed-ness" in $\operatorname{2Rep}(\tilde{C};\tilde R)$ and how they relate to existing structures that have already appeared in the literature.}
    \label{tab:framing}
\end{table}

We show that there are in fact {\it four} levels of "framed-ness" (see tables \ref{tab:framing} and \ref{tab:framed}), each of which correspond to data that trivializes the duality and adjunction structures to a certain degree. We will in particular note in \textit{Remark \ref{strangesnakes}} how being "half-framed" is closely related to the structures studied in \cite{Barrett_2024}. 

\medskip

    A synoptic diagram organizing all of the 2-categorical structures involved in this paper are exhibited below (as inspired by \cite{Pennys:2016}).
{\small \[\hspace*{-2.4cm}\begin{tikzcd}
	\begin{array}{c} \text{monoidal with} \\ \text{adjoints for 1-morphisms} \end{array} & \text{braided} && \begin{array}{c} \text{balanced} \\ \vartheta,\bar\vartheta: 1_{\tilde C}\Rightarrow1_{\tilde C}\\ e:\vartheta^\dagger\circ\vartheta\Rrightarrow \id \end{array} \\
	{\text{rigid}} & {\text{braided rigid}} && {\text{balanced rigid}} \\
	{\text{lax-pivotal}} & {\text{braided lax-pivotal}} && \begin{array}{c} \textbf{2-ribbon}/SO(3)\text{-volutive} \\ \omega_\cD:\bar\vartheta_\cD^*\Rightarrow \vartheta_{\cD^*} \end{array} \\
	& {\text{pivotal}} & {\text{braided pivotal}} && \begin{array}{c} \text{strict 2-ribbon} \\ \omega_\cD: \vartheta_\cD^*\Rightarrow \vartheta_{\cD^*} \end{array} \\
	{\text{lax-spherical}} & {"\text{braided lax-spherical}"} && {\text{"2-modular"}} \\
	& {\text{spherical}} & {\text{braided spherical}} && {"\text{strict 2-modular}"}
	\arrow["{Z_1}"', curve={height=12pt}, dashed, from=1-1, to=1-2]
	\arrow[from=1-2, to=1-1]
	\arrow[from=1-4, to=1-2]
	\arrow[from=2-1, to=1-1]
	\arrow["{Z_1}"', curve={height=12pt}, dashed, from=2-1, to=2-2]
	\arrow[from=2-2, to=1-2]
	\arrow[from=2-2, to=2-1]
	\arrow[from=2-4, to=1-4]
	\arrow[from=2-4, to=2-2]
	\arrow["\begin{array}{c} \text{half-framed} \\ \bar\vartheta\cong \vartheta^\dagger \end{array}"{description, pos=0.6}, squiggly, tail reversed, from=2-4, to=3-2]
	\arrow["\begin{array}{c} \text{unframed} \\ \bar\vartheta\cong \vartheta \end{array}"{description, pos=0.8}, curve={height=18pt}, squiggly, tail reversed, from=2-4, to=4-3]
	\arrow[from=3-1, to=2-1]
	\arrow["{Z_1}"', curve={height=12pt}, dashed, from=3-1, to=3-2]
	\arrow[from=3-2, to=2-2]
	\arrow[from=3-2, to=3-1]
	\arrow[from=3-4, to=2-4]
	\arrow[from=3-4, to=3-2]
	\arrow[from=4-2, to=2-1]
	\arrow[from=4-2, to=3-1]
	\arrow["{Z_1}"'{pos=0.4}, curve={height=12pt}, dashed, from=4-2, to=4-3]
	\arrow[from=4-3, to=3-2]
	\arrow[from=4-3, to=4-2]
	\arrow[from=4-5, to=3-4]
	\arrow[from=4-5, to=4-3]
	\arrow[from=5-1, to=3-1]
	\arrow["{Z_1}"'{pos=0.7}, curve={height=12pt}, dashed, from=5-1, to=5-2]
	\arrow[dotted, from=5-2, to=3-2]
	\arrow[from=5-2, to=5-1]
	\arrow[from=5-4, to=3-4]
	\arrow[from=5-4, to=5-2]
	\arrow[curve={height=-6pt}, from=6-2, to=4-2]
	\arrow[from=6-2, to=5-1]
	\arrow["{Z_1}"', curve={height=12pt}, dashed, from=6-2, to=6-3]
	\arrow[from=6-3, to=4-3]
	\arrow[from=6-3, to=5-2]
	\arrow[from=6-3, to=6-2]
	\arrow[from=6-5, to=4-5]
	\arrow[from=6-5, to=5-4]
	\arrow[from=6-5, to=6-3]
\end{tikzcd}\]}
where the arrows $\rightarrow$ are forgetful functors and $Z_1$ denotes taking the Drinfel'd centre. The various refined notions of "framing" described in Table \ref{tab:framing} can be attributed to the interaction of the ribbon balancing $\vartheta$ with \textit{lax} structures of rigidity.

\medskip

The new insight here is that all of these geometric structures, including the old known ones (eg. planar-pivotality {\bf Theorem \ref{planaruni}} and "2-category with duals" \textbf{Theorem \ref{withduals}}) as well as the new ones (eg. the ribbon balancing structures \S \ref{ribbon} and the higher-Hopf links \S \ref{hopflinks}), were extracted from studying properties of $\operatorname{2Rep}(\tilde{C};\tilde R)$ --- namely the observables in quantum 2-Chern-Simons theory. Though not completely general, this perspective has the advantage that it allowed us to pinpoint exactly when the 2-category $\operatorname{2Rep}(\tilde{C};\tilde R)$ has, for instance,
\begin{itemize}
    \item not just braided/$E_2$-structure but a \textit{sylleptic}/$E_3$-structure (see \textit{Remark \ref{syllepsis}}), and
    \item a braiding of finite-order (see \textit{Remark \ref{braidingorder}}).
\end{itemize} 
Moreover, we will show in \S \ref{classicalpivotal} that in the undeformed classical limit, we recover $\operatorname{2Rep}(\tilde{C}\mid_{q=1};\id\otimes\id)$ as a symmetric (namely $E_4=E_\infty$) 2-category equipped with a pivotal structure in the sense of \cite{Douglas:2018}. The triviality of the {\it quadruple object-level dual} will be proven by the author and collaborators in a soon upcoming work (see also \textit{Remark \ref{categoricalS4}}).

\subsection{Overview}
We will begin in \S \ref{lattice2gau} with a concise review of what we mean by a "Hopf category $\tilde{C}$", and how such an algebraic structure appear in the quantum symmetries of 2-Chern-Simons theory \cite{Chen1:2025?}. It models a categorical version of the quantum enveloping algebra. Then, in \S \ref{2rep}, we will use its Hopf categorical structures to determine the braided monoidal structures of the 2-category $\operatorname{2Rep}(\tilde{C};\tilde R)$ of its finite semisimple linear 2-representations, by leveraging previous works \cite{neuchl1997representation,GURSKI20114225,KongTianZhou:2020,Chen:2023tjf}.

Then, the strategy is as follows: 
\begin{enumerate}
    \item We introduce the adjoints and duals in \S \ref{adjoints} and \S \ref{duals}, respectively, then we study their mutual compatibility in \S \ref{rigidagger}. This led us to the notion of a "rigid dagger tensor 2-category", and we use this structure in \S \ref{2catdimension} to unveil the main cause of the issue behind \textit{Warning 2.2.5} of \cite{Douglas:2018}.
    \item Then, in \S \ref{adj-dualbraid} and \S \ref{dualitybraid}, we include the braiding into the discussion. By examining the planar-unitarity of the braiding, we recover the notion of "braided 2-category with duals" described in \cite{BAEZ2003705}, as well as its \textit{writhing} and the \textit{fold-crossings} coherence 2-morphisms \cite{BAEZ2003705,CARTER19971}.
\end{enumerate}
These data and properties make $\operatorname{2Rep}(\tilde{C};\tilde R)$ into a \textit{ribbon tensor 2-category}. In \S \ref{ribbon}, we introduce the ribbon balancing from the above braided rigid structure. These ribbon balancings are used to define various notions of "framing" of an object, organized in table \ref{tab:framed}. We showed how these notions of framing reduce to those described in previous literature, as listed in table \ref{tab:framing}.

Next, in \S \ref{hopflinks}, we studied and constructed the \textit{Hopf link} functors. Some were found to be trivializable, and we describe the coherence conditions they satisfy. Further, we also list in table \ref{tab:hopf} the different types of Hopf links that one can form depending on the framing.

Finally, in \S \ref{classicalpivotal}, we prove that $\operatorname{2Rep}(\tilde{C};\tilde R)$ becomes symmetric and pivotal in the classical limit. In fact, every object becomes "unframed".


\subsubsection*{Statements and conflicts of interest declaration} 
The author (HC) is affiliated with and supported by Beijing Institute of Mathematical Sciences and Applications (BIMSA). There are no conflicts of interest to declare at this time. 

\subsubsection*{Data availability statement}
There are no data associated with this work.

\subsubsection*{Acknowledgments}
The author (HC) would like to extend his thanks to Yilong Wang, Nils Carqueville, Zhi-Hao Zhang, Jinsong Wu and David Green for lengthy and very enlightening discussions throughout the completion of this work. He would also like to thank an anonymous referee for the very helpful comments.

\section{Preliminaries on Hopf monoidal categories}\label{lattice2gau}
We now begin with a brief introduction on the algebraic structures in this work. Let $(\tilde {C},\otimes,\tilde I)$ denote an (additively complete $\bbC$-linear) monoidal category with an unit object $\tilde I$ (see  \cite{etingof2016tensor}). As is conventional, we denote by $\boxtimes$ the symmetric monoidal Delign tensor product on linear categories.

We say $\tilde{C}$ is a \textbf{(strict) bimonoidal category} (cf. \cite{Crane:1994ty,DAY199799,batista2016hopf}) iff it is equipped with 
\begin{enumerate}
    \item a strictly coassociative coproduct functor $\tilde\Delta: \tilde C\rightarrow \tilde C\boxtimes \tilde C$, and
    \item a counit functor $\epsilon: \tilde C\rightarrow \mathsf{Vect}$,
\end{enumerate}
such that there are invertible natural transformations
\begin{gather}
     \tilde \Delta\circ (-\otimes -) \cong (-\otimes - \boxtimes -\otimes -) \circ (1_{\tilde C }\otimes \sigma\otimes 1_{\tilde C }) \circ(\tilde\Delta \boxtimes\tilde\Delta)\label{bimonoid}\\ 
     (\tilde\epsilon\boxtimes 1_{\tilde C })\circ \tilde\Delta \cong 1_{\tilde C },\qquad (1_{\tilde C }\boxtimes\tilde\epsilon)\circ\tilde\Delta \cong 1_{\tilde C}\nonumber
\end{gather}
which satisfy the obvious coherence conditions against the coassociator, 
\begin{equation*}
    (\tilde\Delta\boxtimes 1_{\tilde C })\circ\tilde\Delta \cong (1_{\tilde C }\boxtimes\tilde \Delta)\circ\tilde\Delta ,
\end{equation*}
where $\sigma$ swaps the Deligne tensor factors. 

Given a bimonoidal category $(\tilde C,\otimes,\tilde I,\tilde \Delta,\tilde \epsilon)$, we now introduce an \textbf{antipode functor} $\tilde S: \tilde C\rightarrow \tilde C^\text{m-op,c-op}$ into the monoidal/comonoidal opposite, for which there are invertible natural transformations
\begin{equation}
    (-\otimes- )\circ(\tilde S\boxtimes 1_{\tilde C})\circ\tilde \Delta \cong \tilde I\otimes \tilde \epsilon\cong (-\otimes- )\circ(1_{\tilde C}\boxtimes \tilde S)\circ\tilde \Delta,\label{antpod}
\end{equation}
together with the appropriate coherence relations (cf. a "Hopf algebroid" of \cite{DAY199799}). We are now in position to introduce the central algebraic gadget in this paper.
\begin{definition}\label{cobrhopfdaggercat}
    We say $(\tilde C,\otimes,\tilde I,\tilde\Delta,\tilde \epsilon,\tilde S;\tilde R)$ is a \textbf{cobraided Hopf category} iff $\tilde C$ is equipped with a comonoidal natural transformation $\tilde R: \tilde\Delta \Rightarrow \tilde\Delta^\text{op}$. We say $\tilde C$ is \textbf{strictly cobraided} iff $\tilde R$ is invertible, and we say $\tilde C$ is a \textbf{cobraided Hopf dagger category} iff it is equipped with an involutive dagger structure $-^\dagger: \tilde C\rightarrow \tilde C^\text{op}$ for which the functors $\otimes,\tilde\Delta,\tilde S$ are dagger, and their accompanying invertible coherence 2-morphisms are unitary.
\end{definition}
\noindent In the rest of this paper, we will also often denote $-^\dagger =\tilde S_v$ by the dagger structure.

\begin{rmk}\label{hopfalgd}
    Similarly, a \textbf{Hopf cocategory} $C$ is a linear additive \textit{co}monoidal \textit{co}category equipped with a compatible monoidal structure $-\otimes-$ --- and thus making it also into a linear additive bimonoidal category --- as well as the appropriate antipodes. A closely related notion of a \textit{Hopf (op)algebroid} was studied in \cite{DAY199799}.
\end{rmk}

The goal of this work is to focus on the structure of the \textit{2-representations} of a strictly cobraided Hopf dagger category $(\tilde C;\tilde R)$ --- or just briefly a Hopf category in the following --- the structures of which were shown in \cite{Chen1:2025?} to naturally appear in the 4-dimensional \textit{2-Chern-Simons theory}.

\begin{tcolorbox}[breakable]
    \paragraph{Cobraiding from 2-$R$-matrices.} Let us unpack the meaning of the \textit{cobraiding} $\tilde R: \tilde \Delta\Rightarrow \tilde \Delta^\text{op}$, with a specific instance of its construction. Consider the following ingredients:
    \begin{enumerate}
        \item a \textbf{2-$R$-matrix} $\tilde R\in \tilde C\boxtimes\tilde C$, which satisfies the \textit{2-Yang-Baxter relations} (cf. \cite{Chen:2023tjf,Chen1:2025?})
\begin{equation}
    (\tilde\Delta\boxtimes 1)\tilde R = \tilde R^{13}\otimes\tilde R^{12},\qquad (1\boxtimes\tilde \Delta)\tilde R = \tilde R^{13}\otimes \tilde R^{23},\label{2yb}
\end{equation}
\item an invertible 1-morphism witnessing the intertwining relations
\begin{equation*}
    \tilde R'_\zeta :\tilde\Delta^\text{op}(\zeta)\otimes \tilde R \cong \tilde R\otimes \tilde \Delta(\zeta),
\end{equation*}
which is natural in $\zeta\in\tilde C$.
    \end{enumerate}
Then, consider a natural isomorphism $\tilde R: \tilde\Delta\Rightarrow \tilde\Delta^\text{op}$ whose components are given by
\begin{equation*}
    \tilde R_\zeta = \tilde R'_\zeta \circ\operatorname{ad}_{\tilde R}: \tilde\Delta(\zeta)\rightarrow \tilde\Delta^\text{op}(\zeta),
\end{equation*}
where $\operatorname{ad}_{\tilde R}=\tilde R \otimes -\otimes \tilde R^{-1}$ is the conjugation action by the 2-$R$-matrix $\tilde R$.\footnote{It is worth mentioning here that the notion of a \textit{universal} 2-$R$-matrix, for the so-called "Hopf 2-algebras" \cite{Wagemann+2021,Majid:2012gy,Lu:1996}, was defined in \cite{Chen:2023tjf}. The conditions \eqref{2yb}, \eqref{2ndreidemeister} has all be derived in that context.} 

It can then be directly deduced that the 2-Yang-Baxter \eqref{2yb} implies the comonoidality of $\tilde R$. Moreover, together with the antipode $\tilde S$, there exist natural isomorphisms
\begin{equation}
    (\tilde S\boxtimes 1)\tilde R \otimes \tilde R \cong \tilde I\boxtimes\tilde I,\qquad \tilde R\otimes (1\boxtimes \tilde S)\tilde R\cong\tilde I\boxtimes\tilde I\label{2ndreidemeister}
\end{equation}
coming from $\tilde R'$ which witnesses the quasitriangularity condition. This specific construction of a strict cobraiding $\tilde R$ is in fact what appears in the quantization of 2-Chern-Simons theory; see \cite{Chen1:2025?} and \textit{Remark \ref{quantumcobraiding}}.
\end{tcolorbox}

In the next subsection, we will give a brief review of how the structure of a strictly cobraided Hopf dagger category appears in the context of a 4-dimensional topological gauge-field theory. The uninterested readers can skip directly to \S \ref{2rep}.

\subsection{Categorical gauge transformations as a Hopf category}
This general sentiment --- namely that higher-categorical structures should appear in higher-dimensional physics --- were well-understood \cite{Baez:1995xq} and has seen many successful applications in the past, we here focus on the explicit higher categorical algebra which appears in the \textit{2-Chern-Simons theory}.

By "2-Chern-Simons theory", we refer to the topological higher-gauge theory \cite{Jurco:2018sby,Soncini:2014,Chen:2024axr} based on a structure (strict) Lie 2-group $\mathbb{G} = \mathsf{H}\xrightarrow{t}G$ \cite{Baez:2004,Porst2008Strict2A}. On a 4-manifold $M^4$, the classical action
\begin{equation*}
    S_{2CS}[A,B] = \int_{M^4}\langle B,F_A-\frac{1}{2}\mu_1B\rangle,\qquad \begin{cases}
        A\in \Omega^1(M^4,\g) \\ 
        B\in\Omega^2(M^4,\h)
    \end{cases}
\end{equation*}
involves a tuple of gauge fields $(A,B)$ valued in the Lie 2-algebra associated \cite{Chen:2012gz} to $\mathbb{G}$, which constitute (on-shell flat) \textit{$\mathbb{G}$ 2-connection} on $M^4$ \cite{schreiber2013connectionsnonabeliangerbesholonomy,Baez:2012,Martins:2006hx,Kim:2019owc}. It is worth mentioning briefly here that 2-Chern-Simons theory, or higher-gauge theory in general, has been known to be deeply relevant to various fields of physics \cite{Chen:2022hct,Ritter:2016,Mikovic:2015hza,Mikovic:2011si,Girelli:2021khh,Baez:1995ph,Baez:2002highergauge,Benini:2018reh,Cordova:2018cvg,Song_2023,Song:2021,Gaiotto:2014kfa,Wen:2019,Kong:2020wmn,Wang:2016rzy,Delcamp:2018kqc,Bullivant:2016clk,Bochniak_2021,Kapustin:2013uxa,Thorngren2015,Dubinkin:2020kxo,Kim:2019owc,Sati:2009ic,Schreiber:2013pra}.

\begin{rmk}
A Lie 2-group $\mathbb{G}$ can also be described in terms of a groupoid $$\mathsf{H}\rtimes G\rightrightarrows G,\qquad  a\xrightarrow{(a,\gamma)} at(\gamma),$$ which comes equipped invertible horizontal (group) and vertical (groupoid) multiplications \cite{Chen:2012gz}
\begin{equation*}
    (a,\gamma)\cdot (a',\gamma') = (aa',\gamma(a\rhd \gamma')),\qquad (a,\gamma)\circ (at(\gamma),\gamma')= (a,\gamma\gamma'),
\end{equation*}
with the units given by $(1,{\bf 1}_1)$ and ${\bf 1}_a$ for all $a\in G.$ They are compatible through the so-called {\it interchange law} \cite{Baez:2003fs}
\begin{equation*}
    ((a_1,\gamma_1)\cdot (a_2,\gamma_2))\circ ((a_3,\gamma_3)\cdot (a_4,\gamma_4))=((a_1,\gamma_1)\cdot (a_3,\gamma_3))\circ ((a_2,\gamma_2)\cdot (a_4,\gamma_4)),
\end{equation*}
where $(a_1,\gamma_1),\dots,(a_4,\gamma_4)\in\mathbb{G}$ are appropriately composable 2-group elements. In this paper, we will often refer to the horizontal multiplication $\cdot$ the "product", and the vertical multiplication $\circ$ the "composition".
\end{rmk}

Through the theory of principal 2-bundles \cite{Wockel2008Principal2A}, the surface holonomies associated to the principal $\mathbb{G}$ 2-bundle on $M^4$ can be constructed as a \textit{2-functor} which assigns an element of $\mathbb{G}$ to a surface-boundary pair $(\Sigma,\gamma)$,
\begin{equation*}
    \operatorname{2Hol}_\mathbb{G}(\Sigma,\gamma) = (W_\Sigma,V_\gamma) \in \mathsf{H}\rtimes G,\qquad tW_\Sigma = V_{\partial\Sigma}.
\end{equation*}
For an explicit construction of these 2-holonomies, see eg. \cite{Yekuteli:2015,Baez:2012,Chen:2024axr}.

\medskip

In \cite{Chen1:2025?} (and reviewed in detail in \cite{chen2:2025}), the \textit{combinatorial} quantization framework for Hamiltonian 2-Chern-Simons theory was developed, taking inspiration from the works \cite{Alekseev:1994pa,Alekseev:1994au} for 3d Chern-Simoms theory. By "combinatorial", we refer to the discretized 2-holonomies $\mathbb{G}^{\Gamma^2}$ \cite{Bullivant:2016clk,Bochniak:2020vil,Bochniak_2021}, namely assignments $(b_f,h_e) \in \mathsf{H}\rtimes G$ of 2-group elements to polygonal face-edge pairs $(e,f)\in\Gamma^2$, on the 2-skeleton $\Gamma^2$ of a simplicialization of a Cauchy slice (ie. a time-like 3-dimensional submanifold of $M^4$). See also \cite{Kapustin2017} for finite 2-groups.

These discrete 2-holonomies were found to inherit a natural action by \textit{2-gauge transformations}. Such 2-gauge transformations are parameterized by assignments of $\mathbb{G}$ to 1-simplices,
\begin{equation*}
    \mathbb{G}^{\Gamma^1} = \Big\{ a_{v}\xrightarrow{(a_v,\gamma_e)}a_{v'}\mid v\xrightarrow{e}v'\in \Gamma^1\Big\}
\end{equation*}
with directed edges $e$ with source/target vertices $v,v'$. This is a monoidal groupoid, whose composition is given by the confluence of the structures on $\Gamma^1$ and $\mathbb{G}$,
\begin{equation*}
    a_v\xrightarrow{(a_v,\gamma_e)}a_{v'}\xrightarrow{(a_{v'},\gamma_{e'})}a_{v''} = a_v\xrightarrow{(a_v,\gamma_e\gamma_e')}a_{v''},
\end{equation*}
and whose monoidal structure is invertible. In the following, we will simplify our notation and neglect the reference to the source vertex decoration $a_v$ in the 1-morphisms, and denote by
\begin{equation*}
    \zeta = a_v\xrightarrow{\gamma_e}a_{v'}
\end{equation*}
a decorated 1-graph in $\mathbb{G}^{\Gamma^1}$.

The natural 2-gauge transformation action of $\mathbb{G}^{\Gamma^1}$ on $\mathbb{G}^{\Gamma^2}$ is given by the horizontal/group conjugation:
\begin{equation*}
    ((a_v,\gamma_e),(h_e,b_f))\mapsto (h'_e,b'_f) = \operatorname{hAd}_{(a_v,\gamma_e)}^{-1}(h_e,b_f).
\end{equation*}
More details on this can be found in the companion paper \cite{Chen1:2025?}.\footnote{This action in fact makes the 2-functor category 
\begin{equation*}
    \operatorname{Fun}(\Gamma^{2\leq},\mathbb{G}) \simeq \mathbb{G}^{\Gamma^2}//\mathbb{G}^{\Gamma^1}
\end{equation*}
into an action 2-groupoid.}

\subsection{Quantization of the categorical gauge symmetry}
The above describes the natural action of the decorated 1-graphs $\mathbb{G}^{\Gamma^1}$ by 2-gauge transformations on the discrete 2-holonomies. This action was then extended in particular to \textit{states} on these 2-holonomies, which can be understood as additive categorical functionals on the combinatorial 2-holonomy configurations. Details of this construction can be found in \cite{Chen1:2025?}. 

The upshot is that such 2-holonomy states were modelled --- in the framework of the \textit{meausreable categories} of Crane-Yetter \cite{Crane:2003gk,Yetter2003MeasurableC,Baez:2012} --- as a certain additive Hopf \textit{co}category  denoted by $\mathfrak{C}(\mathbb{G}^{\Gamma^2})$, for which the 2-gauge transformation action induces a module structure
\begin{equation*}
    \Lambda: \mathbb{G}^{\Gamma^1} \times \mathfrak{C}(\mathbb{G}^{\Gamma^2})\rightarrow \mathfrak{C}(\mathbb{G}^{\Gamma^2}).
\end{equation*}
This extends to an additive linear structure 
\begin{equation*}
    \Lambda_{\zeta\oplus\zeta'} = \Lambda_\zeta \oplus \Lambda_{\zeta'}
\end{equation*}
for which the 2-gauge parameters are the homogeneous elements (in the sense of \cite{SOZER2023109155}). We denote by the additive completion of $\mathbb{G}^{\Gamma^1}$ by "$\mathbb{U}\G^{\Gamma^1}$". It was then shown in \cite{Chen1:2025?} that $\Lambda$ determines $\mathfrak{C}(\mathbb{G}^{\Gamma^2})$ as a \textit{monoidal} $\mathbb{U}\G^{\Gamma^1}$-module.

\medskip

Now by following \cite{Alekseev:1994pa}, given the 2-Chern-Simons Hamiltonian obtained from $S_{2CS}$, one can extract a \textit{2-graded 2-$R$-matrix} \cite{Bai_2013,Chen:2023tjf} which quantizes into an invertible cobraiding  on $\mathfrak{C}(\mathbb{G}^{\Gamma^2})$, and makes the monoidal structure on $\mathfrak{C}(\mathbb{G}^{\Gamma^2})$ non-symmetric (denoted $\ostar$). For an idea of what this cobraiding is, see \textit{Remark \ref{quantumcobraiding}}). This gives rise to a deformation quantization $\mathfrak{C}(\mathbb{G}^{\Gamma^2})\rightsquigarrow \mathfrak{C}_q(\mathbb{G}^{\Gamma^2})$ induced from the Lie 2-bialgebra $(\G;r)$ underlying 2-Chern-Simons action \cite{chen:2022}.

\begin{rmk}\label{quantumcobraiding}
    Let us elaborate a bit more on the cobraiding. In \cite{Chen1:2025?}, we have defined a (measureable) Hopf cocategory $\mathfrak{C}_q(\mathbb{G})$ whose objects are given by sheaves $\phi=\Gamma_\mathbb{G}[[\hbar]]$ of smooth $C(\mathbb{G})\otimes\bbC[[\hbar]]$-module algebras over $\mathbb{G}$, where $C_q(\mathbb{G}) = C(\mathbb{G})\otimes \bbC[[\hbar]]$ denotes the quantized function algebra on $\mathbb{G}$ (see \cite{Chen:2023tjf}). The components of the cobraiding $R$ at $\phi$ is, by definition, a measureable morphism
\begin{equation*}
    R_\phi: \Delta_\phi=\bigoplus \phi_{(1)}\boxtimes \phi_{(2)} \rightarrow \Delta_\phi^{\text{op}}=\bigoplus \phi_{(2)}\boxtimes \phi_{(1)}
\end{equation*}
given by a sheaf of bounded linear operators 
\begin{equation*}
    (R_\phi)_{\mathrm{z},\mathrm{z}'}: \bigoplus (\phi_{(1)})_\mathrm{z}\boxtimes(\phi_{(2)})_{\mathrm{z}'} \rightarrow \bigoplus (\phi_{(2)})_{\mathrm{z}}\boxtimes (\phi_{(1)})_{\mathrm{z}'}
\end{equation*}
at each stalk $(\mathrm{z},\mathrm{z}')\in\mathbb{G}^{\times 2}$.\footnote{Here the structure sheaf $C_q(\mathbb{G}^{\times 2}) \cong C_q(\mathbb{G})\bar \otimes C_q(\mathbb{G})$ on $\mathbb{G}^{\times 2}$ is defined in terms of the topological tensor product $\bar \otimes$, which is given by norm-completing the usual tensor product of the measureable $L^2$-sections.}  Given the invertibility of $R$, we can decompose the *-automorphism $$R_\phi =  R_\phi'\circ \operatorname{ad}_{R(\phi)},\qquad R(\phi) \in \bigoplus\phi_{(1)}\boxtimes\phi_{(2)}$$ into an "inner" and "outer" part. The sheaf $R\in\mathfrak{C}_q(\mathbb{G})\boxtimes \mathfrak{C}_q(\mathbb{G})$ formed by all of the sections $R(\phi)$ then constitute precisely the so-called "\textit{2-$R$-matrix}" mentioned earlier in the beginning of \S \ref{lattice2gau}.
\end{rmk}

In light of this deformation, for $\mathfrak{C}_q(\mathbb{G}^{\Gamma^2})$ to remain as a \textit{monoidal} module under 2-gauge transformations $\Lambda$, the category $\mathbb{U}\G^{\Gamma^1}\rightsquigarrow \mathbb{U}_q\G^{\Gamma^1}$ must itself receive a non-trivial quantum deformation. This leads to the introduction of an invertible cobraiding $\tilde R: \tilde\Delta\Rightarrow\tilde\Delta^\text{op}$ which satisfies the following conditions against $\Lambda$:
\begin{enumerate}
    \item  the presence of a coherent invertible \textbf{module tensorator} 
            $$\varpi: (-\ostar-)\circ \Lambda_{\tilde\Delta} \cong \Lambda\circ  (-\ostar-): \tilde{\mathbb{U}}_q\G^{\Gamma^1}\boxtimes\mathfrak{C}_q(\mathbb{G}^{\Gamma^2})^{\boxtimes 2}\rightarrow \mathfrak{C}_q(\mathbb{G}^{\Gamma^2}),$$
            and that
            \item the cobraidings satisfy
        \begin{align*}
      \Lambda_{\tilde\Delta }R \big( \Delta \big)\cong \Lambda_{\tilde R\big(\tilde \Delta\big)}\Delta:  \tilde{\mathbb{U}}_q\G ^{\Gamma^1}\boxtimes\mathfrak{C}_q(\mathbb{G}^{\Gamma^2}) \rightarrow \mathfrak{C}_q(\mathbb{G}^{\Gamma^2})^{\boxtimes 2}.
\end{align*}
\end{enumerate}
Together with the \textit{2-dagger structure} on the 2-skeleton $\Gamma^2$ induced by orientation reversal and framing reversal (see {\bf Example 5.5} of \cite{ferrer2024daggerncategories}), the following was then proven in \cite{Chen1:2025?}.
\begin{proposition}\label{2gthopfsymm}
    The categorical quantum gauge symmetries $\mathbb{U}_q\G^{\Gamma^1}$ of 2-Chern-Simons theory give rise to a strictly cobraided Hopf dagger category, graded by the groupoid $\mathbb{G}^{\Gamma^1}$ (cf. \cite{SOZER2023109155}) --- where the antipode $\tilde S$ and the dagger structure $-^\dagger = \tilde S_v$ are induced by the 2-dagger structure on the lattice $\Gamma$.
\end{proposition}

\begin{rmk}\label{catquantevelop}
    We note here that the above Hopf categorical structure of $\tilde{C}=\tilde{C}^\Gamma$ requires one to specify an underlying lattice $\Gamma$, but it does \textit{not} depend on which lattice it is. Particularly, if $\Gamma^1 = \{v\xrightarrow{e}v\}$ consist of a single edge loop based at a vertex $v\in\Gamma^0$, the homogeneous elements of the Hopf category $\tilde{C}$ corresponds to a single copy of $\mathbb{G}$. This led the author to call \cite{Chen1:2025?}, in this case, the \textit{categorical quantum enveloping algebra} $\mathbb{U}_q\G$, where $\G=\operatorname{Lie}\mathbb{G}$ denotes the Lie 2-algebra underlying the Lie 2-group \cite{Chen:2012gz,Bai_2013,chen:2022}. This is only a suggestive notation for now, but a future work will substantiate this notation by studying its categorical quantum duality with the \textit{categorical quantum coordinate ring}. This latter notion was defined in a much more concrete manner in \cite{Chen1:2025?}.
\end{rmk}

\subsection{Hopf structure on the quantum 2-gauge symmetries}\label{2gthopf}
In light of \textit{Remark \ref{catquantevelop}}, a special case of \textbf{Proposition \ref{2gthopfsymm}} is then the following.
\begin{theorem}
    Let $\Gamma^1= \big\{v\xrightarrow{e}v\big\}$ denote a 1-graph with a single loop edge $e$. The corresponding 2-gauge transformations on $\Gamma^1$ gives rise to a $\mathbb{G}$-graded strictly-cobraided Hopf dagger category $\Big(\mathbb{U}_q\G^{\Gamma^1}=\mathbb{U}_q\G,\cdot,\tilde I,\Delta,\tilde\epsilon,\tilde S;\tilde R\Big)$.
\end{theorem}
We call the Hopf category $\tilde{C}=\mathbb{U}_q\G$ the {\bf categorical quantum symmetries,} and serves as the central motivation for \textbf{Definition \ref{cobrhopfdaggercat}}.

\paragraph{The classical limit.}
Recall $q$ denotes the formal deformation parameter which deforms the coprodut functor $\Delta$ on $\tilde{C}$. As such, it is worth mentioning that in the classical limit $q\rightarrow 1$ we have 
\begin{enumerate}
    \item $\tilde{C}$ becomes cocommutative,
    \item $\tilde R \rightarrow \id\otimes\id$ becomes trivial, and
    \item $\tilde S$ becomes unipotent. 
\end{enumerate}
These facts will become important later in \S \ref{classicalpivotal}; they are analogues of the properties of ordinary quantum groups \cite{Woronowicz1988,Majid:1996kd}.


\medskip

We emphasize here that much of what follows should hold with $\tilde{C}$ replaced by a generic (ie. weak) Hopf dagger category equipped with a weak/lax cobraiding. However, we will prove several characterization results for the 2-representations specifically for the case where $\tilde{C}$ describes quantum 2-gauge transformations.

\section{Unitary 2-representations of $\tilde C$}\label{2rep}
We say a linear finite semisimple category $\cD$ (ie. a Kapranov-Voevodsky 2-vector space \cite{Kapranov:1994}) is a finite-dimensional {\it 2-representation} of $\tilde{C}  $ iff it is equipped with a lax monoidal functor $\rho: \tilde{C} \rightarrow \operatorname{End}(\cD)$, or equivalently a $\tilde{C}$-module structure $\rhd: \tilde{C}\times\cD\rightarrow\cD$ such that
\begin{equation*}
    \rho(\zeta)(d) = \zeta\rhd d,\qquad \forall~ \zeta\in \tilde{C},~ d\in\cD.  
\end{equation*}
We will often use both descriptions interchangeably. Note we do not a priori require $\cD$ to be representable as sheaves over some $\mathbb{G}^{\Gamma^1}$-space $P$ \cite{Nikolaus2011FOUREV,Schommer_Pries_2011}; this notion will become important elsewhere, but not here.

These 2-representations form a 2-category denoted by $\operatorname{2Rep}(\tilde{C})$, in which the 1-morphisms are module functors $F:\cD\rightarrow \cD'$ equipped with intertwining natural transformations $F_\zeta: F(\zeta\rhd-)\rightarrow \zeta\rhd F(-)$ for each $\zeta\in \tilde{C}$, and the 2-morphisms are module natural transformations $\alpha: F\Rightarrow F'$ which commute with $F_\zeta,F'_\zeta$. Such 2-representation 2-categories and their applications have been studied extensively for \textit{finite} 2-groups $\mathbb{G}$ in, eg., \cite{Bartsch:2022mpm,Bartsch:2023wvv,Huang:2024,Delcamp:2023kew,Chen2z:2023}, for which $\operatorname{2Rep}(\mathbb{G})$ is known to be finite semisimple (in fact fusion; see \cite{Douglas:2018}). In contrast, however, we emphasize here that it is so far unknown whether $\operatorname{2Rep}(\tilde{C};\tilde R)$ itself is finite semisimple as a 2-category --- it just \textit{contains} finite semisimple objects.

Throughout the following, we will use the "geometric/left-to-right convention for products $\boxtimes$ of objects and the "functorial"/right-to-left convention for composition $\circ,\bullet$ of 1-, 2-morphisms (see \cite{Douglas:2018} for a discussion on the distinction). In accordance with \textit{Remark \ref{catquantevelop}}, we will without loss of generality consider $\Gamma^1=\{v\xrightarrow{e}v\}$ consisting of a single loop, and $\tilde{C}=\mathbb{U}_q\G$. The main results in \cite{neuchl1997representation,Chen:2023tjf} then give us the following.
\begin{theorem}
    The 2-category $\operatorname{2Rep}(\tilde C;\tilde R)$ of 2-representations of the cobraided Hopf category $\tilde C=\mathbb{U}_q\G$, equipped with a cobraiding natural transformation $\tilde R: \tilde\Delta\Rightarrow\sigma\tilde\Delta $, is braided monoidal.
\end{theorem}
\noindent We will give a brief review  in \S \ref{prelims} of how the coproduct/cobraiding on $\tilde C$ introduce respectively the monoidal/braiding structures on $\operatorname{2Rep}(\tilde C;\tilde R)$. 

\medskip

However, here we can do better, because we have access to the antipode $\tilde S:\tilde C \rightarrow \tilde C^{\text{m-op,c-op}}$ and the dagger structure on $\tilde C$. The goal in this paper is to show that, over the $\bbC$-linear category $\mathsf{Hilb}$ of Hilbert spaces (namely we work with $\mathsf{2Hilb}$ \cite{Baez1996HigherDimensionalAI,Ganter:2006} instead of $\mathsf{2Vect}$), these give rise to the notions of compatible {\it duals and adjoints} in $\operatorname{2Rep}(\tilde C;\tilde R)$ (cf. \cite{Douglas:2018}). 


\subsection{Some monoidal and braided preliminaries}\label{prelims}
Let us first describe briefly the monoidal and braided structures of $\operatorname{2Rep}(\tilde C;\tilde R)$ as following from the Hopf structure of $\tilde C$, as we shall use them explicitly later. Details of these descriptions can be found in \cite{neuchl1997representation,Chen:2023tjf}. We shall work in the linear context, in which all 2-representations $\cD\in\operatorname{2Rep}(\tilde C)$ are $\mathsf{Hilb}$-modules. 

The key observation throughout this section is the fact that the action functor $\rhd$ sends objects in $\tilde C$ to endofunctors and morphisms to endonatural transformations. Thus we can record natural transformations by the edge parameters $\cE$ by the following
\[\zeta\rhd (\cD\xrightarrow{F}\cD') =(a_v\rhd\cD)\xrightarrow{\gamma_e\rhd F}(a_v'\rhd\cD')= \begin{tikzcd}
	\cD & \cD & {\cD'} & {\cD'}
	\arrow[""{name=0, anchor=center, inner sep=0}, "{a_v}", curve={height=-12pt}, from=1-1, to=1-2]
	\arrow[""{name=1, anchor=center, inner sep=0}, "{a_v'}"', curve={height=12pt}, from=1-1, to=1-2]
	\arrow["F", from=1-2, to=1-3]
	\arrow[""{name=2, anchor=center, inner sep=0}, "{a_{v}}"', curve={height=12pt}, from=1-3, to=1-4]
	\arrow[""{name=3, anchor=center, inner sep=0}, "{a_v'}", curve={height=-12pt}, from=1-3, to=1-4]
	\arrow["{\gamma_e}", shorten <=3pt, shorten >=3pt, Rightarrow, from=0, to=1]
	\arrow["{\gamma_e^{-1}}"', shorten <=3pt, shorten >=3pt, Rightarrow, from=3, to=2]
\end{tikzcd}\] 
for each functor $F\in\operatorname{Hom}(\cD,\cD')$, where $\zeta = a_v\xrightarrow{\gamma_e} a_{v
'}\in\tilde C$ is written in terms of its source and target.

In the following, we keep track of the $\tilde C$-module coherence conditions satisfied by the natural transformation $F_{a_v}: F(a_v\rhd -) \Rightarrow a_v\rhd' F(-)$ \cite{Delcamp:2023kew} by the following commuting diagram\footnote{This notation is suggestive. Indeed, for an endofunctor $F:\cI\rightarrow \cI$ on the tensor unit $\cI$ (defined later), $F_{a_v}:a_v\rhd-\circ F\circ a_v^{-1}\rhd- \cong F$ defines an endo-natural transformation on $F$ and hence determines an action of $\cV$ on $\operatorname{End}(\cI)$. See also \cite{Chen2z:2023,Bartsch:2023wvv}.}
\begin{equation}
    \gamma_e\rhd F = 
\begin{tikzcd}
	{\rho'(a_v)^{-1}\circ F\circ\rho(a_v)} & F \\
	{\rho'(a_v')^{-1}\circ F\circ\rho(a_v')}
	\arrow["{F_{a_v}}", Rightarrow, from=1-1, to=1-2]
	\arrow["{\rho(\gamma_e)^{-1}\circ F\circ \rho(\gamma_e)}"', Rightarrow, from=1-1, to=2-1]
	\arrow["{F_{a_v'}}"', Rightarrow, from=2-1, to=1-2]
\end{tikzcd}.\label{Vstruct}
\end{equation}
We call this the \textbf{$\cE$-structures} of the $\tilde C$-module functors; here "$\cE$" refers to the \textbf{e}dge decorations/1-morphisms in $\tilde C=\mathbb{U}_q\G$. 


\subsubsection{Tensor products}
Recall from \S \ref{2gthopf} the coproduct functor $\tilde\Delta$ on $\tilde C  $. Putting $$\tilde\Delta_0 = (\iota\otimes\iota)\tilde\Delta\mid_{\cV},\qquad \tilde\Delta_1^l = (\iota\otimes\pi)\tilde\Delta\mid_{\cE},\qquad \tilde\Delta^r_1 = (\pi\otimes\iota)\tilde\Delta\mid_\cE,$$ we write for any $\tilde C$-module category $\cD$ and functor $F:\cA\rightarrow \cA'$
\begin{align*}
    &\tilde\Delta_{\zeta}\rhd (\cA\boxtimes\cD\xrightarrow{F\boxtimes\cD}\cA'\boxtimes\cD) = (\tilde\Delta_0)_{a_v}\rhd ({\cA\boxtimes\cD}) \xrightarrow{(\tilde\Delta_1^l)_{\gamma_e}\rhd(F\boxtimes\cD)}(\tilde\Delta_0)_{a_v'}\rhd({\cA'\boxtimes\cD}), \\
    & \tilde\Delta_{\zeta}\rhd (\cD\boxtimes\cA\xrightarrow{\cD\boxtimes F}\cD\boxtimes\cA') = (\tilde\Delta_0)_{a_v} \rhd ({\cD\boxtimes\cA}) \xrightarrow{(\tilde\Delta_1^r)_{\gamma_e} \rhd (\cD\boxtimes F)}(\tilde\Delta_0)_{a_v'} \rhd({\cD\boxtimes\cA'})
\end{align*}
where $\zeta = a_v\xrightarrow{\gamma_e}a_v'\in\tilde C$. 

The naturality of this definition, as well as the compatibility against the $\tilde C$-module associator
\begin{equation*}
    \zeta\rhd (\zeta'\rhd -) \xrightarrow{\sim} (\zeta\cdot\zeta')\rhd -,
\end{equation*}
follow respectively from the fact that $\tilde \Delta$ defines a functor and the bimonoidal axioms for $\tilde C$. These facts were proven in Lemmas 6.11, 6.12 in \cite{Chen:2023tjf}; see also \cite{neuchl1997representation}. There is then also a natural {\it interchanger} of functors $F:\cA\rightarrow \cA'$ and $G:\cB\rightarrow\cB'$,
\begin{equation*}
    \upsilon_{G,F}: (\cA' \boxtimes G) \circ (F\boxtimes \cB)\Rightarrow (F \boxtimes \cB') \circ (\cA\boxtimes G),
\end{equation*}
which can always be chosen to be invertible/equivalences.

\begin{rmk}\label{nudge}
    Take any pair of functors $F: \cD\rightarrow \cD'$ and $G:\cA\rightarrow \cA'$. The invertibility of the interchanger $\upsilon_{F,G}$ implies that the tensor product $F\boxtimes G: \cD\boxtimes \cA\rightarrow \cD'\boxtimes\cA'$ is well-defined 
\begin{equation*}
    F\boxtimes G = ({\cD'} \boxtimes G)\circ (F\boxtimes \cA) \cong (F\boxtimes {\cA'})\circ({\cD}\boxtimes G)
\end{equation*}
up to 2-isomorphism; this is called \textbf{nudging} of functors in \cite{Douglas:2018}. In the following, we will assume that all functors between monoidal products of $\tilde C$-module categories can be written in this way, ie. using nudging. This is a 2-categorical version of the condition Definition 1.7 (b) in \cite{Deligne2007}.
\end{rmk}

Denote by the $\tilde C$-module associators on $\operatorname{2Rep}(\tilde C;\tilde R)$ by $\alpha^{\tilde C}$. The strict coassociativity of $\tilde\Delta$ gives rise to an invertible natural transformations fitting into commutative squares of the form
\[\begin{tikzcd}
	{((\tilde\Delta\otimes1)\tilde\Delta)_\zeta \rhd (\cD_1\boxtimes\cD_2)\boxtimes\cD_3} & {((\tilde\Delta\otimes1)\tilde\Delta)_\zeta \rhd (\cD_1\boxtimes\cD_2)\boxtimes\cD_3} \\
	{((1\otimes\tilde\Delta)\tilde\Delta)_\zeta \rhd\cD_1\boxtimes(\cD_2\boxtimes\cD_3)} & {((1\otimes\tilde\Delta)\tilde\Delta)_\zeta \rhd\cD_1'\boxtimes(\cD_2\boxtimes\cD_3)}
	\arrow[from=1-1, to=1-2]
	\arrow["\cong"', from=1-1, to=2-1]
	\arrow["\cong"', from=1-2, to=2-2]
	\arrow["\cong", shorten <=14pt, shorten >=14pt, Rightarrow, from=2-1, to=1-2]
	\arrow[from=2-1, to=2-2]
\end{tikzcd},\]
for each $\zeta\in\tilde C$ and functors $F_1:\cD_1\rightarrow\cD_1'$, where the horizontal maps are given by $((\tilde\Delta\otimes1)\tilde\Delta)_\zeta \rhd (F\boxtimes\cD_2)\boxtimes\cD_3$ and $((1\otimes \tilde\Delta)\tilde\Delta)_\zeta \rhd F\boxtimes(\cD_2\boxtimes\cD_3)$. Similar constructions can be made for diagrams arising from insertions of functors $F_i$ at positions $i=2,3$. We also have $\tilde C$-module natural transformations witnessing the following 2-cell
\begin{equation}\alpha_{F23}^{\tilde{\cD}}=\begin{tikzcd}
	{(\cD_1\boxtimes\cD_2)\boxtimes\cD_3} & {\cD_1\boxtimes(\cD_2\boxtimes\cD_3)} \\
	{(\cD_1'\boxtimes\cD_2)\boxtimes\cD_3} & {\cD_1'\boxtimes(\cD_2\boxtimes\cD_3)}
	\arrow["{\alpha_{123}^{\tilde C}}", from=1-1, to=1-2]
	\arrow["{(F\boxtimes \cD_2)\boxtimes\cD_3}"', from=1-1, to=2-1]
	\arrow["{F\boxtimes(\cD_2\boxtimes\cD_3)}", from=1-2, to=2-2]
	\arrow[shorten <=9pt, shorten >=9pt, Rightarrow, from=2-1, to=1-2]
	\arrow["{\alpha_{1'23}^{\tilde C}}", from=2-1, to=2-2]
\end{tikzcd};\label{moduassoc}\end{equation}
see Lemma 6.15 of \cite{Chen:2023tjf}, and also \cite{neuchl1997representation,gurski2006algebraic}. Since $\tilde C$ is strict, these associators and pentagonators \cite{Delcamp:2023kew,KongTianZhou:2020} are always invertible and have identity components. We shall therefore suppress them in the following. 


\begin{rmk}\label{descendant}
    When $\mathbb{G}$ is a weakly-associative smooth 2-group \cite{Schommer_Pries_2011}, its associator $\tau:G^{\times 3}\rightarrow\mathsf{H}$ (ie. the representative of its \textit{Postnikov class}) directly contributes to a {\it non-invertible} $\tilde C$-module associator $\alpha^{\tilde C}$ through the vertical maps in \eqref{moduassoc}. The monoidal witness $\zeta\rhd(\zeta'\rhd-)\Rightarrow(\zeta\cdot\zeta')\rhd -$ must satisfy a module pentagon equation against this associator. These non-invertible 1-morphisms must therefore be kept track of when $\mathbb{G}$ is weakly-associative.
\end{rmk}

The counit functor $\tilde\epsilon:\tilde C\rightarrow\mathsf{Hilb}$ identifies a distinguished object $\cI\in\operatorname{2Rep}(\tilde C)$ as the trivial 2-representation $a_v\rhd \cI = \epsilon(a_v)\otimes \cI\cong \cI$ in terms of the $\mathsf{Hilb}$-module structure of $\cI$. Furthermore, $\epsilon$ also selects a counit $({\bf 1}_{a_v})_e$ over each object $a_v\in\tilde C$ (ie. the identity arrow), such that the identity endofunctor $1_\cD\in\operatorname{End}(\cD)$ transforms as
\begin{equation*}
    \zeta\rhd (\cD\xrightarrow{1_\cD}\cD) = (a_v\rhd\cD)\xrightarrow{\epsilon_{a_v}\cdot 1_\cD} (a_v\rhd\cD)= (a_v\rhd\cD)\xrightarrow{1_\cD}(a_v\rhd\cD).
\end{equation*}
Concretely, $\epsilon_{\gamma_e}$ is represented as an invertible linear map, and it "acts" on $1_\cD$ as an element of $\mathsf{Hilb}$. The counitality axiom $(\tilde\epsilon\otimes1)\circ\tilde\Delta = (1\otimes \tilde\epsilon)\circ\tilde\Delta = \id$ gives rise to the following invertible $\tilde C$-module unitors 
\[\begin{tikzcd}
	{\cD\boxtimes \cI} & \cD \\
	{\cD'\boxtimes \cI} & {\cD'}
	\arrow["{r_\cD}", from=1-1, to=1-2]
	\arrow["{F\boxtimes \cI}"', from=1-1, to=2-1]
	\arrow["F", from=1-2, to=2-2]
	\arrow["{r_F}"', shorten <=5pt, shorten >=5pt, Rightarrow, from=2-1, to=1-2]
	\arrow["{r_{\cD'}}"', from=2-1, to=2-2]
\end{tikzcd},\qquad 
\begin{tikzcd}
	{\cI\boxtimes \cD} & \cD \\
	{\cI\boxtimes \cD'} & {\cD'}
	\arrow["{\ell_\cD}", from=1-1, to=1-2]
	\arrow["{\cI\boxtimes F}"', from=1-1, to=2-1]
	\arrow["F", from=1-2, to=2-2]
	\arrow["{\ell_F}"', shorten <=5pt, shorten >=5pt, Rightarrow, from=2-1, to=1-2]
	\arrow["{\ell_{\cD'}}"', from=2-1, to=2-2]
\end{tikzcd}\]
such that the usual triangle axioms follow from the counit axioms and coassociativity,
\begin{equation*}
    (\tilde\epsilon\otimes 1\otimes 1)(\tilde\Delta\otimes 1)\tilde\Delta = \tilde\Delta = (\tilde\epsilon\otimes 1\otimes 1)(1\otimes \tilde\Delta)\tilde\Delta,\qquad \text{etc.}
\end{equation*}

\subsubsection{Braiding}\label{braidingstructure}
We now briefly introduce the braiding structure. Recall in \S \ref{lattice2gau} that we are presently working under the assumption that the cobriading $\tilde R$ is constructed from the quantum 2-$R$-matrix on $\tilde C$, which we as an abuse of notation also denote by $\tilde R \in \tilde C\boxtimes\tilde C$. The braiding map then takes the form
$$c = \text{flip}\circ (\rho \boxtimes\rho')\tilde R,\qquad \rho,\rho'\in\operatorname{2Rep}(\tilde C,\tilde R),$$ where the flip map $\cD\boxtimes\cD'\mapsto\cD'\boxtimes\cD$ swaps the Delign tensor product factors. 
More explicitly, writing $$\tilde R_0=\tilde R \mid_{\cV\otimes\cV},\qquad \tilde R_1^r =\tilde R \mid_{\cE\otimes\cV},\qquad \tilde R_1^l=\tilde R \mid_{\cV\otimes\cE}, $$ we put
\begin{gather*}
    c_{\cD,\cA} = \text{flip}(\tilde R_0(\rhd -\boxtimes \rhd-)),\\
    c_{F,\cD} = \text{flip}(\tilde R_1^l(\rhd -\boxtimes \rhd-)),\qquad c_{\cD, F} = \text{flip}(\tilde R_1^r(\rhd -\boxtimes \rhd-)),
\end{gather*}
where we recall the vertex transforms act by natural transformations on functors. Hence given functors $F\in\operatorname{Hom}(\cD,\cD')$ and $F'\in\operatorname{Hom}(\cA,\cA')$, we can then write
\begin{align*}
    &  c(\cD \boxtimes \cA \xrightarrow{F\boxtimes \cA} \cD'\boxtimes \cA) = c_{\cD,\cA}(\cD\boxtimes\cA) \xrightarrow{c_{F,\cA}(F\boxtimes \cA)}c_{\cD',\cA}(\cD'\boxtimes\cA),\\ 
     & c(\cD \boxtimes \cA \xrightarrow{\cD\boxtimes F'} \cD\boxtimes \cA') = c_{\cD,\cA}(\cD\boxtimes\cA) \xrightarrow{c_{\cD, F'}(\cD\boxtimes F')} c_{\cD,\cA'}(\cD\boxtimes\cA').
\end{align*}
The fact that these define $\tilde C$-module functors/natural transformations is a result of the  quasitriangularity condition \eqref{2yb}. The naturality of $\tilde R$ as a cobraiding transformation on $\tilde C$ implies that these braiding structures fit into the following squares,
\[\begin{tikzcd}
	{\cD\boxtimes\cA} & {\cD'\boxtimes\cA} \\
	{\cA\boxtimes\cD} & {\cA\boxtimes\cD'}
	\arrow["{F\boxtimes\cA}", from=1-1, to=1-2]
	\arrow["{c_{\cD,\cA}}"', from=1-1, to=2-1]
	\arrow["{c_{\cD',\cA}}", from=1-2, to=2-2]
	\arrow["{c_{F,\cA}}"', shorten <=5pt, shorten >=5pt, Rightarrow, from=2-1, to=1-2]
	\arrow["{\cA\boxtimes F}"', from=2-1, to=2-2]
\end{tikzcd},\qquad 
\begin{tikzcd}
	{\cA\boxtimes\cD} & {\cA'\boxtimes\cD} \\
	{\cD\boxtimes\cA} & {\cD\boxtimes\cA'}
	\arrow["{F'\boxtimes \cD}", from=1-1, to=1-2]
	\arrow["{c_{\cA,\cD}}"', from=1-1, to=2-1]
	\arrow["{c_{\cA',\cD}}", from=1-2, to=2-2]
	\arrow["{c_{F',\cD}}"', shorten <=5pt, shorten >=5pt, Rightarrow, from=2-1, to=1-2]
	\arrow["{\cD\boxtimes F'}"', from=2-1, to=2-2]
\end{tikzcd}\]
where $F:\cD\rightarrow\cD'$ and $F':\cA\rightarrow\cA'$. See Lemmas 7.4, 7.3 in \cite{Chen:2023tjf} or \cite{neuchl1997representation}; the graphical representation of these 2-morphisms can be found in fig. 55 (d) in \cite{Barrett_2024}.




Due to the strictness of the cobraiding, the leftover part $\tilde R'$ of the cobraiding can be seen to contribute directly to the invertible hexagonator $\Omega$, which witnesses the hexagon relation/third Reidemeister move \cite{BAEZ2003705,KongTianZhou:2020,Chen:2023tjf,Chen2z:2023}. The naturality of the braiding gives a {\it braid-exchange} 2-morphism
\begin{equation}
    \begin{tikzcd}
	{(\cD_1\boxtimes\cD_2)\boxtimes\cD_3} && {\cD_3\boxtimes(\cD_1\boxtimes\cD_2)} \\
	\\
	{(\cD_2\boxtimes\cD_1)\boxtimes\cD_3} && {\cD_3\boxtimes(\cD_2\boxtimes\cD_1)}
	\arrow["{c_{\cD_1\boxtimes\cD_2,\cD_3}}", from=1-1, to=1-3]
	\arrow[""{name=0, anchor=center, inner sep=0}, "{c_{\cD_1,\cD_2}\boxtimes\cD_3}"', from=1-1, to=3-1]
	\arrow[""{name=1, anchor=center, inner sep=0}, "{\cD_3\boxtimes c_{\cD_1,\cD_2}}", from=1-3, to=3-3]
	\arrow["{c_{\cD_2\boxtimes\cD_1,\cD_3}}"', from=3-1, to=3-3]
	\arrow["{c_{c_{\cD_1,\cD_2},\cD_3}}"', shorten <=25pt, shorten >=25pt, Rightarrow, from=0, to=1]
\end{tikzcd}\label{braidexch}
\end{equation}
for each object $\cD_1,\cD_2,\cD_3\in\operatorname{2Rep}(\tilde C;\tilde D)$, which relates the two hexagonators through the following invertible 2-morphism (see \cite{KongTianZhou:2020,Chen:2023tjf}, also fig. 55 (c) in \cite{Barrett_2024})
\begin{equation*}
    \Omega_{c_{\cD_1,\cD_2}|\cD_3} = \Omega_{\cD_1|\cD_3\cD_2}^{-1} \bullet c_{c_{\cD_1,\cD_2},\cD_3}\bullet\Omega_{\cD_1|\cD_2\cD_3}.
\end{equation*}
Moreover, \eqref{2yb} and the strict coassociativity of $\tilde C$ also allows us to deduce the compatibility of the hexagonators against tensor products,
\begin{align}
    & \Omega_{(\cD_3\boxtimes\cD_4)|\cD_1\cD_2}\bullet (\cD_3\boxtimes\Omega_{\cD_4|\cD_1\cD_2}\circ\Omega_{\cD_3|\cD_1\cD_2}\boxtimes\cD_4):\nonumber \\
    &\qquad\qquad (c_{\cD_2,\cD_4}\circ c_{\cD_1,\cD_4})\circ (c_{\cD_1,\cD_3}\circ c_{\cD_2,\cD_3})\Rightarrow c_{\cD_1\boxtimes\cD_2,\cD_3\boxtimes\cD_4},\label{quadbraid}
\end{align}
for any quadruple of objects $\cD_1,\dots,\cD_4$. In the context of a Hopf 2-algebra, the four braided monoidal coherence axioms \cite{GURSKI20114225,KongTianZhou:2020} were explicitly checked to hold in Theorem 7.11 of \cite{Chen:2023tjf}, hence we will not reproduce them here. 

The following is true in any braided monoidal 2-category \cite{BATANIN2008334}.
\begin{proposition}\label{symmetric}
    The endomorphism category $\operatorname{End}(\cI)$ of the unit 2-representation $\cI\in\operatorname{2Rep}(\tilde C;\tilde R)$ is symmetric.
\end{proposition}

\subsection{Adjunction of 2-representations}\label{adjoints}
Recall the notion of 2-Hilbert spaces $\mathsf{2Hilb}$ in \cite{Baez1996HigherDimensionalAI}; see also \cite{Chen2024ManifestlyUH} for a more recent and accurate account. Let us start light by studying the adjoints. We shall inherit the left-/right-adjoints for the hom-categories in $\mathsf{2Hilb}$ from the left-/right-dualities in $\mathsf{Hilb}$ --- ie. that of taking the dual or the predual Hilbert spaces.\footnote{For $V\in \mathsf{Hilb}$, the left- and right-duals coincide and we in fact have non-canonical isomorphisms $V\cong V^*\cong \,^*V$ due to Riesz representation theorem. Hence we technically do not need to distinguish between left- and right-adjoints here, but we do it anyway for bookkeeping.} 

The reason for this is the following: since $\mathsf{Hilb}$ itself is bi-involutive \cite{Jones:2017}, making $\operatorname{2Rep}(\tilde C;\tilde R)$ 2-$\mathsf{Hilb}$-enriched will automatically make it into a {\it dagger 2-category} \cite{stehouwer2023dagger,ferrer2024daggerncategories}. In fact, this was the original motivation for higher-dagger structures. 

Leveraging this observation, we shall describe adjunctions on the  $\tilde C$-module functors, such that they come equipped with the proper $\cE$-structures, in the following. This is accomplished by {\it orientation reversal}.

\paragraph{Orientation reversal.}
Observe that, under an orientation reversal $e\mapsto \bar e$, the appropriate swaps are achieved
\begin{equation*}
     (\bar \zeta)_{(v,e)} = a_v'\xrightarrow{\gamma_{\bar e}}a_v
\end{equation*}
for each $\zeta=a_v\xrightarrow{\gamma_e}a_v'\in\tilde C$. The induced involution $\tilde C\rightarrow\tilde C^\text{op}$ is precisely the dagger structure $\tilde S_v=-^\dagger$.

Now suppose, for each $\cD$, its $\tilde C$-action functor $\rho=\rhd$ satisfy the following \textbf{unitarity property}
\begin{equation}
    (a_v\rhd-)^\dagger = \,^\dagger(a_v\rhd-) = a_v \rhd -,\qquad \forall~a_v\in\cV,\label{unitary}
\end{equation}
then we define the following $\cE$-structures of the adjunctions
\begin{align*}
    & \zeta\rhd(\cD'\xrightarrow{F^\dagger}\cD) = (a_v'\rhd)^\dagger\cD'\xrightarrow{(\gamma_{\bar e}\rhd)^\dagger (F^\dagger)}(a_v\rhd)^\dagger \cD,\\
    & \zeta\rhd(\cD'\xrightarrow{\,^\dagger F}\cD) = \,^\dagger(a_v'\rhd)\cD'\xrightarrow{\,^\dagger(\gamma_{\bar e}\rhd)(\,^\dagger F)}\,^\dagger(a_v\rhd) \cD.
\end{align*}
The following left-/right-adjunction-mates of $\tilde C$-module natural transformations\footnote{This means that $\alpha:F\Rightarrow G$ intertwines $\rho(\gamma_e)$, and commutes with the natural transformations $F_{a_v},G_{a_v}$.} $\alpha:F\Rightarrow G$,
\begin{equation}
    (F\xRightarrow{\alpha} G)^\dagger = (G^\dagger \xRightarrow{\alpha^\dagger}F^\dagger),\qquad\,^\dagger(F\xRightarrow{\alpha} G) = (\,^\dagger G \xRightarrow{\,^\dagger\alpha}\,^\dagger F),\label{adjmates}
\end{equation}
are themselves $\tilde C$-module natural transformations. This follows directly from the naturality of the duals in $\mathsf{Hilb}$. We shall mainly focus on the left-adjoint $-^\dagger$ in the following.

\medskip

From the geometry (or the definition of the dagger $-^\dagger=\tilde S_v$), orientation reversal swaps the sources and targets of the arrows in $\tilde C$. The composition of arrows in $\tilde C$ then then leads to the condition
\begin{equation}
    (F\circ G)^\dagger = G^\dagger\circ F^\dagger\label{adjcompose}
\end{equation}
satisfied by the adjunctions. Furthermore, given the unit representation $\tilde\epsilon: \tilde C\rightarrow\mathsf{Hilb}$ lands in {\it real} Hilbert spaces (ie. those which are self-dual under $-^\dagger$), this implies 
\begin{equation*}
    1_\cI^\dagger = 1_{\cI}.
\end{equation*}
This can be understood as a certain {\it reality condition} on the unit $\cI$.

To be clear, the above is describing the proper $\cE$-structure -- as given by the orientation reversal/dagger involution structure --- for the adjoints of the $\tilde C$-module functors, such that the adjunction co/unit 2-morphisms (called "folds for the adjunctions in the following) have a canonical $\tilde C$-module structure.

\subsubsection{Folds for the adjoints}
Now take a 1-graph $v\xrightarrow{e}v'$ and its orientation reversal. Their composition bounds a contractible 2-cell which is null-homotopic,
\[\begin{tikzcd}
	v && {v'}
	\arrow[""{name=0, anchor=center, inner sep=0}, "e", curve={height=-12pt}, from=1-1, to=1-3]
	\arrow[""{name=1, anchor=center, inner sep=0}, "{\bar e}", curve={height=-12pt}, from=1-3, to=1-1]
	\arrow["\cong", shorten <=3pt, shorten >=3pt, Rightarrow, from=0, to=1]
\end{tikzcd}\]
On the other hand, the decorations on the 1-graphs by construction respect their groupoid compositions, the $\cE$-structure of $F^\dagger\circ F$ is given by $((\gamma_{\bar e}\rhd -)^\dagger \bullet (\gamma_{e'}\rhd -))(F)$. We are therefore led to the following notion.
\begin{definition}\label{planarunitary}
A $\tilde C$-module functor $F:\cD\rightarrow \cD'$ is said to have \textbf{planar-unitary structure} iff
\begin{enumerate}
    \item $\cD,\cD'$  have equipped $\cV$-action functors $\rho=\rhd,\rho'=\rhd'$ that satisfy \eqref{unitary}, and
    \item $F$ come equipped with the following 2-morphisms
\begin{equation*}
    e_F: F^\dagger\circ F \Rightarrow 1_\cD,\qquad \iota_F: 1_{\cD'} \Rightarrow F\circ \,^\dagger F,
\end{equation*}
called \textit{adjunction-folds}, such that they induce the following commutative diagrams,
\[\begin{tikzcd}
	{\rho’(a_v)^{-1}\circ (F^\dagger\circ F)\circ\rho(a_v)} &&& {\rho’(a_v)^{-1}\circ 1_\cD\circ\rho(a_v)} \\
	& {F^\dagger\circ F} & {1_\cD} \\
	{\rho’(a_v')^{-1}\circ (F^\dagger\circ F)\circ\rho(a_v')} &&& {\rho’(a_v')^{-1}\circ 1_\cD\circ\rho(a_v')}
	\arrow["{\rho’(\gamma_e)^{-1}\circ e_F\circ\rho(\gamma_{e})}", Rightarrow, from=1-1, to=1-4]
	\arrow["{(F^\dagger\circ F)_{a_v}}", Rightarrow, from=1-1, to=2-2]
	\arrow["{(\gamma_e^\dagger\bullet \gamma_e)F}"', Rightarrow, from=1-1, to=3-1]
	\arrow["{(1_\cD)_{a_v}}"', Rightarrow, from=1-4, to=2-3]
	\arrow["{(\gamma_e^\dagger\bullet \gamma_e)1_\cD}", Rightarrow, from=1-4, to=3-4]
	\arrow["{e_F}", Rightarrow, from=2-2, to=2-3]
	\arrow["{(F^\dagger\circ F)_{a_v'}}"', Rightarrow, from=3-1, to=2-2]
	\arrow["{\rho’(\gamma_{e’})^{-1}\circ e_F\circ\rho(\gamma_{e’})}"', Rightarrow, from=3-1, to=3-4]
	\arrow["{(1_\cD)_{a_v'}}", Rightarrow, from=3-4, to=2-3]
\end{tikzcd}\]
\[\begin{tikzcd}
	{\rho’(a_v)^{-1}\circ (F\circ \,^\dagger F)\circ\rho(a_v)} &&& {\rho’(a_v)^{-1}\circ 1_\cD\circ\rho(a_v)} \\
	& {F\circ \,^\dagger F} & {1_\cD} \\
	{\rho’(a_v')^{-1}\circ (F\circ \,^\dagger F)\circ\rho(a_v')} &&& {\rho’(a_v')^{-1}\circ 1_\cD\circ\rho(a_v')}
	\arrow["{(F\circ \,^\dagger F)_{a_v}}", Rightarrow, from=1-1, to=2-2]
	\arrow["{(\gamma_e\bullet\,^\dagger \gamma_e)F}"', Rightarrow, from=1-1, to=3-1]
	\arrow["{\rho’(\gamma_e)^{-1}\circ \iota_F\circ\rho(\gamma_{e})}"', Rightarrow, from=1-4, to=1-1]
	\arrow["{(1_\cD)_{a_v}}"', Rightarrow, from=1-4, to=2-3]
	\arrow["{(\gamma_e\bullet \,^\dagger \gamma_e)1_\cD}", Rightarrow, from=1-4, to=3-4]
	\arrow["{\iota_F}", Rightarrow, from=2-3, to=2-2]
	\arrow["{(F\circ \,^\dagger F)_{a_v'}}"', Rightarrow, from=3-1, to=2-2]
	\arrow["{(1_\cD)_{a_v'}}", Rightarrow, from=3-4, to=2-3]
	\arrow["{\rho’(\gamma_{e’})^{-1}\circ \iota_F\circ\rho(\gamma_{e’})}", Rightarrow, from=3-4, to=3-1]
\end{tikzcd}\]
Here we have used a shorthand $(\gamma_{\bar e}^\dagger \bullet \gamma_e)F = (\rho^\dagger(\gamma_{\bar e})^{-1}\circ F^\dagger\circ\rho^\dagger(\gamma_{\bar e}))\bullet( \rho(\gamma_e)^{-1}\circ F\circ \rho(\gamma_e))$, and similarly for $(\gamma_e\bullet \gamma_{\bar e}^\dagger)$.
\item A 2-morphism $\alpha:F\Rightarrow G$ between functors $F,G:\cD\rightarrow \cD'$ is called {\it planar-unitary} (or just unitary) if it has left-/right-inverses given by the left-/right-adjunction-mates
\begin{equation*}
    \alpha^\dagger \bullet\alpha = \id_{F},\qquad \alpha\bullet\,^\dagger\alpha = \id_{G}.
\end{equation*}
Note this follows directly from the naturalty of the adjunction-folds $e,\iota$ if $F,G$ are themselves planar-unitary.
    \end{enumerate}
The 2-category $\operatorname{2Rep}(\tilde C;\tilde R)$ is said to have \textbf{planar-unitarity} if all of its functors are planar-unitary.
\end{definition}
By the invertibility of 2-gauge transformations, $\rho(\gamma_e)^{-1}\rho(\gamma_e)={\bf 1}_e$ is equivalent to the trivial edge transform for all $\tilde C$-actions functors $\rho$. The above definition then implies that the adjunction-folds are intertwining
\begin{equation}
    e_F \bullet (\gamma_{\bar e}^\dagger\bullet \gamma_e) = ({\bf 1}_e\rhd -) \bullet e_F,\qquad \iota_F\bullet ({\bf 1}_e\rhd-) = (\gamma_e\bullet \gamma_{\bar e}^\dagger)\bullet \iota_F.\label{unitaryintw}
\end{equation}
Since the identity functor itself $1_\cD$ satisfies $1_\cD\circ 1_\cD=1_\cD$ and is attached the trivial $\cE$-structure $a_v\rhd - = \id$, we also have
\begin{equation*}
    e_{1_\cD} = \id_{1_\cD},\qquad \iota_{1_\cD} = \id_{1_\cD}.
\end{equation*}
In conjunction with \eqref{adjcompose}, the following null-homotopies
\[\begin{tikzcd}
	& {v'} \\
	v & {v'} & {v''}
	\arrow["\cong", shorten <=2pt, Rightarrow, from=1-2, to=2-2]
	\arrow["{e_2}", from=1-2, to=2-3]
	\arrow["{e_1}", from=2-1, to=1-2]
	\arrow["{\bar e_1}", from=2-2, to=2-1]
	\arrow["{\bar e_2}", from=2-3, to=2-2]
\end{tikzcd}=\begin{tikzcd}
	v & {v'} & {v''}
	\arrow[""{name=0, anchor=center, inner sep=0}, "{e_1}", curve={height=-12pt}, from=1-1, to=1-2]
	\arrow[""{name=1, anchor=center, inner sep=0}, "{\bar e_1}", curve={height=-12pt}, from=1-2, to=1-1]
	\arrow[""{name=2, anchor=center, inner sep=0}, "{e_2}", curve={height=-12pt}, from=1-2, to=1-3]
	\arrow[""{name=3, anchor=center, inner sep=0}, "{\bar e_2}", curve={height=-12pt}, from=1-3, to=1-2]
	\arrow["\cong", shorten <=3pt, shorten >=3pt, Rightarrow, from=0, to=1]
	\arrow["\cong", shorten <=3pt, shorten >=3pt, Rightarrow, from=2, to=3]
\end{tikzcd}\]
lead to the following compatibility
\begin{equation}
    e_{F\circ G} = e_G\bullet(G^\dagger \circ e_F\circ G),\qquad \iota_{F\circ G} = (F\circ \iota_G\circ F^\dagger)\bullet\iota_F \label{adjfoldcompose}
\end{equation}
for composable planar-unitary $\tilde C$-module functors $\cD\xrightarrow{G}\cD'\xrightarrow{F}\cD''$. In the following, all 2-representations of $\tilde C$ will be planar-unitary.

Similar statements as above of course hold for the right-adjoints of $\tilde C$-module functors. In particular, we also have the adjunction-folds
\begin{equation*}
        \overline{e}_G: G\circ\,^\dagger G\Rightarrow 1_{\cD'},\qquad \overline{\iota}_F: \,^\dagger F\circ F\Rightarrow 1_\cD.
\end{equation*}
However, since we know that adjunctions are left-/right-involutive $\,^\dagger F^\dagger\cong F$, we see that
\begin{equation*}
    \overline{e}_G = e_{\,^\dagger G},\qquad \overline{\iota}_F = \iota_{\,^\dagger F};
\end{equation*}
this is part of the conditions for "planar-pivotality" for a 2-category in \cite{Douglas:2018}.

\subsubsection{Snake equations for the adjunctions}\label{adjsnake}
We now turn to the left- and right-adjoint-mate 2-morphisms.
\begin{proposition}
    Let $\alpha:F\Rightarrow G$ be a 2-morphism in $\operatorname{2Rep}(\tilde C,\tilde R)$. Recall the left-/right-adjoint-mates in \eqref{adjmates}; we have
    \begin{align*}
        & \alpha^\dagger = (e_G\circ F^\dagger)\bullet(G^\dagger\circ \alpha\circ F^\dagger)\bullet(G^\dagger\circ \iota_F)\\
        & \,^\dagger\alpha = (\,^\dagger F\circ {e}_{\,^\dagger G})\bullet(\,^\dagger F\circ\alpha\circ \,^\dagger G)\bullet({\iota}_{\,^\dagger F}\circ \,^\dagger G).
    \end{align*}    
    Moreover, they coincide.
\end{proposition}
\begin{proof}
    By planar-unitarity, we only need to show that the right-adjunction-mate $\,^\dagger\alpha$ coincides with
    $$\,^\dagger\alpha' = (\,^\dagger F\circ \overline{e}_G)\bullet(\,^\dagger F\circ\alpha\circ \,^\dagger G)\bullet(\overline{\iota}_F\circ \,^\dagger G).$$ Further, since both sides of the above equations are by construction the same natural transformations between the adjoints of $F,G$, we only need to show that they both also have the same $\cE$-structures \eqref{Vstruct}. 
    
    The left-hand sides $\alpha^\dagger,\,^\dagger\alpha'$ by definition \eqref{adjmates} intertwines between the $\cE$-structures $a_v^\dagger = (a_v\rhd -)^\dagger$ on $F,G$. By chasing through some diagrams in the definition of planar-unitarity, \eqref{unitaryintw} states that the right-hand sides have the same $\cE$-structure, hence we achieve the desired equality. These adjunction-mates coincide because $\tilde S_v$ is unipotent, which implies $-^\dagger$ is involutive.
\end{proof}
\noindent By computing the double-adjoint $\alpha^{\dagger\dagger}$ in two different ways,
\begin{equation*}
    \big((e_G\circ F^\dagger)\bullet(G^\dagger\circ \alpha\circ F^\dagger)\bullet(G^\dagger\circ \iota_F)\big)^\dagger = (e_{F^\dagger}\circ G)\bullet(F\circ \alpha^\dagger\circ G)\bullet(F\circ \iota_{G^\dagger})
\end{equation*}
we can deduce the planar-pivotal pre-adjunction datum
\begin{equation*}
    e_F^\dagger  =\iota_{F^\dagger},\qquad \iota_{F}^\dagger = e_{F^\dagger}
\end{equation*}
on the hom-categories.

An immediate consequence of this is the following. Taking $\alpha=\id_{F}$ to be the identity natural transformation on an endofunctor $F:\cD\rightarrow\cD$, then we have the following \textbf{adjunction-snake equations}
\begin{equation*}
    \id_F = (e_F\circ F^\dagger)\bullet(F^\dagger\circ \iota_F) ,\qquad \id_{\,^\dagger F} = (\,^\dagger F\circ {e}_{F})\bullet({\iota}_{F}\circ \,^\dagger F).
\end{equation*}
Furthermore, since the adjunction is involutive, the above proposition as well as planar-unitarity implies
\begin{equation*}
    \iota_F^\dagger = e_{F^\dagger},\qquad e_F^\dagger = \iota_{F^\dagger}.
\end{equation*}
This is a part of the condition of pivotality for 1-categories \cite{BARRETT1999-1}.

\subsubsection{Adjunctions of the tensor product}
Recall that the tensor product of (planar-unitary) 2-representations are determined by the coproduct functor $\tilde\Delta$ on $\tilde C$. Though orientation reversal is contravariant on $\tilde C$, the fact that it does \textit{not} land in the comonoidal-opposite means the following
\begin{equation*}
        (\cD\boxtimes F')^\dagger = \cD\boxtimes F'^\dagger,\qquad (F\boxtimes\cA)^\dagger = F^\dagger\boxtimes\cA
\end{equation*}
for each $F:\cD\rightarrow\cD'$ and $F':\cA\rightarrow\cA'$. Moreover, the fact that the identity endofunctor $1_\cD$ has equipped the trivial $\cE$-structure given by the counit/identity arrow $\tilde\epsilon_1(a_v) = {\bf 1}_v$ gives
\begin{align*}
    & e_{\cD\boxtimes F'} = 1_\cD\boxtimes e_{F'},\qquad e_{F\boxtimes \cA} = e_F\boxtimes 1_{\cA}\\
    & \iota_{\cD\boxtimes F'} = 1_\cD\boxtimes\iota_{F'},\qquad \iota_{F\boxtimes\cA} = \iota_F\boxtimes1_{\cA}.
\end{align*}
From Definition of 2.2.3 of \cite{Douglas:2018}, we thus have the following.
\begin{theorem}\label{planaruni}
    The \textbf{planar-unitary} 2-representations form a \textbf{planar-pivotal} monoidal 2-category $\operatorname{2Rep}(\tilde C;\tilde{R})$.
\end{theorem}
In other words, planar-unitarity implies planar-pivotality, meaning that the end-categories of $\operatorname{2Rep}(\tilde C;\tilde{R})$ are themselves pivotal monoidal. To make a planar-pivotal 2-category bona fide \textit{pivotal}, \cite{Douglas:2018} introduced a notion of object-level duality satisfying several further coherence axioms. In the following, we shall do the same by using the antipode functor $\tilde S$, but we will see that we in fact do {\it not} produce a pivotal 2-category.

\subsection{Duality of 2-representations}\label{duals}
We now turn to the (object-level) duality in $\operatorname{2Rep}(\tilde C;\tilde R)$. We shall introduce the notion of a {\it left/right-dual} through the antipode functor $\tilde S$,\footnote{The functor $\tilde S^{-1}$ is interpreted as both left- and right-adjoint to $\tilde S$, witnessed by natural transformations 
\begin{equation*}
    \tilde S \circ \tilde S^{-1} \Rightarrow 1_{\tilde C} \Leftarrow \tilde S^{-1}\circ \tilde S.
\end{equation*}}
\begin{equation*}
    \rho^*(\zeta) = \rho(\tilde S\zeta),\qquad \,^*\rho(\zeta) = \rho(\tilde S^{-1}\zeta),\qquad \forall~\zeta\in\tilde C.
\end{equation*}
To express this definition more explicitly in terms of the action functor $\rhd$, we define the restrictions
\begin{equation*}
    \tilde S\mid_{\cV} = \tilde S_0,\qquad \tilde S\mid_{\cE} = \tilde S_1.
\end{equation*}
We can then define for each functor $F:\cD\rightarrow\cD'$ the following left- and right-\textit{mates},
\begin{align*}
    & \zeta\rhd (\cD'^*\xrightarrow{F^*} \cD^*) = (\tilde S_0a_v')\rhd \cD'^* \xrightarrow{\tilde S_1\gamma_e\rhd F^*} (\tilde S_0a_v)\rhd \cD^*, \\ 
    & \zeta \rhd (\,^*\cD'\xrightarrow{\,^*F}\,^*\cD) =  (\tilde S_0^{-1}a_v')\rhd \,^*\cD' \xrightarrow{\tilde S_1^{-1}\gamma_e\rhd \,^* F} (\tilde S_0^{-1}a_v)\rhd \,^*\cD,
\end{align*}
which inherits the $\cE$-structure given by $\tilde S_1\gamma_e \rhd-$ from that of the functor $F:\cD\rightarrow\cD'$ under $a_v$. 

Note that, in writing "$\tilde S_1^{-1}$" here, we are implicitly using the "fully-cofaithful-ness" of $\tilde S$; that is, it is bijective on arrows. Since, similar to orientation reversal, $\tilde S$ swaps the sources and targets on arrows, the left- and right-mates behaves in the following way
\begin{align*}
    & (\cD\xrightarrow{F}\cD'\xrightarrow{G}\cD'')^* = (\cD''^* \xrightarrow{G^*}\cD'^*\xrightarrow{F^*}\cD^*),\qquad (F\xRightarrow{\alpha}F')^* = (F^*\xRightarrow{\alpha^*}F'^*) \\
    & \,^*(\cD\xrightarrow{F}\cD'\xrightarrow{G}\cD'') = (\,^*\cD'' \xrightarrow{\,^* G}\,^*\cD'\xrightarrow{\,^*F}\,^*\cD),\qquad \,^*(F\xRightarrow{\alpha}F') = (\,^*F \xRightarrow{\,^*\alpha}\,^*F') 
\end{align*}
with respect to composition and the naturality of in the hom-categories.

\begin{rmk}
    By the hypothesis that $\tilde S$ is an equivalence, we have in fact from the definition that $\,^*\cD^*\cong \cD$, so the right-dual can be thought of as the "pre-left-dual". However, $\tilde S$ is in general not going to be unipotent $\tilde S^2\not\simeq 1_{\tilde C}$, hence neither the left- nor right-dualities are involutive, eg. $(\cD^*)^* \not\cong \cD$. As such, without assuming additional "pivotality conditions", our 2-category $\operatorname{2Rep}(\tilde C;\tilde R)$ cannot be bona fide pivotal, and must differ in certain respects from similar structures studied in the literature (eg. \cite{Douglas:2018,Johnson_Freyd_2023,Mackaay:ek}). We will discuss this in more detail later.
\end{rmk}

However, unlike orientation reversal, $\tilde S:\tilde C\rightarrow \tilde C^{\text{m-op,c-op}}$ (and its inverse) lands in the comonoidal opposite. This allows us to deduce the compatibility of mates against the tensor products which is different from the adjoints. To see this, we first note the fact that $\tilde S_1$ preserves the identity arrows in $\tilde C$. This then implies that the left- and right-dual preserves the identity functors for all $\cA$,
\begin{equation}
    1_{\cA^*} = 1_{\cA}^*,\qquad 1_{\,^*\cA} = \,^*1_\cA,\label{dualfunctor}
\end{equation}
since both sides of each equations above have the same $\tilde C$-module structure (under $\cV$). Thanks to this, we have achieve
\begin{align*}
    (\cD\boxtimes\cA\xrightarrow{F\boxtimes\cA} \cD'\boxtimes \cA)^* &= \cA^*\boxtimes\cD'^* \xrightarrow{\cA^*\boxtimes F^*}\cA^*\boxtimes \cD^* \\ 
    (\cA\boxtimes\cD\xrightarrow{\cA\boxtimes F} \cA\boxtimes \cD')^* &= \cD'^*\boxtimes\cA^* \xrightarrow{F^*\boxtimes \cA^*}\cD^*\boxtimes \cA^*
\end{align*}
for each $\tilde C$-module category $\cA$ and functors $F:\cD\rightarrow\cD'$. Similarly for the right-dual, as the inverse $\tilde S^{-1}$ works the same way. Note the condition $\tilde\epsilon \circ \tilde S = \tilde\epsilon = \tilde\epsilon \circ \tilde S^{-1}$ implies that $\cI = \cI^* = \,^*\cI$ on the nose. 

\subsubsection{Folds for the duals}
Now take a $\tilde C$-module category $\cD$ and an endofunctor $F\in\operatorname{End}(\cD)$. Define the functor $F^*\boxtimes F:\cD^*\boxtimes\cD\rightarrow \cD^*\boxtimes\cD$ by nudging. The following tensor product object-functor pairs admit the following $\tilde C$-module structure
\begin{align*}
    & \zeta\rhd (\cD^*\boxtimes\cD \xrightarrow{F^*\boxtimes F} \cD^*\boxtimes\cD) \\
    &\qquad \qquad = (\tilde S_0 \otimes 1)(\tilde\Delta_0)_{a_v}\rhd(\cD^*\boxtimes \cD)\xrightarrow{(\tilde S_1\otimes 1)(\tilde\Delta_1)_{\gamma_e}\rhd( F^*\boxtimes F)}(\tilde S_0 \otimes 1)(\tilde\Delta_0)_{a_v'}\rhd(\cD^*\boxtimes\cD),\\
    & \zeta\rhd (\cD\boxtimes\cD^* \xrightarrow{F\boxtimes F^*} \cD^*\boxtimes\cD^*) \\
    &\qquad \qquad = (1\otimes \tilde S_0)(\tilde\Delta_0)_{a_v}\rhd(\cD\boxtimes \cD^*)\xrightarrow{(1\otimes \tilde S_1)(\tilde\Delta_1)_{\gamma_e}\rhd( F\boxtimes F^*)}(1 \otimes \tilde S_0)(\tilde\Delta_0)_{a_v'}\rhd(\cD\boxtimes\cD^*),
\end{align*}
where $\zeta=a_v\xrightarrow{\gamma_e}a_v'\in\tilde C$. 

Now since the above action functors $\rho=\rhd$ are the same (they are all those of $\cD$), the module associator $(\rho(-)\otimes\rho(-))(-)\Rightarrow \rho(-\cdot -)(-)$ allows us to contract the 2-gauge transformations. The antipode axioms \eqref{antpod} then tell us that this the above $\tilde C$-module structures on $\cD^*\boxtimes\cD$, etc. are the same as that for $\cI$, along with its identity endomorphism $1_\cI$. This is witnessed by $\tilde C$-module functors,
\begin{equation*}
    \operatorname{ev}_\cD: \cD^*\boxtimes\cD\rightarrow \cI,\qquad \operatorname{cev}_\cD: \cI\rightarrow \cD\boxtimes\cD^*,
\end{equation*}
called the \textit{right-folds}, which fit into the following naturality diagrams,
\[\begin{tikzcd}
	{\cD^*\boxtimes\cD} & \cI \\
	{\cD^*\boxtimes\cD} & \cI
	\arrow["{\operatorname{ev}_\cD}", from=1-1, to=1-2]
	\arrow["{F^*\boxtimes F}"', from=1-1, to=2-1]
	\arrow["{1_\cI}", from=1-2, to=2-2]
	\arrow["{\operatorname{ev}_F}"', shorten <=5pt, shorten >=5pt, Rightarrow, from=2-1, to=1-2]
	\arrow["{\operatorname{ev}_{\cD}}"', from=2-1, to=2-2]
\end{tikzcd},\qquad 
\begin{tikzcd}
	\cI & {\cD\boxtimes\cD^*} \\
	\cI & {\cD\boxtimes\cD^*}
	\arrow["{\operatorname{cev}_\cD}", from=1-1, to=1-2]
	\arrow["{1_\cI}"', from=1-1, to=2-1]
	\arrow["{F\boxtimes F^*}", from=1-2, to=2-2]
	\arrow["{\operatorname{cev}_F}", shorten <=5pt, shorten >=5pt, Rightarrow, from=2-1, to=1-2]
	\arrow["{\operatorname{cev}_\cD}"', from=2-1, to=2-2]
\end{tikzcd}\] 
for each {\it endo}functor $F:\cD\rightarrow\cD$. We also have the following 2-morphisms
\begin{equation*}
    \operatorname{ev}_\cD\circ (F^*\boxtimes \cD) \Rightarrow \operatorname{ev}_\cD\circ({\cD^*}\boxtimes F),\qquad (F\boxtimes {\cD^*})\circ \operatorname{cev}_\cD\Rightarrow ({\cD}\boxtimes F)\circ \operatorname{cev}_\cD,
\end{equation*}
which satisfy the coherence condition 
\[\begin{tikzcd}[ampersand replacement=\&]
	{\cD^*\cD} \&\& {\cD^*\cD} \\
	\\
	{\cD^*\cD} \&\& {\cD^*\cD} \\
	\&\&\& \cI
	\arrow["{\cD^*\boxtimes F}", from=1-1, to=1-3]
	\arrow["{F^*\boxtimes\cD}"', from=1-1, to=3-1]
	\arrow["{\upsilon_{F^*,F}}"', shorten <=12pt, shorten >=12pt, Rightarrow, from=1-1, to=3-3]
	\arrow["{F^*\boxtimes\cD}"', from=1-3, to=3-3]
	\arrow["{\operatorname{ev}_\cD}", from=1-3, to=4-4]
	\arrow["{\cD^*\boxtimes F}", from=3-1, to=3-3]
	\arrow["{\operatorname{ev}_\cD}"', from=3-1, to=4-4]
	\arrow[shorten <=5pt, shorten >=5pt, Rightarrow, from=3-3, to=4-4]
\end{tikzcd} =
\begin{tikzcd}[ampersand replacement=\&]
	{\cD^*\cD} \\
	\& \cI \\
	{\cD^*\cD}
	\arrow["{\operatorname{ev}_\cD}", from=1-1, to=2-2]
	\arrow[""{name=0, anchor=center, inner sep=0}, "{F^*\boxtimes F}"', from=1-1, to=3-1]
	\arrow[from=3-1, to=2-2]
	\arrow["{\operatorname{ev}_F}"', shorten <=6pt, shorten >=6pt, Rightarrow, from=0, to=2-2]
\end{tikzcd}\]
similar to that of a \textit{2-coend} \cite{Loregian_2021}.
At the unit $\cD=\cI$, we of course have
\begin{equation*}
    \operatorname{ev}_\cI = 1_\cI = \operatorname{cev}_\cI 
\end{equation*}
through the unitors $r_\cI,\ell_\cI=1_\cI$. A similar construction with the inverse antipode $\tilde S^{-1}$ yields the left-folds
\begin{equation*}
    \overline{\operatorname{ev}}_\cD: \cD\boxtimes \,^*\cD\rightarrow \cI,\qquad \overline{\operatorname{cev}}_\cD:\cI\rightarrow \,^*\cD\boxtimes \cD.
\end{equation*}
We shall without loss of essential generality focus on the right-duals in the following.

\begin{rmk}\label{preduals}
    The reason we can just focus on the right-duals is the following. By the hypothesis that $\tilde S$ is an equivalence and $\,^*\cD^*\cong \cD$, the right-folds admit invertible $\tilde C$-module natural transformations
\begin{equation*}
    \overline{\operatorname{ev}}_{\cD^*}\cong \operatorname{ev}_\cD,\qquad \overline{\operatorname{cev}}_{\cD^*}\cong \operatorname{cev}_\cD.
\end{equation*}
This allows us to transport all arguments that we shall make for the right-dual/folds to the left-dual/folds, but this does {\it not} force the left- and right-duals to coincide. Another subtlety is that we should not treat this property as giving the left-duality datum $\cD^*$ given the right-duality datum of $\cD$; this will be elaborated more in \S \ref{2catdimension}.
\end{rmk}


\subsubsection{Folds on tensor products}
Now consider the fold map $\operatorname{ev}_{\cD\boxtimes\cD'}: (\cD\boxtimes\cD')^*\boxtimes (\cD\boxtimes \cD') \rightarrow \cI$. We can also achieve a map of this form by using the property $(\cD\boxtimes\cD')^* = \cD'^*\boxtimes \cD^*$. Indeed, the following series of 1-morphisms gives
\begin{align*}
    & (\cD'^*\boxtimes \cD^*)\boxtimes (\cD\boxtimes \cD')\xrightarrow{\alpha^{\tilde C}_{(\cD'^*\boxtimes\cD^*)\cD\cD'}} \cD'^*\boxtimes (\cD^*\boxtimes(\cD\boxtimes\cD')) \\
    &\qquad \xrightarrow{\cD^*\boxtimes(\alpha^{\tilde C}_{\cD'^*\cD'\cD})^{-1}} \cD'^*\boxtimes ((\cD^*\boxtimes\cD)\boxtimes\cD') \xrightarrow{\cD'^*\boxtimes\operatorname{ev}_{\cD}\boxtimes\cD'} \cD'^*\boxtimes (\cI\boxtimes\cD') \\
    &\qquad \xrightarrow{\cD'^*\boxtimes\ell_{\cD'}} \cD'^*\boxtimes \cD' \xrightarrow{\operatorname{ev}_{\cD'}} \cI.
\end{align*}
By construction, this functor admits the same $\tilde C$-module structure as $\operatorname{ev}_{\cD\boxtimes\cD'}$; similar computations hold for the other folds. The \textit{fold condition} then states that these functors coincide {on-the-nose}. Suppressing the associators and the unitors, we thus get 
\begin{equation}
    \operatorname{ev}_{\cD\boxtimes\cD'} = \operatorname{ev}_{\cD'}\circ (\cD'^*\boxtimes\operatorname{ev}_{\cD}\boxtimes\cD'),\qquad \operatorname{cev}_{\cD\boxtimes\cD'} = (\cD\boxtimes \operatorname{cev}_{\cD'}\boxtimes\cD^*) \circ \operatorname{cev}_{\cD}.\label{foldtensor}
\end{equation}
We now prove a consistency formula for the mates $F^*,\,^*F$ of a functor $F:\cD\rightarrow \cD'$. 
\begin{proposition}\label{mates}
    There are invertible $\tilde C$-module natural transformations that identify the left- and right-mates
    \begin{align*}
        & F^* \cong \ell_{\cD^*}\circ (\operatorname{ev}_{\cD'}\boxtimes\cD^*)\circ (\cD'^* \boxtimes F \boxtimes \cD^*) \circ (\cD'^*\boxtimes\operatorname{cev}_{\cD}) \circ r_{\cD'^*}^{-1},\\
        & \,^* F \cong r_{\,^*\cD}\circ (\,^*\cD\boxtimes \overline{\operatorname{ev}}_{\cD'})\circ (\,^*\cD \boxtimes F \boxtimes \,^*\cD') \circ (\overline{\operatorname{cev}}_{\cD}\boxtimes\,^*\cD') \circ \ell_{\,^*\cD'}^{-1}
    \end{align*}
    of a functor $F:\cD\rightarrow \cD'$.
\end{proposition}
\begin{proof}
    We can use the above fold maps to construct a functor $F'^*:\cD'^*\rightarrow \cD^*$ by
\begin{equation*}
    F'^* = \ell_{\cD^*}\circ (\operatorname{ev}_{\cD'}\boxtimes\cD)\circ (\cD'^* \boxtimes F \boxtimes \cD^*) \circ (\cD'^*\boxtimes\operatorname{cev}_{\cD}) \circ r_{\cD'^*}^{-1},
\end{equation*}
where we have used  the invertible associator 2-morphism arising from the coassociativity of $\tilde\Delta_1$ to disambiguiate the middle factor
\begin{equation*}
    \alpha_{\cD'^*F\cD^*}^{-1}: \cD'^* \boxtimes (F \boxtimes \cD^*)\Rightarrow (\cD'^*\boxtimes F)\boxtimes \cD^*.
\end{equation*}
By construction, $F^*,F'^*$ takes the same values on objects, hence it suffices to prove that they have the same $\cE$-structure \eqref{Vstruct}. Recall $F^*$ inherits a $\cE$-structure given by $\tilde S_1\gamma_e$. The computation using the antipode \eqref{antpod} and counit axioms lead to a series of natural isomorphisms
\begin{align*}
    (-\times-\otimes 1)\circ (\tilde S_1\otimes & 1\otimes\tilde S_1) \circ (\tilde\Delta_1\otimes1)\circ\tilde\Delta_1 \\
    & = (((-\times-)\circ (\tilde S_1\otimes 1)\circ \tilde\Delta_1) \otimes 1 )\circ ((1\otimes \tilde S_1)\circ \tilde \Delta_1 \\ 
    &= (\tilde\epsilon_1 \cdot1 \otimes 1)\circ ((1\otimes\tilde S_1)\circ\tilde\Delta) \\
    &= \tilde S_1\circ(\tilde\epsilon_1\otimes1)\circ\tilde\Delta = \tilde S_1\circ \id,
\end{align*}
which determines the same $\cE$-structure as $F'^*$. This gives the $\tilde C$-module functor identification $F^*\cong F'^*$. The same argument with $\tilde S^{-1}$ works for the left-mate.
\end{proof}
\noindent By computing the double dual-mate $F^{**}$ in two different ways, we can deduce 
\begin{equation}
 \operatorname{ev}^*_{\cD} \cong \operatorname{cev}_{\cD^*},\qquad \operatorname{cev}^*_\cD \cong \operatorname{ev}_{\cD^*}.\label{dualintertwine}
\end{equation}
Note this does \textit{not} determine the object-level pre-duality datum, since $\operatorname{ev}_F$ is only defined for endofunctors.

\paragraph{A pivotality condition.} Recall that we have the identification $\,^*\cD^*\cong \cD$, thus if we replace $\cD$ by its left-dual $\cD^*$, then the above proposition allows us to write the {\it right}-mate of a functor $F:\cD^*\rightarrow \cD'^*$ as
\begin{equation*}
    \,^*F \cong r_{\cD}\circ (\cD\boxtimes \overline{\operatorname{ev}}_{\cD'^*})\circ (\cD \boxtimes F \boxtimes \,^*\cD') \circ (\overline{\operatorname{cev}}_{\cD^*}\boxtimes \cD') \circ \ell_{\cD'}^{-1}: \cD'\rightarrow\cD.
\end{equation*}
If we further impose a condition in which the left-dual is involutive $(\cD^*)^*\cong \cD$ such that the left- and right-folds coincide, then we acquire a {pivotality condition} $\,^*\cD\cong \cD^*$, and the above formula recovers the one given in pg. 49 of \cite{Douglas:2018} (modulo the convention in which the duals are folded). 

Of course, this condition in general does not hold unless $\tilde S^2\cong 1_{\tilde C}$ is unipotent --- namely $\tilde C$ itself has a "copivotal" structure. This is why $\operatorname{2Rep}(\tilde C;\tilde R)$ is not a bona fide pivotal 2-category in the sense of \cite{Douglas:2018} --- the failure is measured by the difference between the notions of left- and right-duality $\cD^*,\,^*\cD$. We will analyze this issue in much greater detail in \S \ref{2catdimension} and \S \ref{ribbon2cat}.

\begin{rmk}\label{classicalantipode}
    Note $\tilde S=\tilde S$ is indeed unipotent in the undeformed/classical case, as it is merely given by the (horizontal) inversion $\zeta\mapsto \zeta^{-1}$ of 2-gauge parameters. Here, the 2-$R$-matrix $\tilde R= \id\otimes\id$ is simply the unit, hence it would be possible for $\operatorname{2Rep}(\mathbb{U}_{q=0}\G;\id\otimes\id)$ to be pivotal. This situation is similar to many well-known examples of compact quantum groups \cite{Majid:1996kd,Woronowicz1988}: the quantum deformation destroys the symmetry of their representation categories.
\end{rmk}

\subsubsection{snakerators/cusps/cusps of the left dual}
An immediate consequence of {\bf Proposition \ref{mates}} is the following. Setting $F = 1_\cD \in \operatorname{End}(\cD)$, we obtain a formula for the left- and right-mates of the identity
\begin{align*}
    & 1_\cD^* \cong \ell_{\cD^*}\circ (\operatorname{ev}_{\cD}\boxtimes\cD^*)\circ (\cD^*\boxtimes\operatorname{cev}_{\cD}) \circ r_{\cD^*}^{-1},\\
    & \,^* 1_\cD \cong r_{\,^*\cD}\circ (\,^*\cD\boxtimes \overline{\operatorname{ev}}_{\cD}) \circ (\overline{\operatorname{cev}}_{\cD}\boxtimes\,^*\cD) \circ \ell_{\,^*\cD}^{-1}.
\end{align*}
However, we know from \eqref{dualfunctor} that these are in fact formulas for $1_{\cD^*},1_{\,^*\cD}$. The identification $\cD\cong \,^*\cD^*$ and {\it Remark \ref{preduals}} then allow us to define the following invertible 2-morphisms (neglecting the invertible associators and unitors) 
\begin{equation}
    \varrho_{\cD}: 1_{\cD^*}\Rightarrow (\operatorname{ev}_{\cD}\boxtimes\cD^*)\circ (\cD^*\boxtimes\operatorname{cev}_{\cD}),\qquad \varphi_\cD: (\cD\boxtimes \operatorname{ev}_{\cD})\circ ({\operatorname{cev}}_{\cD}\boxtimes \cD)\Rightarrow 1_{\cD},\label{cusps}
\end{equation}
which we call the \textit{snakerators/cusps} for the folds. It is then easy to see that 
\begin{equation*}
    \varphi_\cI = \id_{1_\cI},\qquad \varrho_{\cI} = \id_{1_\cI}
\end{equation*}
are identity 2-morphisms. 

\begin{proposition}\label{dualevals}
    There are 2-morphisms such that
    \begin{equation*}
        a_\cD:\,^*\operatorname{ev}_\cD\Rightarrow \overline{\operatorname{cev}}_\cD,\qquad b_\cD: \overline{\operatorname{ev}}_\cD\Rightarrow \,^*\operatorname{cev}_\cD.
    \end{equation*}
\end{proposition}
\begin{proof}
    First, from {\it Remark \ref{preduals}} and \eqref{foldtensor}, we have
    \begin{equation*}
        \operatorname{cev}_{\,^*\cD\boxtimes\cD} = (\,^*\cD\boxtimes\operatorname{cev}_\cD\boxtimes\cD)\circ\operatorname{cev}_{\,^*\cD}\cong (\,^*\cD\boxtimes\operatorname{cev}_\cD\boxtimes\cD)\circ\overline{\operatorname{cev}}_{\cD}.
    \end{equation*}
    Then from the fact that $\overline{\operatorname{ev}}_{\cI} \cong 1_\cI$ through the invertible unitors, {\bf Proposition \ref{mates}} gives 
    \begin{align*}
        \,^*\operatorname{ev}_\cD &\cong  (\,^*(\cD^*\boxtimes\cD)\boxtimes\operatorname{ev}_\cD\boxtimes\cI)\circ (\overline{\operatorname{cev}}_{\,^*\cD\boxtimes\cD}\boxtimes\cI) \\
        & = \,^*\cD \boxtimes\big( (\cD\boxtimes\operatorname{ev}_\cD)\circ (\operatorname{cev}_\cD\boxtimes\cD)\big) \circ \overline{\operatorname{cev}}_\cD,
    \end{align*}
    which by \eqref{cusps} admits a snakerator $\,^*\cD \boxtimes\varphi_\cD \circ \overline{\operatorname{cev}}_\cD $ into 
    \begin{equation*}
        \,^*\cD \boxtimes 1_\cD \circ \overline{\operatorname{cev}}_\cD = \overline{\operatorname{cev}}_\cD,
    \end{equation*}
    as desired. A similar argument can be applied to $\,^*\operatorname{cev}_\cD$ by using the snakerator $\varrho_\cD$.
\end{proof}

Now consider the tensor product $\cD\boxtimes\cA$ and the snakerator $\varphi_{\cD\boxtimes\cA}$. The domain of this 2-morphism is a 1-morphism
\begin{equation*}
    \cD\boxtimes\cA \rightarrow (\cD\boxtimes\cA)\boxtimes(\cD\boxtimes\cA)^*\boxtimes(\cD\boxtimes\cA) \rightarrow \cD\boxtimes\cA,
\end{equation*}
in which the object in the centre is $(\cD\boxtimes\cA)\boxtimes(\cA^*\boxtimes\cD^*)\boxtimes(\cD\boxtimes\cA)$. By naturality, this object fits into a diagram of the form (here we have neglected the associators and the symbol $\boxtimes$ to save space)
\[\begin{tikzcd}
	{\cD\cA} && {\cD\cD^*\cD\cA} && {\cD\cA\cA^*\cD^*\cD\cA} \\
	&& {\cD\cA} && {\cD\cA\cA^*\cA} \\
	&&&& {\cD\cA}
	\arrow["{{\operatorname{cev}}_{\cD}\cD\cA}", from=1-1, to=1-3]
	\arrow[""{name=0, anchor=center, inner sep=0}, "{1_{\cD\cA}}"', curve={height=12pt}, from=1-1, to=2-3]
	\arrow["{\cD{\operatorname{cev}}_{\cA}\cD^*\cD\cA}", from=1-3, to=1-5]
	\arrow["{\cD{\operatorname{ev}}_{\cD}\cA}", from=1-3, to=2-3]
	\arrow["{\cD\cA\cA^*{\operatorname{ev}}_{\cD}\cA}", from=1-5, to=2-5]
	\arrow[shorten <=15pt, shorten >=15pt, Rightarrow, from=2-3, to=1-5]
	\arrow["{\cD{\operatorname{cev}}_{\cA}\cA}"{description}, from=2-3, to=2-5]
	\arrow[""{name=1, anchor=center, inner sep=0}, "{1_{\cD\cA}}"', curve={height=30pt}, from=2-3, to=3-5]
	\arrow["{\cD\cA{\operatorname{ev}}_{\cA}}", from=2-5, to=3-5]
	\arrow["{\varphi_{\cD^*}\cA}", shorten <=7pt, Rightarrow, from=0, to=1-3]
	\arrow["{\cD\varphi_{\cA^*}}", shorten <=12pt, shorten >=12pt, Rightarrow, from=1, to=2-5]
\end{tikzcd}\]
in which the 2-cell in the middle is filled by the interchanger $\cD \upsilon_{{\operatorname{cev}}_{\cA},{\operatorname{ev}}_{\cD}}\cA$. In other words, we have the following formula
\begin{equation*}
    \varphi_{\cD\cA} = (\varphi_\cD\cA \circ\cD\varphi_\cA)\bullet({\operatorname{ev}}_{\cD}\cD\cA\circ \cD \upsilon_{{\operatorname{cev}}_{\cA},{\operatorname{ev}}_{\cD}}\cA\circ \cD\cA{\operatorname{cev}}_{\cA}).
\end{equation*}
Similar arguments lead to the formula
\begin{equation*}
    \varrho_{\cD\cA} = (\cA^*\cD^*{\operatorname{ev}}_{\cD}\circ \cA^* \upsilon_{{\operatorname{ev}}_{\cD},{\operatorname{cev}}_{\cA}}\cD^*\circ {\operatorname{cev}}_{\cA}\cA^*\cD^*)\bullet (\cA^*\varrho_\cD \circ \varrho_\cA\cD^*).
\end{equation*}

\subsubsection{The swallowtail 2-morphisms}\label{swallotail}
Consider the identity natural transformation on the fold $\id_{\operatorname{ev}_\cD}: \operatorname{ev}_\cD\Rightarrow\operatorname{ev}_\cD$, whose legs are $\cD^*\cD\rightarrow \cI$. Notice that there are two ways in which to write the identity functor $1_{\cD\boxtimes\cD^*} = 1_\cD\boxtimes\cD^* = \cD\boxtimes1_{\cD^*}$. From  \eqref{cusps}, the former has a cusp $\cD^*\varphi_\cD$ into it and the latter has a cusp $\varrho_\cD\cD$ out of it. The consistency of these two expressions for the cusps are mediated by the interchanger, as can be seen in the following diagram
\[\mathcal{S}'_\cD = \begin{tikzcd}
	& {\cD^*\cD} & {\cD^*\cD\cD^*\cD} & {\cD^*\cD} \\
	{\cD^*\cD\cD^*\cD} && \cI \\
	& {\cD^*\cD} & {\cD^*\cD\cD^*\cD} & {\cD^*\cD}
	\arrow[from=1-2, to=1-3]
	\arrow[""{name=0, anchor=center, inner sep=0}, "{1_{\cD^*}\cD}", curve={height=-30pt}, from=1-2, to=1-4]
	\arrow[""{name=1, anchor=center, inner sep=0}, from=1-2, to=2-3]
	\arrow[from=1-3, to=1-4]
	\arrow[from=1-4, to=2-3]
	\arrow[""{name=2, anchor=center, inner sep=0},"=", from=1-4, to=3-4]
	\arrow[from=2-1, to=1-2]
	\arrow[""{name=3, anchor=center, inner sep=0}, from=2-1, to=3-2]
	\arrow[from=3-2, to=2-3]
	\arrow[from=3-2, to=3-3]
	\arrow[""{name=4, anchor=center, inner sep=0}, "{\cD^*1_{\cD}}"', curve={height=30pt}, from=3-2, to=3-4]
	\arrow[from=3-3, to=3-4]
	\arrow[from=3-4, to=2-3]
	\arrow["{\varrho_\cD\cD}", shorten <=3pt, Rightarrow, from=0, to=1-3]
	\arrow["{\id_{\operatorname{ev}_\cD}}"', shorten <=11pt, shorten >=11pt, Rightarrow, from=2, to=2-3]
	\arrow["{\upsilon_{\operatorname{ev}_\cD,\operatorname{ev}_\cD}}", shorten <=13pt, shorten >=13pt, Rightarrow, from=3, to=1]
	\arrow["{\cD^*\varphi_\cD}", shorten >=3pt, Rightarrow, from=3-3, to=4]
\end{tikzcd}.\]
Similarly, the identity natural transform $\id_{\operatorname{cev}_\cD}: \operatorname{cev}_\cD\Rightarrow\operatorname{cev}_\cD$ has each of it legs given by $\cI\rightarrow \cD\cD^*$. The identity functor $1_{\cD\cD^*} = \cD1_{\cD^*}=1_{\cD}\cD^*$ leads to the following 
\[\mathcal{S}_\cD=\begin{tikzcd}
	& {\cD\cD^*} & {\cD\cD^*\cD\cD^*} & {\cD\cD^*} \\
	{\cD\cD^*\cD \cD^*} && \cI \\
	& {\cD \cD^*} & {\cD\cD^*\cD \cD^*} & {\cD \cD^*}
	\arrow[from=1-2, to=2-1]
	\arrow[from=1-3, to=1-2]
	\arrow[""{name=0, anchor=center, inner sep=0}, "{\cD 1_{\cD^*}}"', curve={height=30pt}, from=1-4, to=1-2]
	\arrow[from=1-4, to=1-3]
	\arrow[""{name=1, anchor=center, inner sep=0}, "{=}", from=1-4, to=3-4]
	\arrow[""{name=2, anchor=center, inner sep=0}, from=2-3, to=1-2]
	\arrow[from=2-3, to=1-4]
	\arrow[from=2-3, to=3-2]
	\arrow[from=2-3, to=3-4]
	\arrow[""{name=3, anchor=center, inner sep=0}, from=3-2, to=2-1]
	\arrow[from=3-3, to=3-2]
	\arrow[""{name=4, anchor=center, inner sep=0}, "{1_{\cD} \cD^*}", curve={height=-30pt}, from=3-4, to=3-2]
	\arrow[from=3-4, to=3-3]
	\arrow["{\cD \varrho_\cD}", shorten <=3pt, Rightarrow, from=0, to=1-3]
	\arrow["{\upsilon_{\operatorname{cev}_\cD,\operatorname{cev}_\cD}}", shorten <=13pt, shorten >=13pt, Rightarrow, from=2, to=3]
	\arrow["{\id_{\operatorname{cev}_\cD}}"', shorten <=11pt, shorten >=11pt, Rightarrow, from=2-3, to=1]
	\arrow["{\varphi_\cD \cD^*}", shorten >=3pt, Rightarrow, from=3-3, to=4]
\end{tikzcd}.\]
The commutativity of these diagrams of 2-morphisms, namely the equations $$\mathcal{S}_\cD=\id_{1_{\cD\boxtimes\cD^*}},\qquad \mathcal{S}'_\cD =\id_{1_{\cD^*\boxtimes\cD}},$$ are known as the \textbf{swallowtail equations}; see C1, Definition 2.2.4 in \cite{Douglas:2018}, as well as fig. 19 (c,d) in \cite{Barrett_2024}. However, we shall impose a more general condition involving a larger diagram of 2-morphisms, which "pastes" these swallowtail diagrams together using the braiding structure $c$. This condition is explained in \S \ref{adj-dualbraid}.

\medskip

So far, the above sections dealt with structures that have been known for some time. They were studied mainly in the context of {\it spherical} 2-categories \cite{Mackaay:ek,Douglas:2018,Johnson_Freyd_2023}, but the fact that the left- and right-duals do not coincide means that $\operatorname{2Rep}(\tilde C;\tilde R)$ is not pivotal, and hence cannot be spherical. In the following, we will work to analyze its structures more thoroughly.

\section{Rigid dagger tensor 2-category $\operatorname{2Rep}(\tilde C;\tilde R)$}\label{rigidagger}
Let us now return to $\operatorname{2Rep}(\tilde C;\tilde R)$ and investigate the interplay between its duality and adjunction.  Recall in \S \ref{adjoints} that we have introduced a notion of {\it planar-unitarity}. The natural condition to impose is then to ask for all of the structural functors, eg. the associators $\alpha$, unitors $r_\cD,\ell_\cD$, and particularly the folds $\operatorname{ev}_\cD,\operatorname{cev}_\cD$, to be planar-unitary. 

Following the definition of a \textit{rigid} dagger tensor category in \cite{Jones:2017}, we propose the following definition.
\begin{definition}
    A \textit{rigid dagger tensor 2-category} is a tensor 2-category with adjoints equipped with (natural) left-duals $\cD^*$ for each object $\cD$ such that
    \begin{enumerate}
        \item there exists a "left-dual" $\,^*\cD$ such that it is the pre-dual of the right-dual, $\,^*\cD^*\cong \cD$, and
        \item all folds $\operatorname{ev}_\cD,\operatorname{cev}_{\cD}$ and snakerators/cusps $\varphi_\cD,\varrho_\cD$ are planar-unitary. 
    \end{enumerate}
\end{definition}
\noindent It is sensible to impose these conditions on $\operatorname{2Rep}(\tilde C;\tilde R)$, as its duals and adjunctions --- which we recall are defined by the antipode and the dagger structure $\tilde S,\tilde S_v=-^\dagger$ --- are compatible.
\begin{equation*}
    \tilde S^\text{op}\circ \tilde S_v = \tilde S_v^\text{m-op,c-op}\circ \tilde S.
\end{equation*}
The goal in this section is to exhibit the coherence relations which make this definition more transparent.

\begin{rmk}
    Note that in \cite{Deligne1982}, the definition of a "rigid tensor category" has reflexitivity built in: there are isomorphisms $D\cong (D^*)^*$ trivializing the double dual of every object $D$. We do not a priori assume this property for $\operatorname{2Rep}(\tilde C;\tilde R)$, but we shall see in \S \ref{reidemeisterI} that those objects which \textit{are} reflexive have particular geometric properties.
\end{rmk}

\subsection{Compatibility between the folds}
We begin by noting that duality and adjunctions {\it strongly} commute on endofunctors $F:\cD\rightarrow \cD$,
\begin{equation}
    (F^\dagger )^*  = (F^*)^\dagger,\qquad \,^*(F^\dagger) = (\,^* F)^\dagger,\label{daggercommute}
\end{equation}
which allows the following triangles to commute,
\[\begin{tikzcd}
	& {\cD^*\boxtimes \cD} \\
	{\cD^*\boxtimes \cD} && {\cD^*\boxtimes \cD}
	\arrow["{F^*\boxtimes F}", from=1-2, to=2-3]
	\arrow["{(F^\dagger)^*\boxtimes F^\dagger}", from=2-1, to=1-2]
	\arrow[""{name=0, anchor=center, inner sep=0}, "{((F^*)^\dagger\circ F^*)\boxtimes(F^\dagger \circ F)}"', from=2-1, to=2-3]
	\arrow["\cong"{description}, draw=none, from=1-2, to=0]
\end{tikzcd},\qquad 
\begin{tikzcd}
	& {\cD\boxtimes \cD ^*} \\
	{\cD\boxtimes \cD ^*} && {\cD\boxtimes \cD ^*}
	\arrow["{F^\dagger\boxtimes (F^\dagger)^*}", from=1-2, to=2-3]
	\arrow["{F\boxtimes F^*}", from=2-1, to=1-2]
	\arrow[""{name=0, anchor=center, inner sep=0}, "{((F^*)^\dagger\circ F^*)\boxtimes(F^\dagger \circ F)}"', from=2-1, to=2-3]
	\arrow["\cong"{description}, draw=none, from=1-2, to=0]
\end{tikzcd}.\]
These then lead to the following 2-morphism commutative diagrams (neglecting to label the 1-morphisms),
\[\begin{tikzcd}
	\cI & \cI & \cI \\
	& {\cD^* \boxtimes\cD} \\
	{\cD^* \boxtimes\cD} && {\cD^* \boxtimes\cD} \\
	\\
	\cI && \cI
	\arrow["\shortmid\shortmid"{marking}, from=1-1, to=1-2]
	\arrow["\shortmid\shortmid"{marking}, from=1-2, to=1-3]
	\arrow[from=2-2, to=1-2]
	\arrow["{\operatorname{ev}_{F}}"', shorten <=7pt, shorten >=7pt, Rightarrow, from=2-2, to=1-3]
	\arrow[from=2-2, to=3-3]
	\arrow["\shortmid"{marking}, from=3-1, to=1-1]
	\arrow["{\operatorname{ev}_{F^\dagger}}", shift left=3, shorten <=10pt, shorten >=10pt, Rightarrow, from=3-1, to=1-2]
	\arrow[from=3-1, to=2-2]
	\arrow[""{name=0, anchor=center, inner sep=0}, from=3-1, to=3-3]
	\arrow[""{name=1, anchor=center, inner sep=0}, "{1_{\cD^* \boxtimes\cD}}"{description}, curve={height=30pt}, from=3-1, to=3-3]
	\arrow["\shortmid"{marking}, from=3-1, to=5-1]
	\arrow["\shortmid"{marking},from=3-3, to=1-3]
	\arrow["\shortmid"{marking}, from=3-3, to=5-3]
	\arrow[""{name=2, anchor=center, inner sep=0}, "\shortmid\shortmid"{marking}, from=5-1, to=5-3]
	\arrow["\cong"{description, pos=0.4}, draw=none, from=2-2, to=0]
	\arrow["{e_{F^*}\boxtimes e_{F}}", shorten <=4pt, shorten >=4pt, Rightarrow, from=0, to=1]
	\arrow["{\operatorname{ev}_{1_{\cD^* \boxtimes\cD}}}"'{pos=0.6}, shorten <=9pt, shorten >=2pt, Rightarrow, from=1, to=2]
\end{tikzcd},\qquad 
\begin{tikzcd}
	\cI & \cI & \cI \\
	& {\cD \boxtimes\cD^*} \\
	{\cD \boxtimes\cD^*} && {\cD \boxtimes\cD^*} \\
	\\
	\cI && \cI
	\arrow["\shortmid\shortmid"{marking}, from=1-1, to=1-2]
	\arrow["{\operatorname{cev}_F}"', shift right=2, shorten <=7pt, shorten >=7pt, Rightarrow, from=1-1, to=2-2]
	\arrow["\shortmid"{marking}, from=1-1, to=3-1]
	\arrow["\shortmid\shortmid"{marking}, from=1-2, to=1-3]
	\arrow[from=1-2, to=2-2]
	\arrow["{\operatorname{cev}_{F^\dagger}}", shift left=3, shorten <=10pt, shorten >=10pt, Rightarrow, from=1-2, to=3-3]
	\arrow[from=1-3, to=3-3]
	\arrow[from=2-2, to=3-3]
	\arrow[from=3-1, to=2-2]
	\arrow[""{name=0, anchor=center, inner sep=0}, from=3-1, to=3-3]
	\arrow[""{name=1, anchor=center, inner sep=0}, "{1_{\cD \boxtimes\cD^*}}"{description}, curve={height=30pt}, from=3-1, to=3-3]
	\arrow["\shortmid"{marking}, from=5-1, to=3-1]
	\arrow[""{name=2, anchor=center, inner sep=0}, "\shortmid\shortmid"{marking}, from=5-1, to=5-3]
	\arrow["\shortmid"{marking}, from=5-3, to=3-3]
	\arrow["\cong"{description, pos=0.4}, draw=none, from=2-2, to=0]
	\arrow["{\iota_F\boxtimes \iota_{F^*}}"', shorten <=4pt, shorten >=4pt, Rightarrow, from=1, to=0]
	\arrow["{\operatorname{cev}_{1_{\cD\boxtimes\cD^*}}}"{pos=0.6}, shorten <=5pt, shorten >=7pt, Rightarrow, from=2, to=1]
\end{tikzcd}\]

Since the bottom square commutes on the nose, $\operatorname{ev}_{1_{\cD^*\boxtimes\cD}},\operatorname{cev}_{1_{\cD^*\boxtimes\cD}}$ are trivial whence we obtain the following compatibility between the left-folds and the adjunction-folds
\begin{align}
    & \operatorname{ev}_{F^\dagger}\circ \operatorname{ev}_F = e_{F^*}\boxtimes e_{F},\qquad \operatorname{cev}_{F^\dagger}\circ\operatorname{cev}_F= \iota_F\boxtimes \iota_{F^*}, \nonumber\\
    &  \operatorname{ev}_{F}\circ\operatorname{ev}_{F^\dagger} =(\iota_{F^*}\boxtimes\iota_F)^{-1},\qquad \operatorname{cev}_F\circ\operatorname{cev}_{F^\dagger}= (e_F\boxtimes e_{F^*})^{-1},\label{foldfold}
\end{align}
where the second row of conditions can be deduced from swapping the order of the composition of $F,F^\dagger$ in the central triangles above.

\subsection{Defect decorations on closed surfaces}\label{2catdimension}
Consider the adjunction-folds on the left-fold maps,
\begin{align*}
    &e_{\operatorname{ev}_\cD}: \operatorname{ev}_\cD^\dagger\circ \operatorname{ev}_\cD \Rightarrow 1_{\cD^*\boxtimes\cD},\qquad \iota_{\operatorname{ev}_\cD}: 1_{\cI}\Rightarrow \operatorname{ev}_\cD\circ\operatorname{ev}_\cD^\dagger,\\
    & e_{\operatorname{cev}_\cD}: \operatorname{cev}_\cD^\dagger\circ\operatorname{cev}_\cD\Rightarrow 1_{\cI},\qquad \iota_{\operatorname{cev}_\cD}: 1_{\cD\boxtimes\cD^*}\Rightarrow \operatorname{cev}_\cD\circ\operatorname{cev}_\cD^\dagger.
\end{align*}
These were called respectively the "crotch", "birth-of-a-circle", "saddle", and "death-of-a-circle" 2-morphisms in \cite{Douglas:2018} (see figs. 27 \& 28 in \cite{BAEZ2003705}, and also fig. 19 (a,b) in \cite{Barrett_2024}); collectively, we shall refer to them as the \textit{fold-on-folds} 2-morphisms. These can be used to construct surface defects decorated with 2-morphisms in $\operatorname{2Rep}(\tilde C;\tilde R)$. 

By naturality and planar-unitarity, these 2-morphisms admit an action by $\tilde C$-module endofunctors $F\in\operatorname{End}(\cD)$ via
\begin{align}
    & e_{\operatorname{ev}_\cD}\mapsto e_{\operatorname{ev}_\cD} \bullet (\operatorname{ev}_F^\dagger\circ\operatorname{ev}_F) ,\qquad \iota_{\operatorname{ev}_\cD}\mapsto (\operatorname{ev}_F\circ \operatorname{ev}_F^\dagger)\bullet \iota_{\operatorname{ev}_\cD},\nonumber \\
    & e_{\operatorname{cev}_\cD}\mapsto e_{\operatorname{cev}_\cD}\bullet (\operatorname{cev}_F\circ\operatorname{cev}_F^\dagger),\qquad \iota_{\operatorname{cev}_\cD}\mapsto (\operatorname{cev}_F^\dagger\circ\operatorname{cev}_F)\bullet\iota_{\operatorname{cev}_\cD}.\label{F-module}
\end{align}
It is important to note here that since our left- and right-duals do not coincide, $\operatorname{cev}_{\cD^*}: \cI\rightarrow \cD^*\boxtimes(\cD^*)^*$ is not parallel with $\operatorname{ev}_\cD^\dagger$. We will see how this can be addressed through the ribbon balancings in \S \ref{ribbon}.

\begin{rmk}\label{warning}
Suppose for the moment that the left- and right-duals coincide, so that the duality is involutive. The fact that $\operatorname{cev}_{\cD^*}: \cI\rightarrow \cD^*\boxtimes(\cD^*)^*\cong \cD^*\boxtimes\cD$ now is parallel with $\operatorname{ev}_\cD^\dagger$ allows us to impose (C5, C6 Definition 2.2.4, \cite{Douglas:2018})
\begin{equation}
    \operatorname{ev}_{\cD^*}=\operatorname{cev}_\cD^\dagger,\qquad \varphi_{\cD^*}=\varrho_\cD^\dagger,\label{adjdual}
\end{equation}
which states that adjunction intertwines the duality folds. Under these conditions, one can form the 2-morphism
\begin{equation*}
    \mathfrak{Dim}(\cD)= e_{\operatorname{cev}_{\cD^*}}\bullet\iota_{\operatorname{ev}_{\cD}} : 1_\cI\Rightarrow 1_\cI,
\end{equation*}
defining the 2-categorical dimension of $\cD$ (see Definition 2.3.8 of \cite{Douglas:2018}) assigned to a $\cD$-decorated closed 2-sphere. However, this 2-morphism $\mathfrak{Dim}(\cD)$ admits an action by $\operatorname{End}(\cD)$ according to \eqref{F-module}, hence we see that the value of this 2-sphere is in general \textit{not} invariant under taking equivalence classes of $\cD$. 
\end{rmk}

    \begin{tcolorbox}[breakable]
    \paragraph{Pairing conventions.} We pause here to give a brief comment regarding the pairing conventions. In the case of the ordinary rigid (pivotal) tensor 1-category $C$ (cf. \cite{etingof2016tensor,Bartlett:2016}), the folds define the duality datum $(c^\dagger,e_c,\iota_c)$ of an object $c\in C$. By pivotality, the pre-duality datum for $c^\dagger$ can be uniquely determined by fixing the value of the quantum dimension,
     \begin{equation*}
        \operatorname{dim}(c)^2 = (e_c\circ \iota_{c^\dagger})\otimes (\iota_c\circ e_{c^\dagger}),
    \end{equation*}
    such that $(c^\dagger,e_{c^\dagger}= \iota_c^\dagger,\iota_{c^\dagger}=e_c^\dagger)$. This is the planar-pivotal pairing convention we have used for the hom-categories of $\operatorname{2Rep}(\tilde{C};\tilde R)$ in \S \ref{adjsnake}. One of course wishes to adopt an analogous pairing convention for the object-level duality. However, as we have noted above, the 2-categorical dimension $\mathfrak{Dim}(\cD)$ cannot be computed naturally without the pivotality condition \eqref{adjdual}, and even if we do have pivotality, the issues mentioned in \textit{Remark \ref{warning}} makes it not clear if the pre-duality datum for $\cD^*$ enforced by \eqref{adjdual} is unique. The most we have access to are the "fold-on-fold 2-morphisms" defined above in \S \ref{2catdimension}; we will use them in \S \ref{hopflinks}.
\end{tcolorbox}



\section{Braiding and rigidity}\label{adj-dualbraid}
We now include the braiding into our analysis in this section, which shall lead to a natural notion of a rigid braided 2-category. Moreover, we shall define the ribbon twist from the braiding and relate this construction to the over-/under-twists defined in \cite{Douglas:2018}. This will then guide us to propose a notion of a \textit{ribbon tensor 2-category}, which are known to play major roles in both the physics of 4d TQFTs and the mathematics of 2-tangles \cite{BAEZ2003705}.

\subsection{Braiding and adjunctions; the second Reidemeister move}\label{adjbraid}
Recall that adjunctions on $\operatorname{2Rep}(\tilde C;\tilde R)$ is induced by an involution $\tilde C\rightarrow\tilde C^\text{op}$ on $\tilde C$. This involution is contravariant, and hence only effects the cocomposition in $\tilde C$. Denoting the image of $\tilde C$ under this involution by  $\overline{C}=(\tilde C)^\text{op}$, it is clear that $\overline{C}$ is equipped with a cobraiding given by 
\begin{equation*}
    \overline{\tilde R} \cong -^\dagger\circ \tilde R\circ-^\dagger.
\end{equation*}
We shall write $\overline{R}$ also for the associated 2-$R$-matrix induced by the cobraiding via \textit{Remark \ref{quantumcobraiding}}.

We are now in the position to investigate the planar-unitarity of the braiding. Let $\cD,\cA\in\operatorname{2Rep}(\tilde C;\tilde R)$. We define the adjoint of the braiding functor $c_{\cD,\cA}^\dagger: \cA\boxtimes\cD\Rightarrow \cD\boxtimes\cA$ to be the action by $\overline{R}$,
\begin{gather*}
    c_{\cD,\cA}^\dagger =\text{flip}\circ(\rho_\cD\otimes\rho_\cA)(\overline{R}),\\
    c_{F,\cA}^\dagger = \text{flip}\circ (\rho_\cD\otimes\rho_\cA)(\overline{R}_1^l),\qquad c_{\cD, F'}^\dagger = \text{flip}\circ(\rho_\cD\otimes\rho_\cA)(\overline{R}_1^r),
\end{gather*}
where $F:\cD\rightarrow\cD',F':\cA\rightarrow\cA'$.  

Now recall that $\tilde S_v=-^\dagger$ implements an orientation reversal. This geometric understanding implies the existence of invertible 2-morphisms 
\begin{equation}
    e_{c_{\cD,\cA}}: c_{\cD,\cA}^\dagger \circ c_{\cD,\cA}\Rightarrow 1_{\cD\boxtimes\cA},\qquad \iota_{c_{\cD,\cA}}:1_{\cA\boxtimes\cD}\Rightarrow c_{\cD,\cA}\circ c_{\cD,\cA}^\dagger\label{secondreide}
\end{equation}
for each $\cD,\cA\in\operatorname{2Rep}(\tilde C;\tilde R)$. The structure of planar-unitarity for the braiding functos then identify \eqref{secondreide} as the adjunction-folds on the braiding, which witness the \textbf{second Reidemeister move} (see fig. 26 in \cite{BAEZ2003705}, and fig. 55 (b) in \cite{Barrett_2024}). By the fact that the adjunction is involutive $(F^\dagger)^\dagger\cong F$ on $\tilde C$-module functors $F$, the (left-)adjunction-mates of these folds read
\begin{equation*}
    e_{c_{\cD,\cA}}^\dagger: 1_{\cD\boxtimes\cA}\Rightarrow c_{\cD,\cA}^\dagger \circ c_{\cD,\cA},\qquad \iota_{c_{\cD,\cA}}^\dagger:1_{\cA\boxtimes\cD}\Rightarrow c_{\cD,\cA}\circ c_{\cD,\cA}^\dagger.
\end{equation*}
The definition of planar-unitarity then implies the unitarity of these adjunction-fold 2-morphisms. Similarly, the hexagonators $\Omega,\Omega^\dagger$ are also unitary.

\begin{rmk}\label{syllepsis}
    Note crucially that, despite the functors $c_{\cD,\cA}^\dagger,c_{\cA,\cD}$ are parallel, they are {\it not} 2-isomorphic; we in general cannot even find a 2-morphism between them. However, if the 2-$R$-matrix $\tilde R$ on $\tilde C$ is \textbf{Hermitian} --- that is, if it satisfies $\overline{R}=\tilde R^T$ where $\tilde R^t= \sigma\tilde R$ denotes the transposed 2-$R$-matrix --- then we can indeed find a(n invertible) $\tilde C$-module natural transformation $c_{\cA,\cD}\Rightarrow c_{\cD,\cA}^\dagger$. Composing this with the fold $e_{c_{\cD,\cA}}$ then gives a 2-morphism which trivializes the \textit{full braiding}
    \begin{equation*}
            C_{\cD,\cA} =c_{\cA,\cD}\circ c_{\cD,\cA}\Rightarrow c_{\cD,\cA}^\dagger\circ c_{\cD,\cA}\Rightarrow 1_{\cD\boxtimes\cA}.
    \end{equation*}
    This makes $\operatorname{2Rep}(\tilde C;\tilde R)$ sylleptic! Indeed, the syllepsis coherences follow directly from underlyinh $\tilde C$-module structures. In other words, a Hermitian 2-$R$-matrix gives rise to a sylleptic representation 2-category.
\end{rmk}

\subsubsection{Braiding on the adjoints}
Of course, for each $\tilde C$-module functor $F:\cD\rightarrow\cD'$ and $F':\cA\rightarrow\cA'$ we have, by naturalty, the following conditions
\begin{equation*}
    e_{c_{\cD',\cA}}= e_{c_{\cD,\cA}}\bullet (c_{F,\cA}^\dagger\circ c_{F,\cA}),\qquad \iota_{c_{\cD,\cA'}}= (c_{\cD,F'}\circ c^\dagger_{\cD,F'})\bullet \iota_{c_{\cD,\cA}}.
\end{equation*}
However, together with the folds $e_F\boxtimes \cA: F^\dagger \boxtimes\cA \circ F\boxtimes\cA\rightarrow 1_{\cD\boxtimes\cA}$, we can form the following diagram
\[\begin{tikzcd}
	{\cD\boxtimes\cA} && {\cA\boxtimes\cD} && {\cD\boxtimes\cA} \\
	{\cD'\boxtimes\cA} && {\cA\boxtimes\cD'} && {\cD'\boxtimes \cA} \\
	{\cD\boxtimes\cA} && {\cA\boxtimes\cD} && {\cD\boxtimes\cA}
	\arrow[from=1-1, to=1-3]
	\arrow[""{name=0, anchor=center, inner sep=0}, "{1_{\cD\boxtimes\cA}}", curve={height=-25pt}, from=1-1, to=1-5]
	\arrow[from=1-1, to=2-1]
	\arrow[""{name=1, anchor=center, inner sep=0}, "{1_{\cD\boxtimes\cA}}"', curve={height=75pt}, from=1-1, to=3-1]
	\arrow[from=1-3, to=1-5]
	\arrow[from=1-3, to=2-3]
	\arrow[from=1-5, to=2-5]
	\arrow[""{name=2, anchor=center, inner sep=0}, "{1_{\cD\boxtimes\cA}}", curve={height=-75pt}, from=1-5, to=3-5]
	\arrow["{c_{F,\cA}}", shorten <=10pt, shorten >=10pt, Rightarrow, from=2-1, to=1-3]
	\arrow[from=2-1, to=2-3]
	\arrow[from=2-1, to=3-1]
	\arrow["{c_{F,\cA}^\dagger}"', shorten <=10pt, shorten >=10pt, Rightarrow, from=2-3, to=1-5]
	\arrow[from=2-3, to=2-5]
	\arrow[from=2-3, to=3-3]
	\arrow[from=2-5, to=3-5]
	\arrow["{c_{F^\dagger,\cA}}", shorten <=10pt, shorten >=10pt, Rightarrow, from=3-1, to=2-3]
	\arrow[from=3-1, to=3-3]
	\arrow[""{name=3, anchor=center, inner sep=0}, "{1_{\cD\boxtimes\cA}}"', curve={height=25pt}, from=3-1, to=3-5]
	\arrow["{c_{F^\dagger,\cA}^\dagger}"', Rightarrow, from=3-3, to=2-5]
	\arrow[from=3-3, to=3-5]
	\arrow["{e_{c_{\cD,\cA}}}", shorten >=2pt, Rightarrow, from=1-3, to=0]
	\arrow["{e_F\boxtimes\cA}", shorten >=1pt, Rightarrow, from=2-1, to=1]
	\arrow["{e_F\boxtimes \cA}", shorten >=1pt, Rightarrow, from=2-5, to=2]
	\arrow["{e_{c_{\cD,\cA}}}"', shorten >=2pt, Rightarrow, from=3-3, to=3]
\end{tikzcd}\]
which expresses the adjoints of the mixed braiding 2-morphisms $c_{F,\cA}$; similarly for $c_{\cD,F'}$. 

To conclude that the 2-morphism $c_{F,\cA}$ is unitary, we must now exhibit a coherence relation between its two adjoints. This is done by super-imposing the top-right square in the above diagram by an orientation reversal of the bottom-left square, whence we achieve of the following commutative diagram of 2-morphisms
\[\begin{tikzcd}[ampersand replacement=\&]
	{\cA\boxtimes\cD} \&\&\&\& {\cD\boxtimes\cA} \\
	\& {\cA\boxtimes\cD} \&\& {\cD\boxtimes\cA} \\
	\\
	\& {\cA\boxtimes\cD’} \&\& {\cD’\boxtimes \cA} \\
	{\cA\boxtimes\cD’} \&\&\&\& {\cD’\boxtimes\cA}
	\arrow[from=1-5, to=1-1]
	\arrow[""{name=0, anchor=center, inner sep=0}, "{1_{\cA\boxtimes\cD}}"', from=2-2, to=1-1]
	\arrow[from=2-2, to=2-4]
	\arrow[from=2-2, to=4-2]
	\arrow["{c_{\cA,F}^\dagger}", shift left=3, shorten <=13pt, shorten >=13pt, Rightarrow, from=2-2, to=4-4]
	\arrow[""{name=1, anchor=center, inner sep=0}, "{1_{\cD\boxtimes\cA}}"', from=2-4, to=1-5]
	\arrow[from=2-4, to=4-4]
	\arrow[from=4-2, to=4-4]
	\arrow[""{name=2, anchor=center, inner sep=0}, "{1_{\cA\boxtimes\cD'}}"', from=4-2, to=5-1]
	\arrow[""{name=3, anchor=center, inner sep=0}, "{1_{\cD'\boxtimes\cA}}"', from=4-4, to=5-5]
	\arrow[from=5-1, to=1-1]
	\arrow[from=5-5, to=1-5]
	\arrow[from=5-5, to=5-1]
	\arrow["{e_{c_{\cD,\cA}}}", shorten <=42pt, shorten >=42pt, Rightarrow, from=0, to=1]
	\arrow["{e_{F\boxtimes A}}", shorten <=19pt, shorten >=19pt, Rightarrow, from=1, to=3]
	\arrow["{\iota_{\cA\boxtimes F}}", shorten <=19pt, shorten >=19pt, Rightarrow, from=2, to=0]
	\arrow["{\iota_{c_{\cD',\cA}}}", shorten <=42pt, shorten >=42pt, Rightarrow, from=3, to=2]
\end{tikzcd}=
\begin{tikzcd}[ampersand replacement=\&]
	{\cA\boxtimes\cD} \&\& {\cD\boxtimes\cA} \\
	\\
	{\cA\boxtimes\cD’} \&\& {\cD’\boxtimes\cA}
	\arrow[from=1-3, to=1-1]
	\arrow[from=3-1, to=1-1]
	\arrow["{c_{F^\dagger,\cA}}"', shorten <=13pt, shorten >=13pt, Rightarrow, from=3-3, to=1-1]
	\arrow[from=3-3, to=1-3]
	\arrow[from=3-3, to=3-1]
\end{tikzcd}\]

This can be written concisely using \eqref{adjfoldcompose} as
\begin{equation}
    e_{F\boxtimes\cA \circ c_{\cD\boxtimes\cA}^\dagger} \bullet c_{F,\cA}^\dagger\bullet \iota_{\cA\boxtimes F^\dagger \circ c_{\cD'\boxtimes\cA}} = \iota_{\cA\boxtimes F^\dagger \circ c_{\cD'\boxtimes\cA}}^\dagger  \bullet c_{F^\dagger,\cA}\bullet e_{F\boxtimes\cA \circ c_{\cD\boxtimes\cA}^\dagger}^\dagger.\label{adjfoldbraid}
\end{equation}
From Definition 12 of \cite{BAEZ2003705}, we thus have the following.
\begin{theorem}\label{withduals}
    The rigid dagger tensor 2-category $\operatorname{2Rep}(\tilde C;\tilde R)$, equipped with planar-unitary braiding $c,$ is a "braided monoidal 2-category with duals".
\end{theorem}
\noindent Note crucially that the notion of "duals" in this theorem refers to the adjunction $-^\dagger$, not actual object-level duality. This is because of the way the word "dual" is used in \cite{BAEZ2003705} is different from how we are using it here. 

It is thus not possible to leverage the result in \cite{Baez1996HigherDimensionalAI} and use adjunctions to construct a ribbon balancing/twist; we must use the object-level duality for this. This will be the subject of \S \ref{ribbon}.

\subsubsection{Adjunctions and higher-dagger structures}\label{higherdagger}
    We pause here to elaborate more on the above comment. The object-level duality described in \cite{BAEZ2003705} --- specifically for the "2-category of 2-tangles" --- can in a sense be understood as a conflation of both the words "duals" and "adjoints" used in this paper and \cite{Douglas:2018}: the functor adjoints coincide with their duality-mates. This can be attributed to the interpretation in \cite{BAEZ2003705} that the "generator of the 2-tangles" exist in the \textit{unframed} universe. We now know how to refine this notion; indeed, it was conjectured in \cite{ferrer2024daggerncategories} (Conjecture 5.3) that the dagger-autoequivalences $\operatorname{Aut}^\dagger(\mathsf{AdjCat}_{(\infty,n)})$ of $(\infty,n)$-categories with all adjoints is equivalent to the piecewise linear group $PL(n)$, away from $n=4$. The group $PL(n)$ certainly does not just consist of one single duality $\bbZ_2$, unless $n=1$. This issue matters, as it underpins the cobordism hypothesis and the validity of graphical calculus. 
    
    The attentive reader may notice that the double delooping $B^2\operatorname{2Rep}(\tilde C;\tilde R)$ has precisely $n=4$, for which the statement $PL(4)\simeq \operatorname{Aut}_{\text{sym}\otimes}(\mathsf{Bord}_f^\text{fr})$ is equivalent to the open 4d PL Schoenflies conjecture \cite{lurie2008classification}. It thus seems that we must solve this open problem in geometric topology before being able to do graphical calculus with $\operatorname{2Rep}(\tilde C;\tilde R)$. However, this 2-category only has two adjoints, similar to the $\mathsf{Gray}$-categories with duals \cite{Barrett_2024} and the defect tricategories \cite{Carqueville:2016kdq} studied recently. As such, doing graphical calculus with $\operatorname{2Rep}(\tilde C;\tilde R)$ may not be as sophisticated as the PL Schoenflies conjecture.

\subsection{Braiding and duality}\label{dualitybraid}
We begin our analysis here in a slightly different way as in the previous section. For each $\cD,\cA\in\operatorname{2Rep}(\tilde C;\tilde R)$, the condition \eqref{2ndreidemeister} can be seen to  imply the existence of $\tilde C$-module 2-morphisms
\begin{equation*}
    (\cD^*\boxtimes c_{\cA,\cD})\circ(c_{\cA,\cD^*}\boxtimes\cD)\Rightarrow 1_\cA,\qquad 1_\cA \Rightarrow (\cD\boxtimes c_{\cA,\cD^*})\circ(c_{\cA,\cD}\boxtimes\cD^*)
\end{equation*}
trivializing the subsequent braiding of $\cA$ against $\cD$ with its dual. These can be seen to arise from the following diagram,
\[\begin{tikzcd}
	{(\cA\cD^*)\cD} & {(\cD^*\cA)\cD} & {\cD^*(\cA\cD)} & {\cD^*(\cD\cA)} \\
	{\cA(\cD^*\cD)} &&& {(\cD^*\cD)\cA} \\
	& \cA & \cA \\
	{\cA(\cD\cD^*)} &&& {(\cD\cD^*)\cA} \\
	{(\cA\cD)\cD^*} & {(\cD\cA)\cD^*} & {\cD(\cA\cD^*)} & {\cD(\cD^*\cA)}
	\arrow["{c_{\cA,\cD^*}\cD}", from=1-1, to=1-2]
	\arrow[from=1-1, to=2-1]
	\arrow["{\Omega_{\cA|\cD^*\cD}}", shorten <=29pt, shorten >=29pt, Rightarrow, from=1-1, to=2-4]
	\arrow[from=1-2, to=1-3]
	\arrow["{\cD^*c_{\cA,\cD}}", from=1-3, to=1-4]
	\arrow[""{name=0, anchor=center, inner sep=0}, "{c_{\cA,\cD^*\boxtimes\cD}}", from=2-1, to=2-4]
	\arrow[from=2-1, to=3-2]
	\arrow[from=2-4, to=1-4]
	\arrow[from=2-4, to=3-3]
	\arrow[""{name=1, anchor=center, inner sep=0}, "{1_\cA}"{description}, from=3-2, to=3-3]
	\arrow[from=3-2, to=4-1]
	\arrow[from=3-3, to=4-4]
	\arrow[""{name=2, anchor=center, inner sep=0}, "{c_{\cA,\cD\boxtimes\cD^*}}"', from=4-1, to=4-4]
	\arrow["{\Omega_{\cA|\cD\cD^*}}"', shorten <=29pt, shorten >=29pt, Rightarrow, from=4-1, to=5-4]
	\arrow[from=5-1, to=4-1]
	\arrow["{c_{\cA\boxtimes\cD}\cD^*}"', from=5-1, to=5-2]
	\arrow[from=5-2, to=5-3]
	\arrow["{\cD c_{\cA,\cD^*}}"', from=5-3, to=5-4]
	\arrow[from=5-4, to=4-4]
	\arrow["{c_{\cA,\operatorname{ev}_\cD}}", shorten <=4pt, shorten >=4pt, Rightarrow, from=0, to=1]
	\arrow["{c_{\cA,\operatorname{cev}_\cD}}"', shorten <=4pt, shorten >=4pt, Rightarrow, from=1, to=2]
\end{tikzcd}\]
in which these 2-morphisms arise from composing the unitary hexagonators $\Omega$ with the 2-morphism 
\begin{equation*}
    c_{\operatorname{ev}_\cD,\cA}: c_{\cD^*\boxtimes\cD,\cA}\Rightarrow 1_\cA,\qquad c_{\operatorname{cev}_\cD,\cA}: 1_\cA\rightarrow c_{\cD\boxtimes\cD^*,\cA}.
\end{equation*}
Together with the braid-exchange 2-morphism $c_{\cA,c_{\cD^*,\cD}}$, this diagram can be seen as a trivialization of the 2-morphism $\Omega_{c_{\cA,\cD^*}|\cD}$ for each $\cA,\cD$. In other words, the presence of rigid duals in a braided monoidal 2-category allows us to trivialize the unitary hexagonators $\Omega_{-|--}$, as well as the 2-morphisms $\Omega_{c_{-,-}|-}$, whenever any two of the arguments in them are mutually dual.

We shall demonstrate in \S \ref{writhe} the compatibility of the braiding against both the rigid duality and the involutive $\dagger$-adjunctions, and the square in the centre of the above diagram will play an important role. To set up this exposition, we will first describe an important class of braiding structures.

\subsubsection{Braiding structures on duals}\label{braidduals}
For each $\cD\in\operatorname{2Rep}(\tilde C;\tilde R)$ and endofunctor $F:\cD\rightarrow\cD$, consider the braiding map $c_{\cD,\cD^*}:\cD\boxtimes\cD^*\rightarrow\cD^*\boxtimes\cD$, given by the structures
\begin{gather*}
    c_{\cD,\cD^*} = \text{flip}\circ(\rho \boxtimes\rho^*)(\tilde R_0) = \text{flip} \circ(\rho \boxtimes\rho)((1\otimes \tilde S_0)\tilde R_0),\\
    c_{F,\cD^*} = \text{flip}\circ (\rho \boxtimes\rho)((1\otimes \tilde S_1)\tilde R_1^l),\qquad c_{\cD, F^*} = \text{flip}\circ(\rho \boxtimes\rho)((1\otimes \tilde S_0)\tilde R_1^r).
\end{gather*}
Through the module associator $(\rho(-)\otimes\rho(-))(-) \Rightarrow \rho(-\cdot -)(-)$, these braiding maps are determined by the following quantities
\begin{equation*}
    \nu= (-\cdot -) ( \tilde S\otimes1) \tilde R^T \equiv \nu^l\oplus \nu^r,
\end{equation*}
in $\tilde C$, where the horizontal transpose 2-$R$-matrix $\tilde R^T$ is intertwined by the flip map,
\begin{equation*}
    ((\rho_2\otimes\rho_1)\tilde R^T_\mathrm{h}) \circ \text{flip} = \text{flip}\circ((\rho_1\otimes\rho_2)\tilde R).
\end{equation*}

The naturality of the cobraiding transformation $\tilde R$ implies the following {\it nudging equations}
\begin{equation}
    \nu^l \cdot \nu^r = \nu^r\cdot\nu^l.\label{nudging}
\end{equation}
Notice in the setting of weak 2-gauge theory (see {\it Remark \ref{descendant}}), the associator $\tau$ representing the Postnikov class of $\mathbb{G}$ will contribute to \eqref{nudging}. Similar constructions can be made for the braiding $c_{\cD^*,\cD}:\cD^*\boxtimes \cD\rightarrow\cD\boxtimes\cD^*$ with the dual on the other side, which are associated to the following elements
\begin{equation*}
        \mu= (-\cdot -) ( 1\otimes\tilde S) \tilde R^T
\end{equation*}
satisfying its own nudging equations.  

\begin{rmk}\label{radf}
    This Hopf category object $\nu = (-\cdot-)(\tilde S\otimes 1)\tilde R^T$ is a categorification of a key piece of structural data for Hopf algebras. It is known \cite{RADFORD19921} that for finite-dimensional quasitriangular Hopf algebras $H$, the analogous Hopf algebra element $\nu=m(S\otimes 1)R^T\in H$ is invertible and represents the antipode-squared as an inner automorphism,
    \begin{equation*}
        S^2(a) = \nu a \nu^{-1},\qquad \forall~ a\in H.
    \end{equation*}
    As such, the centrality $\nu\in Z(H)$ implies that $S^2=\id$. This element also plays a central role in the proof of the Radford $S^4$-formula \cite{Radford:1976,Delvaux2006ANO}. It would therefore be reasonable to posit that, if the Hopf category object $\nu$ lift to the Drinfel'd centre $Z_1(\tilde C)$,\footnote{Note $\tilde C$ is only just cobraided, hence it makes sense to talk about its Drinfel'd centre.} then $\tilde C$ is a "cospherical Hopf category" and $\operatorname{2Rep}(\tilde C;\tilde R)$ becomes pivotal in the sense of \cite{Douglas:2018}.
\end{rmk}


\subsubsection{Writhing}\label{writhe}
A direct computation with the 2-$R$-matrix condition $(\tilde \epsilon\otimes1)\tilde R=\id = (1\otimes\tilde\epsilon)\tilde R$ shows that the composite functor $c_{\cD,\cD^*}\circ \operatorname{cev}_\cD: \cI\rightarrow \cD^*\boxtimes\cD$ has the same (trivial) $\tilde\cC$-module structure as $\operatorname{ev}_\cD^\dagger$, and similarly for $c_{\cD^*,\cD}$. As such, we can find 2-morphisms
\begin{equation*}
    K_\cD:c_{\cD,\cD^*}\circ\operatorname{cev}_\cD\Rightarrow \operatorname{ev}_\cD^\dagger,\qquad \bar K_\cD:  \operatorname{cev}_\cD\Rightarrow c_{\cD^*,\cD} \circ\operatorname{ev}_\cD^\dagger
\end{equation*}
called the \textbf{over-/under-writhings}; see fig. 31 of \cite{BAEZ2003705}. Note ibid. used $\operatorname{cev}_{\cD^*}$ instead of $\operatorname{ev}_\cD^\dagger$ in the writhing, hence we recover their definition provided the dual is involutive and the condition \eqref{adjdual} holds.

\begin{rmk}\label{firstreide}
Geometrically, the over-/under-writhings implement a rotation the top portion of a fold. This was interpreted as a null-homotopy witnessing the first Reidemeister move in \cite{BAEZ2003705}, but we will not take this perspective here. We shall return to this issue in \S \ref{reidemeisterI}.
\end{rmk}


We now prove how the second Reidemeister moves can be related to two applications of the writhing.
\begin{proposition}\label{writheadj}
    Suppose the writhings are invertible and the objects $\nu,\mu$ satisfy 
    \begin{equation}
        \bar\nu = \mu^T = (-\otimes -) (\tilde S\boxtimes1) \tilde R,\qquad \bar\mu = \nu^T\label{quasiherm}
    \end{equation}
    under orientation reversal, then we have
    \begin{equation*}
        K_\cD\bullet(c_{\cD,\cD^*}\circ\bar K_\cD^{-1}) = e_{c_{\cD^*,\cD}}\circ \operatorname{ev}_\cD^\dagger,\qquad (c_{\cD^*,\cD}\circ K_\cD^{-1})\bullet \bar K_\cD = \iota_{c_{\cD,\cD^*}}\circ\operatorname{cev}_\cD.
    \end{equation*}
\end{proposition}
\noindent We call the 2-$R$-matrix $\tilde{R}$ satisfying \eqref{quasiherm} \textbf{quasi-Hermitian}.\footnote{It is clear that a Hermitian $\tilde R$ (see \textit{Remark \ref{syllepsis}}) is quasi-Hermitian, by applying an antipode $1\boxtimes \tilde S$ and then a contraction $(-\otimes-)$ to the definition $\bar R= \tilde R^T$.}
\begin{proof}
    We begin by post-composing $\bar K_\cD^{-1}$ with $c_{\cD,\cD^*}$ to achieve the 2-morphism
    \begin{equation}
        k_\cD: C_{\cD^*,\cD}\circ\operatorname{ev}_\cD^\dagger=c_{\cD,\cD^*}\circ c_{\cD^*,\cD}\circ\operatorname{ev}_\cD^\dagger \xRightarrow{c_{\cD,\cD^*}\circ\bar K_\cD^{-1}} c_{\cD,\cD^*}\circ \operatorname{cev}_\cD \xRightarrow{K_\cD}\operatorname{ev}_\cD^\dagger.\label{eq1}
    \end{equation}
If \eqref{quasiherm} holds, then the same argument as from {\it Remark \ref{syllepsis}} implies that there is a(n invertible) 2-morphism $c_{\cD^*,\cD}\cong c_{\cD,\cD^*}^\dagger$. This allows us to form the commutative triangle
\begin{equation}\begin{tikzcd}
	{c_{\cD,\cD^*}\circ\operatorname{cev}_\cD} & {\operatorname{ev}_\cD^\dagger} \\
	{c_{\cD^*,\cD}^\dagger\circ c_{\cD^*,\cD}\circ\operatorname{ev}_\cD^\dagger}
	\arrow["{K_\cD}", Rightarrow, from=1-1, to=1-2]
	\arrow["{c_{\cD,\cD^*}\circ \bar K_\cD}"', Rightarrow, from=1-1, to=2-1]
	\arrow["{e_{c_{\cD^*,\cD}}\circ\operatorname{ev}_\cD^\dagger}"', Rightarrow, from=2-1, to=1-2]
\end{tikzcd}\label{dualwrithe}\end{equation}
which states that the composition in \eqref{eq1} is nothing but $e_{c_{\cD^*,\cD}}\circ \operatorname{ev}_\cD^\dagger$. A completely analogous argument holds for the composite
    \begin{equation}
        \bar k_\cD: \operatorname{cev}_\cD\xRightarrow{\bar K_\cD} c_{\cD^*,\cD}\circ\operatorname{ev}_\cD^\dagger\xRightarrow{c_{\cD^*,\cD}\circ K_\cD^{-1}}c_{\cD^*,\cD}\circ c_{\cD,\cD^*}\circ \operatorname{cev}_\cD = C_{\cD,\cD^*}\circ\operatorname{cev}_\cD.\label{eq2}
    \end{equation}
Notice the commutative triangle \eqref{dualwrithe} is nothing but Lemma 16 in \cite{BAEZ2003705}. 



    
\end{proof}
\noindent Throughout the following, we will often assume that the 2-$R$-matrix $\tilde R$ is quasi-Hermitian. 

\begin{rmk}\label{quasihermcondition}
    Let $\tilde{C}$ denote a cobraided Hopf category, and $\tilde{C}^\text{m-op}$ its monoidal opposite. Its monoidal product is given by $(-\cdot-)^T = (-\cdot-)\sigma$, where $\sigma: \tilde{C}\times\tilde{C}\rightarrow\tilde{C}\times\tilde{C}$ is a swapping of factors. As such, quasi-Hermiticity \eqref{quasiherm} can be understood as a condition relating the orientation reversal of the object $\nu\in\tilde{C}$ with the opposite one $\mu^T\in\tilde{C}^\text{m-op}$. If a monoidal natural transformation $(-\cdot-)\Rightarrow (-\cdot-)^T$ exists --- that is to say, if $\tilde{C}$ were \textit{braided} --- then the quasi-Hermitian condition appears as a coherence condition between the cobraiding and the braiding. However, there is nothing in the underlying 4d gauge theory which indicates $\tilde C$ should have a braiding.
\end{rmk}

\medskip

For posterity, we define a few structural 2-morphisms that will play an important role later. Composing the writhes $K_\cD,\bar K_\cD$ respectively with the folds $\operatorname{cev}_\cD^\dagger,\operatorname{ev}_\cD$ yield the 2-morphisms
\begin{align*}
    & K_\cD':c_{\cD,\cD^*}\xRightarrow{c_{\cD,\cD^*}\circ \iota_{\operatorname{cev}_\cD}}c_{\cD,\cD^*}\circ (\operatorname{cev}_\cD\circ\operatorname{cev}_\cD^\dagger)\xRightarrow{K_\cD\circ\operatorname{cev}_\cD^\dagger}\operatorname{ev}_\cD^\dagger\circ\operatorname{cev}_\cD^\dagger,\\
    & \bar K_\cD': \operatorname{cev}_\cD\circ\operatorname{ev}_\cD \xRightarrow{\bar K_\cD\circ \operatorname{ev}_\cD} c_{\cD^*,\cD}\circ \operatorname{ev}_\cD^\dagger\circ\operatorname{ev}_\cD \xRightarrow{c_{\cD^*,\cD}\circ e_{\operatorname{ev}_\cD}} c_{\cD^*,\cD}.
\end{align*}
On the other hand, consider the so-called "double point arc crossing a fold line" 2-morphisms \cite{CARTER19971}
\begin{align*}
    H_{\cD,\cA}&= \Omega_{\cD|\cD^*\cA}\bullet c_{\operatorname{cev}_\cD^\dagger,\cA}\bullet \cD\boxtimes \iota_{c_{\cD^*,\cA}}\\
    &\qquad \qquad : (\operatorname{cev}_\cD^\dagger\boxtimes\cA)\circ  (\cD\boxtimes c_{\cD^*,\cA}^\dagger)\Rightarrow (\cA\boxtimes\operatorname{cev}_\cD^\dagger)\circ (c_{\cA,\cD}\boxtimes\cD^*), \\ 
    G_{\cD,\cA} &= \Omega_{\cD^*|\cD\cA}\bullet c_{\operatorname{ev}_\cD,\cA} \bullet \cD^*\boxtimes\iota_{c_{\cD,\cA}}\\
    &\qquad \qquad: (\operatorname{ev}_\cD\boxtimes\cA)\circ (\cD^*\boxtimes c_{\cD,\cA}^\dagger)\Rightarrow (\cA\boxtimes\operatorname{ev}_\cD)\circ(c_{\cD^*,\cA}\boxtimes\cD),
\end{align*}
which we shall shorten to \textbf{fold-crossings}; see fig. 32 of \cite{BAEZ2003705}. There are also "dual" versions of these 2-morphisms, denoted by $\bar H,\bar G$, in which the folds are replaced by their appropriate barred versions.

\subsection{Rigid dagger structures and the writhing}
We now study the compatibility of the writhing 2-morphisms with the rigid duality structure. Let us begin with the fold-crossing maps defined above.  Consider the following 2-morphisms
\[(K\bullet H)_{\cD,\cA}=\begin{tikzcd}
	{(\cD\cD^*)\cA} && {\cD\cA\cD^*} && {\cA(\cD\cD^*)} \\
	\\
	&& \cA \\
	\\
	{(\cD^*\cD)\cA} &&&& {\cA(\cD^*\cD)}
	\arrow[""{name=0, anchor=center, inner sep=0}, curve={height=-6pt}, from=1-1, to=1-3]
	\arrow[""{name=1, anchor=center, inner sep=0}, shift left, curve={height=24pt}, from=1-1, to=1-5]
	\arrow[from=1-1, to=3-3]
	\arrow[""{name=2, anchor=center, inner sep=0}, from=1-1, to=5-1]
	\arrow[""{name=3, anchor=center, inner sep=0}, curve={height=-6pt}, from=1-3, to=1-1]
	\arrow[from=1-3, to=1-5]
	\arrow[from=1-5, to=3-3]
	\arrow[""{name=4, anchor=center, inner sep=0}, from=1-5, to=5-5]
	\arrow[""{name=5, anchor=center, inner sep=0}, from=3-3, to=5-1]
	\arrow[from=3-3, to=5-5]
	\arrow[from=5-1, to=5-5]
	\arrow["\dashv"{anchor=center, rotate=-90}, draw=none, from=0, to=3]
	\arrow["{K_\cD’\cA}", shorten <=17pt, shorten >=17pt, Rightarrow, from=2, to=3-3]
	\arrow["{c_{\operatorname{cev}_\cD^\dagger,\cA}}"{pos=0.6}, shorten <=5pt, shorten >=5pt, Rightarrow, from=1, to=3-3]
	\arrow["{\Omega_{\cD|\cD^*\cA}}", shorten <=1pt, shorten >=2pt, Rightarrow, from=1-3, to=1]
	\arrow["{\cA K_\cD’}"', shorten <=17pt, shorten >=17pt, Rightarrow, from=3-3, to=4]
	\arrow["{c_{\operatorname{ev}_\cD^\dagger,\cA}}", shorten <=23pt, shorten >=23pt, Rightarrow, from=5, to=5-5]
\end{tikzcd}\]
\[(\bar K\bullet G)_{\cD,\cA}=\begin{tikzcd}
	{(\cD^*\cD)\cA} &&&& {(\cD\cD^*)\cA} \\
	\\
	{\cD^*\cA\cD} &&& \cA \\
	\\
	{\cA(\cD^*\cD)} &&&& {\cA(\cD\cD^*)}
	\arrow[from=1-1, to=1-5]
	\arrow[""{name=0, anchor=center, inner sep=0}, shift right=3, curve={height=6pt}, from=1-1, to=3-1]
	\arrow[from=1-1, to=3-4]
	\arrow[""{name=1, anchor=center, inner sep=0}, shift left=2, curve={height=-30pt}, from=1-1, to=5-1]
	\arrow[""{name=2, anchor=center, inner sep=0}, from=1-5, to=5-5]
	\arrow[""{name=3, anchor=center, inner sep=0}, curve={height=6pt}, from=3-1, to=1-1]
	\arrow[from=3-1, to=5-1]
	\arrow[""{name=4, anchor=center, inner sep=0}, from=3-4, to=1-5]
	\arrow[from=3-4, to=5-5]
	\arrow[from=5-1, to=3-4]
	\arrow[""{name=5, anchor=center, inner sep=0}, from=5-1, to=5-5]
	\arrow["{\bar K’_\cD\cA}"', shorten <=22pt, shorten >=22pt, Rightarrow, from=1-1, to=4]
	\arrow["{c_{\operatorname{ev}_\cD,\cA}}"{pos=0.6}, shorten <=14pt, shorten >=14pt, Rightarrow, from=1, to=3-4]
	\arrow["\dashv"{anchor=center, rotate=-180}, draw=none, from=3, to=0]
	\arrow["{\Omega_{\cD^*|\cA\cD}}"{pos=0.3}, shorten <=7pt, shorten >=7pt, Rightarrow, from=3, to=1]
	\arrow["{\cA \bar K’_\cD}", shorten <=9pt, shorten >=9pt, Rightarrow, from=3-4, to=5]
	\arrow["{c_{\operatorname{cev}_\cD,\cA}}"', shorten <=4pt, shorten >=4pt, Rightarrow, from=3-4, to=2]
\end{tikzcd}\]
which expresses the crossing of over-/under-writhings over the folds. Similar constructions can be made for the adjoint writhe-crossings $(K\bullet  H^\dagger),(\bar K\bullet G^\dagger)$, in which braidings of the form $c_{\cD\boxtimes\cD^*,\cA},$ are replaced with their adjoints $c_{\cD\boxtimes\cD^*,\cA}^\dagger$.

\subsubsection{Rigid writhing conditions}\label{rigidwrithes}
We now exhibit the compatibility between the writhing and the rigid duality, under the assumption that $\tilde R$ is quasi-Hermitian (ie. \eqref{quasiherm} holds). These are expressed by the so-called \textit{rigid writhing conditions} \eqref{rigidwrithe1}, \eqref{rigidwrithe2}, which we shall now deduce.

Starting from quasi-Hermiticity \eqref{quasiherm}, we deduce a 2-isomorphism $c_{\cD^*,\cD}\cong c_{\cD,\cD^*}^\dagger$. Taking the writhe-crossing $(K\bullet H^\dagger)_{\cD,\cD^*}$ at $\cA=\cD^*$, we see that the cube corresponding to \eqref{adjfoldbraid} --- with $F=c_{\cD,\cD^*}$ --- allows us to stack $(K\bullet H^\dagger)_{\cD,\cD^*}$ with $(\bar H\bullet \bar K)_{\cD,\cD^*}$ by use of the associator equivalences:
\[\begin{tikzcd}
	& {(\cD\cD^*)\cD^*} && {\cD^*(\cD\cD^*)} & {(\cD^*\cD)\cD^*} && {(\cD\cD^*)\cD^*} \\
	{\cD^*} &&&&&&& {\cD^*} \\
	& {(\cD^*\cD)\cD^*} && {\cD^*(\cD^*\cD)} & {\cD^*(\cD^*\cD)} && {\cD^*(\cD\cD^*)}
	\arrow[from=1-2, to=1-4]
	\arrow[draw={rgb,255:red,255;green,51;blue,54}, curve={height=-30pt}, equals, from=1-2, to=1-7]
	\arrow[draw={rgb,255:red,255;green,51;blue,54}, from=1-2, to=2-1]
	\arrow[""{name=0, anchor=center, inner sep=0}, from=1-2, to=3-2]
	\arrow["\simeq", from=1-4, to=1-5]
	\arrow[from=1-4, to=2-1]
	\arrow[""{name=1, anchor=center, inner sep=0}, from=1-4, to=3-4]
	\arrow[from=1-5, to=1-7]
	\arrow[from=1-5, to=2-8]
	\arrow[""{name=2, anchor=center, inner sep=0}, from=1-5, to=3-5]
	\arrow[""{name=3, anchor=center, inner sep=0}, from=1-7, to=3-7]
	\arrow[from=2-1, to=3-2]
	\arrow[draw={rgb,255:red,54;green,51;blue,255}, from=2-1, to=3-4]
	\arrow[draw={rgb,255:red,255;green,51;blue,54}, from=2-8, to=1-7]
	\arrow[from=2-8, to=3-7]
	\arrow[from=3-2, to=3-4]
	\arrow["\simeq", curve={height=24pt}, from=3-2, to=3-7]
	\arrow[draw={rgb,255:red,54;green,51;blue,255}, equals, from=3-4, to=3-5]
	\arrow[draw={rgb,255:red,54;green,51;blue,255}, from=3-5, to=2-8]
	\arrow[from=3-5, to=3-7]
	\arrow["{(K\bullet H)_{\cD,\cA}}", shorten <=20pt, shorten >=20pt, Rightarrow, from=0, to=1]
	\arrow["{(\bar K\bullet G)_{\cD,\cA}}", shorten <=20pt, shorten >=20pt, Rightarrow, from=2, to=3]
\end{tikzcd}\]
where we have utilized a shorthand to express $(K\bullet H)_{\cD,\cA}$ and $(\bar K\bullet G)_{\cD,\cA}$ as square diagrams. Notice crucially that we have kept track of the functors that enter/exit the position $\cA=\cD^*$.

This is important, as we notice that the compositions of the functors  $\cD^*\rightarrow\cD^*$ marked in blue (and respectively those marked in red)  are given precisely respectively by the cusp ${(\operatorname{ev}_{\cD}\boxtimes\cD^*)\circ (\cD^*\boxtimes\operatorname{cev}_{\cD})}$ and its adjoint. We can then apply the snakerators/cusps \eqref{cusps} $${\color{blue}\varphi_\cD: 1_{\cD^*}\Rightarrow (\operatorname{ev}_{\cD}\boxtimes\cD^*)\circ (\cD^*\boxtimes\operatorname{cev}_{\cD})},\qquad {\color{red}\varphi_\cD^\dagger: (\cD^*\boxtimes\operatorname{cev}_{\cD}^\dagger)\circ (\operatorname{ev}_{\cD}^\dagger\boxtimes\cD^*) \Rightarrow 1_{\cD^*}},$$ colour-coded for clarity, to this diagram. 

The unitarity of the snakerator $\varphi_\cD$, which follows from the planar-unitarity of the rigid duality structure described in \S \ref{rigidagger}. We are then able to form the following (rather convoluted) diagram of 2-morphisms
\begin{equation}
    \begin{tikzcd}
	& {(\cD^*\cD)\cD^*} &&& {\cD^*(\cD^*\cD)} \\
	{\cD^*(\cD^*\cD)} &&&&& {(\cD^*\cD)\cD^*} \\
	&& {\cD^*} & {\cD^*} \\
	{\cD^*(\cD\cD^*)} &&&&& {(\cD\cD^*)\cD^*} \\
	& {(\cD\cD^*)\cD^*} &&& {\cD^*(\cD\cD^*)}
	\arrow[""{name=0, anchor=center, inner sep=0}, "\sim", Rightarrow, no head, from=1-2, to=2-6]
	\arrow["\sim"', Rightarrow, no head, from=1-5, to=2-1]
	\arrow[dotted, from=1-5, to=3-4]
	\arrow[from=1-5, to=5-5]
	\arrow[""{name=1, anchor=center, inner sep=0}, from=2-1, to=1-2]
	\arrow[""{name=2, anchor=center, inner sep=0}, from=2-6, to=1-5]
	\arrow[from=2-6, to=3-4]
	\arrow[from=2-6, to=4-6]
	\arrow[from=3-3, to=1-2]
	\arrow[dotted, from=3-3, to=2-1]
	\arrow[""{name=3, anchor=center, inner sep=0}, "{1_{\cD^*}}"', from=3-3, to=3-4]
	\arrow[from=3-4, to=4-6]
	\arrow[dotted, from=3-4, to=5-5]
	\arrow[from=4-1, to=2-1]
	\arrow[dotted, from=4-1, to=3-3]
	\arrow[""{name=4, anchor=center, inner sep=0}, from=4-1, to=5-2]
	\arrow[""{name=5, anchor=center, inner sep=0}, "\sim", Rightarrow, no head, from=4-1, to=5-5]
	\arrow[""{name=6, anchor=center, inner sep=0}, from=4-6, to=5-5]
	\arrow[from=5-2, to=1-2]
	\arrow[from=5-2, to=3-3]
	\arrow["\sim", Rightarrow, no head, from=5-2, to=4-6]
	\arrow[""{name=7, anchor=center, inner sep=0}, draw=none, from=1, to=4]
	\arrow[""{name=8, anchor=center, inner sep=0}, draw=none, from=2, to=6]
	\arrow["{\varphi_\cD}"', shorten <=9pt, shorten >=9pt, Rightarrow, from=3, to=0]
	\arrow["{\varphi_\cD^\dagger}"', shorten <=9pt, shorten >=9pt, Rightarrow, from=5, to=3]
	\arrow["{(\bar K\bullet G)_{\cD,\cD^*}}", shift right, shorten <=7pt, shorten >=15pt, Rightarrow, from=8, to=3-4]
	\arrow["{(K\bullet  H^\dagger)_{\cD,\cD^*}}"', shorten <=15pt, shorten >=15pt, Rightarrow, from=3-3, to=7]
\end{tikzcd}.\label{rigidwrithe1}
\end{equation}
This diagram comes with another version constructed from $(\bar K\bullet \bar H^\dagger)_{\cD,\cD}$. Stacking it with $(K\bullet H)_{\cD,\cD}$ and the unitarity of the snakerator $\varrho_\cD$ gives
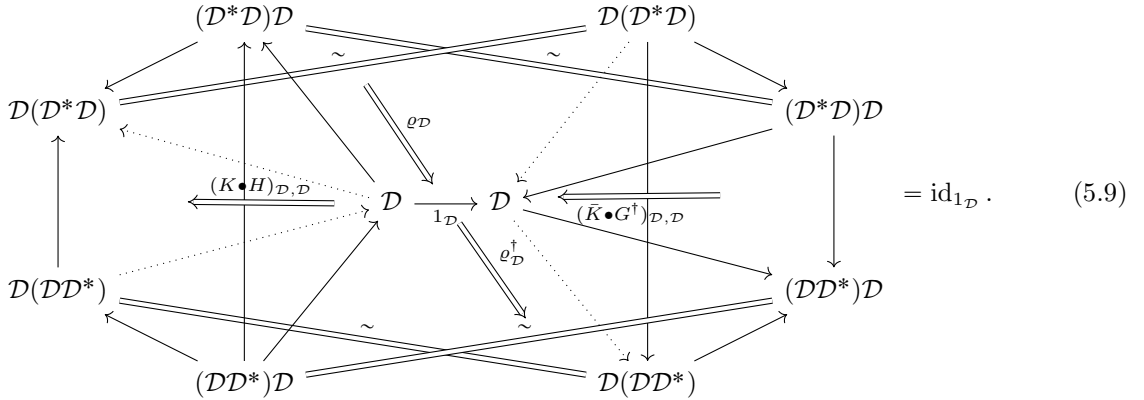
\begin{equation}\begin{tikzcd}
	& {(\cD^*\cD)\cD} &&& {\cD(\cD^*\cD)} \\
	{\cD(\cD^*\cD)} &&&&& {(\cD^*\cD)\cD} \\
	&& \cD & \cD \\
	{\cD(\cD\cD^*)} &&&&& {(\cD\cD^*)\cD} \\
	& {(\cD\cD^*)\cD} &&& {\cD(\cD\cD^*)}
	\arrow[""{name=0, anchor=center, inner sep=0}, from=1-2, to=2-1]
	\arrow["\sim", Rightarrow, no head, from=1-2, to=2-6]
	\arrow[""{name=1, anchor=center, inner sep=0}, "\sim"', Rightarrow, no head, from=1-5, to=2-1]
	\arrow[""{name=2, anchor=center, inner sep=0}, from=1-5, to=2-6]
	\arrow[dotted, from=1-5, to=3-4]
	\arrow[from=1-5, to=5-5]
	\arrow[from=2-6, to=3-4]
	\arrow[from=2-6, to=4-6]
	\arrow[from=3-3, to=1-2]
	\arrow[dotted, from=3-3, to=2-1]
	\arrow[""{name=3, anchor=center, inner sep=0}, "{1_{\cD}}"', from=3-3, to=3-4]
	\arrow[from=3-4, to=4-6]
	\arrow[dotted, from=3-4, to=5-5]
	\arrow[from=4-1, to=2-1]
	\arrow[dotted, from=4-1, to=3-3]
	\arrow["\sim", Rightarrow, no head, from=4-1, to=5-5]
	\arrow[from=5-2, to=1-2]
	\arrow[from=5-2, to=3-3]
	\arrow[""{name=4, anchor=center, inner sep=0}, from=5-2, to=4-1]
	\arrow[""{name=5, anchor=center, inner sep=0}, "\sim", Rightarrow, no head, from=5-2, to=4-6]
	\arrow[""{name=6, anchor=center, inner sep=0}, from=5-5, to=4-6]
	\arrow[""{name=7, anchor=center, inner sep=0}, draw=none, from=0, to=4]
	\arrow["{\varrho_\cD}", shorten <=8pt, shorten >=8pt, Rightarrow, from=1, to=3]
	\arrow[""{name=8, anchor=center, inner sep=0}, draw=none, from=2, to=6]
	\arrow["{\varrho_\cD^\dagger}", shorten <=8pt, shorten >=8pt, Rightarrow, from=3, to=5]
	\arrow["{(\bar K\bullet  G^\dagger)_{\cD,\cD}}", shift right, shorten <=7pt, shorten >=13pt, Rightarrow, from=8, to=3-4]
	\arrow["{(K\bullet H)_{\cD,\cD}}"', shorten <=13pt, shorten >=13pt, Rightarrow, from=3-3, to=7]
\end{tikzcd}.\label{rigidwrithe2}
\end{equation}
The commutativity of these diagrams \eqref{rigidwrithe1}, \eqref{rigidwrithe2}, together with their dual versions involving $\bar H,\bar G$, are the desired \textbf{rigid writhing conditions}. 



These conditions involving \eqref{rigidwrithe1}, \eqref{rigidwrithe2} can be understood as framed versions of the writhing coherence axioms given in \cite{BAEZ2003705}, Definition 14. Geometrically, these conditions express the equivalence between the two ways in which a writhing can be passed through duality cusps $\varphi,\varrho$; see fig. 33 of \cite{BAEZ2003705}. But here, the difference between left and right duals allows us to keep track of the framing --- we will see this in \S \ref{ribbon}, specifically \textbf{Corollary \ref{generalC8}}. 

\subsubsection{The swallowking equations}\label{swallowkingeqs}
Equipped with the writhing, we now construct the diagram which "pastes" the two swallowtail equations together, still under the assumption that $\tilde R$ is quasi-Hemritian. Consider the braiding map $c_{\cD,\cD^*}: \cD\cD^*\rightarrow \cD^*\cD$ and its adjoint $c_{\cD^*,\cD}\cong c_{\cD,\cD^*}^\dagger:\cD\cD^*\rightarrow\cD^*\cD$. The adjoint-mate of the under-writhe
\begin{equation*}
    \bar K_\cD'^\dagger:c_{\cD^*,\cD}^\dagger\cong c_{\cD,\cD^*}\Rightarrow \operatorname{ev}_\cD^\dagger \operatorname{cev}_\cD^\dagger,
\end{equation*}
together with the interchangers, allows us to form the following 2-morphism
\begin{equation}
    \begin{tikzcd}
	& {\cD\cD^*} && {\cD^*\cD} \\
	{\cD\cD^*\cD\cD^*} && \cI && {\cD\cD^*\cD^*\cD} \\
	& {\cD^*\cD} && {\cD\cD^*}
    	\arrow[""{name=0, anchor=center, inner sep=0}, from=1-2, to=1-4]
	\arrow[from=1-2, to=2-3]
	\arrow[from=1-4, to=2-5]
	\arrow["{\upsilon_{\operatorname{cev}_\cD,\operatorname{ev}_\cD}}"{description}, shorten <=6pt, shorten >=6pt, Rightarrow, from=1-4, to=3-4]
	\arrow[from=2-1, to=1-2]
	\arrow[from=2-3, to=1-4]
	\arrow[from=2-3, to=3-2]
	\arrow[from=2-5, to=3-4]
	\arrow["{\upsilon_{\operatorname{ev}_\cD,\operatorname{cev}_\cD}}"{description}, shorten <=6pt, shorten >=6pt, Rightarrow, from=3-2, to=1-2]
	\arrow[from=3-2, to=2-1]
	\arrow[from=3-4, to=2-3]
	\arrow[""{name=1, anchor=center, inner sep=0}, from=3-4, to=3-2]
	\arrow["{\bar K_\cD'^\dagger}"', shorten <=3pt, shorten >=3pt, Rightarrow, from=0, to=2-3]
	\arrow["{K_\cD'}", shorten <=3pt, shorten >=3pt, Rightarrow, from=1, to=2-3]
\end{tikzcd}\label{writheswap}
\end{equation}
On the quadruples $\cD^*\cD\cD^*\cD$ and $\cD\cD^*\cD\cD^*$, we also have the following interchangers,
\[\upsilon_{c_{\cD,\cD^*},c_{\cD,\cD^*}}=\begin{tikzcd}
	{\cD\cD^*\cD\cD^*} & {\cD\cD^*\cD^*\cD} \\
	{\cD^*\cD\cD\cD^*} & {\cD^*\cD\cD^*\cD}
	\arrow[from=1-1, to=1-2]
	\arrow[from=1-1, to=2-1]
	\arrow[shorten <=8pt, shorten >=8pt, Rightarrow, from=1-1, to=2-2]
	\arrow[from=1-2, to=2-2]
	\arrow[from=2-1, to=2-2]
\end{tikzcd},\qquad \upsilon_{c_{\cD^*,\cD},c_{\cD^*,\cD}}=
\begin{tikzcd}
	{\cD^*\cD\cD^*\cD} & {\cD\cD^*\cD^*\cD} \\
	{\cD^*\cD\cD\cD^*} & {\cD\cD^*\cD\cD^*}
	\arrow[from=1-1, to=1-2]
	\arrow[from=1-1, to=2-1]
	\arrow[shorten <=8pt, shorten >=8pt, Rightarrow, from=1-1, to=2-2]
	\arrow[from=1-2, to=2-2]
	\arrow[from=2-1, to=2-2]
\end{tikzcd},\]
which are adjoints of each other $\upsilon_{c_{\cD,\cD^*},c_{\cD,\cD^*}}=\upsilon_{c_{\cD^*,\cD},c_{\cD^*,\cD}}^\dagger$ by quasi-Hemriticity. 

This interchanger fits into the centre of the following diamond 
\begin{equation}\begin{tikzcd}
	&&& \cI \\
	&& {\cD\cD^*} && {\cD\cD^*} \\
	& {\cD^*\cD} && {\cD\cD^*\cD\cD^*} && {\cD^*\cD} \\
	\cI && {\cD\cD^*\cD^*\cD} && {\cD^*\cD\cD\cD^*} && \cI \\
	& {\cD\cD^*} && {\cD^*\cD\cD^*\cD} && {\cD\cD^*} \\
	&& {\cD^*\cD} && {\cD^*\cD} \\
	&&& \cI
	\arrow[from=1-4, to=2-3]
	\arrow[from=1-4, to=2-5]
	\arrow["{\upsilon_{\operatorname{cev}_\cD,\operatorname{cev}_\cD}}"{description}, shorten <=6pt, shorten >=6pt, Rightarrow, from=1-4, to=3-4]
	\arrow[from=2-3, to=3-2]
	\arrow[from=2-3, to=3-4]
	\arrow["{\upsilon_{\operatorname{cev}_\cD,c_{\cD,\cD^*}}}"{description}, shorten <=6pt, shorten >=6pt, Rightarrow, from=2-3, to=4-3]
	\arrow[from=2-5, to=3-4]
	\arrow[from=2-5, to=3-6]
	\arrow["{\upsilon_{c_{\cD,\cD^*},\operatorname{cev}_\cD}}"{description}, shorten <=6pt, shorten >=6pt, Rightarrow, from=2-5, to=4-5]
	\arrow[from=3-2, to=4-1]
	\arrow[from=3-2, to=4-3]
	\arrow["{\upsilon_{\operatorname{cev}_\cD,\operatorname{ev}_\cD}}"{description}, shorten <=6pt, shorten >=6pt, Rightarrow, from=3-2, to=5-2]
	\arrow[from=3-4, to=4-3]
	\arrow[from=3-4, to=4-5]
	\arrow["{\upsilon_{c_{\cD,\cD^*},c_{\cD,\cD^*}}}"{description}, shorten <=6pt, shorten >=6pt, Rightarrow, from=3-4, to=5-4]
	\arrow[from=3-6, to=4-5]
	\arrow[from=3-6, to=4-7]
	\arrow["{\upsilon_{\operatorname{ev}_\cD,\operatorname{cev}_\cD}}"{description}, shorten <=6pt, shorten >=6pt, Rightarrow, from=3-6, to=5-6]
	\arrow[from=4-1, to=5-2]
	\arrow[from=4-3, to=5-2]
	\arrow[from=4-3, to=5-4]
	\arrow["{\upsilon_{c_{\cD,\cD^*},\operatorname{ev}_\cD}}"{description}, shorten <=6pt, shorten >=6pt, Rightarrow, from=4-3, to=6-3]
	\arrow[from=4-5, to=5-4]
	\arrow[from=4-5, to=5-6]
	\arrow["{\upsilon_{\operatorname{ev}_\cD,c_{\cD,\cD^*}}}"{description}, shorten <=6pt, shorten >=6pt, Rightarrow, from=4-5, to=6-5]
	\arrow[from=4-7, to=5-6]
	\arrow[from=5-2, to=6-3]
	\arrow[from=5-4, to=6-3]
	\arrow[from=5-4, to=6-5]
	\arrow["{\upsilon_{\operatorname{ev}_\cD,\operatorname{ev}_\cD}}"{description}, shorten <=6pt, shorten >=6pt, Rightarrow, from=5-4, to=7-4]
	\arrow[from=5-6, to=6-5]
	\arrow[from=6-3, to=7-4]
	\arrow[from=6-5, to=7-4]
\end{tikzcd},\label{diamond}
\end{equation}
from which we notice that the top-/bottom-most squares involve the interchangers in the swallowtails $\mathcal{S}_\cD,\mathcal{S}'_\cD$, and the left-/right-most squares fit into the 2-morphism \eqref{writheswap}. 

\medskip

Let $\mathcal{K}_\cD$ denote the 2-morphism we obtain from the above diamond by "gluing" its left and right sides with \eqref{writheswap}. By subsequently pasting $\mathcal{S}_\cD,\mathcal{S}'_\cD$ to the top/bottom squares, we obtain another 2-morphism which we denote suggestively by $\mathcal{S}_\cD\circ \mathcal{K}_\cD \circ\mathcal{S}'_\cD$. An adjoint version of this 2-morphism can also be constructed, which involves $\operatorname{cev}_\cD^\dagger,\operatorname{ev}_\cD^\dagger$ and the adjoint of the over-writhe $K_\cD^\dagger$. The triviality of these resulting 2-morphisms,
\begin{equation}
    \mathcal{S}_\cD\circ \mathcal{K}_\cD \circ\mathcal{S}'_\cD = \id_{c_{\cD,\cD^*}},\qquad \mathcal{S}^\dagger_\cD\circ \mathcal{K}^\dagger_\cD \circ \mathcal{S}'^\dagger_\cD = \id_{c_{\cD^*,\cD}} ,\label{swallowking}
\end{equation}
is dubbed \textbf{the swallowking equations}. Since this condition glues the two swallowtail 2-morphisms together, the author has named it after the \textit{rat king} phenomenon, where the tails of a group of rats become entangled. We will use this condition later in \S \ref{hopfswallow} and \S \ref{classicalpivotal}.

\section{Ribbon tensor 2-category $\operatorname{2Rep}(\tilde C;\tilde R)$}\label{ribbon2cat}
Equipped with the structural functors and natural transformations that we have deduced from the braiding and rigidity of $\operatorname{2Rep}(\tilde C;\tilde R)$, we now study how these come together in a compatible manner. Following the classic results \cite{Reshetikhin:1990pr,Woronowicz1988,Majid:1996kd} that representations of quantum groups form ribbon tensor categories, we shall leverage the following to develop a notion of a \textbf{ribbon tensor 2-category} equipped with adjoints. 

The central motivation for this is the following; let $\mathcal{T}$ is the ribbon 2-category of 2-tangles \cite{BAEZ2003705}, then one can define decorated 2-ribbons as {\it ribbon 2-functors}
\begin{equation*}
    \mathcal{T}\rightarrow\operatorname{2Rep}(\tilde C;\tilde R)
\end{equation*}
into $\operatorname{2Rep}(\tilde C;\tilde R)$, which serves as an explicit example of the 2-tangle hypothesis in \cite{BAEZ1996196,baez19982}. This moreover allows us to define the 2-Chern-Simons TQFT $Z_{2CS}$ through a \textit{higher-skein} perspective (cf. \cite{Manolescu2022SkeinLM,Morrison2019InvariantsO4}), where as a functor $Z_{2CS}$ assigns a 3-dimensional manifold the space of embedded 2-ribbons. This idea has been pursued in the sequel \cite{chen2:2025}.



\begin{rmk}
    Strictly speaking, the framework developed in \cite{Chen1:2025?} a priori only leads to the definition of the \textit{discretized} 2-Chern-Simons theory, $Z_{d2CS}$. In this context, 2-tangles are treated as embedded as (2-)graphs into the lattice $\Gamma$, such that its boundaries live on the 1-skeleton $\Gamma^1$. The "true" 2-Chern-Simons TQFT by taking a direct limit over the lattice refinements. However, the idea of the 2-tangle hypothesis is that, given the ribbon 2-functors into $\operatorname{2Rep}(\tilde{C};\tilde R)$,  the 2-Chern-Simons TQFT can also directly be obtained through higher-skein theory.
\end{rmk}

\subsection{Tortile objects; the ribbon balancing}\label{ribbon}
We now finally come to the ribbon balancing/twist. Consider first the braiding structure $c_{\cD^*,\,^*\cD}: \cD^*\boxtimes\,^*\cD\rightarrow\,^*\cD\boxtimes\cD^*$. For the same reason as described in \S \ref{braidduals}, these are associated with the following quantities
\begin{equation*}
    \mathfrak{t}= (-\cdot-)(\tilde{S}^{-1}\otimes\tilde{S})(\tilde R^T)=\frak{t}^l+\frak{t}^r,
\end{equation*}
which comes with its own nudging equations,
\begin{equation*}
    \frak{t}^l\cdot \frak{t}^r=\frak{t}^r\cdot \frak{t}^l
\end{equation*}
following from the compatibility of the cobraiding natural transformation $\tilde R$. We call these, specifically $\frak{t}$, the \textbf{tortile object} of $\tilde C$. Note $\tilde R$ being quasi-Hermitian does {\it not} imply $c_{\cD^*,\,^*\cD}^\dagger \cong c_{\,^*\cD,\cD^*}$, since the tortile object is different from the objects $\nu,\mu$ introduced in \S \ref{braidduals}.

From the braiding map $c_{\cD^*,\,^*\cD}$, we define for each object $\cD$ the \textbf{left-over/right-under balancings}
\begin{align*}
    & \vartheta_\cD= (\overline{\operatorname{ev}}_\cD\boxtimes\cD^*)\circ (\cD\boxtimes c_{\cD^*,\,^*\cD})\circ(\operatorname{cev}_\cD\boxtimes\,^*\cD): \,^*\cD\rightarrow \cD^* \\
    & \bar\vartheta_\cD = (\,^*\cD\boxtimes{\operatorname{ev}}_\cD)\circ (c_{\cD^*,\,^*\cD}\boxtimes\cD)\circ(\cD^*\boxtimes\overline{\operatorname{cev}}_\cD):\cD^*\rightarrow\,^*\cD.
\end{align*}
The naturality of their composites gives, for each $\tilde C$-module functor $F:\cD\rightarrow\cD'$, the balancing 2-morphisms
\[\begin{tikzcd}
	{\,^*\cD'} & {\cD'^*} \\
	{\,^*\cD} & {\cD^*}
	\arrow["{\vartheta_{\cD'}}", from=1-1, to=1-2]
	\arrow["{\,^*F}"', from=1-1, to=2-1]
	\arrow["{F^*}", from=1-2, to=2-2]
	\arrow["{\vartheta_F}"', shorten <=4pt, shorten >=4pt, Rightarrow, from=2-1, to=1-2]
	\arrow["{\vartheta_\cD}"', from=2-1, to=2-2]
\end{tikzcd},\qquad \begin{tikzcd}
	{\cD'^*} & {\,^*\cD'} \\
	{\cD^*} & {\,^*\cD}
	\arrow["{\bar\vartheta_{\cD'}}", from=1-1, to=1-2]
	\arrow["{F^*}"', from=1-1, to=2-1]
	\arrow["{\,^*F}", from=1-2, to=2-2]
	\arrow["{\bar\vartheta_F}"', shorten <=4pt, shorten >=4pt, Rightarrow, from=2-1, to=1-2]
	\arrow["{\bar\vartheta_\cD}"', from=2-1, to=2-2]
\end{tikzcd}.\]
One can show (through tedious but straightforward computations) from \eqref{foldtensor} and \eqref{quadbraid}, as well as \eqref{antpod}, \eqref{bimonoid}, that we have the following {balancing} equations
\begin{align}
    & \vartheta_{\cD\boxtimes\cA}\cong c_{\cD^*,\cA^*}\circ (\vartheta_\cD\boxtimes\vartheta_\cA)\circ c_{\,^*\cA,\,^*\cD},\nonumber\\
    & \bar\vartheta_{\cD\boxtimes\cA}\cong c_{\,^*\cD,\,^*\cA}\circ (\bar\vartheta_\cD\boxtimes\bar\vartheta_\cA)\circ c_{\cA^*,\cD^*}\label{balancing}
\end{align}
for each $\cD,\cA\in\operatorname{2Rep}(\tilde C;\tilde R)$; indeed, \eqref{quadbraid} implies that these 2-morphisms are given by the unitary hexagonators, and so are unitary themselves and therefore invertible.

\begin{rmk}\label{pivottwist}
    Notice $\vartheta_{\cD^*}: \cD\rightarrow(\cD^*)^*$ is precisely the functor comparing $\cD$ and its double-dual mentioned in {\it Remark \ref{warning}}. It is then clear that the 2-categorical dimension $\mathfrak{Dim}(\cD)$ admits an action only by the centralizer subcategory $C_{\operatorname{End}(\cD)}(\vartheta)$, which consist of functors $F\in\operatorname{End}(\cD)$ that commute with $\vartheta_{\cD^*}$, and natural transformations that commute with the 2-Drinfel'd modification $\omega_{\cD}$ (which we shall introduce soon). The issue raised in {\it Warning 2.2.5} of \cite{Douglas:2018} can thus be circumvented if $C_{\operatorname{End}(\cD)}(\vartheta)\simeq\mathsf{Hilb}$ is trivial. We call braided rigid tensor 2-categories with this property \textbf{maximally imbalanced}.
\end{rmk}

\subsubsection{2-Drinfel'd modifications}
Recall $\tilde S^2,\tilde S^{-2}:\tilde C\rightarrow\tilde C$ are by hypothesis monoidal autoequivalences. By combining \eqref{dualintertwine} and \textit{Remark \ref{preduals}}, we construct invertible 2-morphisms such that
\begin{equation}
    \overline{\operatorname{c/ev}}_\cD^{**}\Rightarrow \operatorname{c/ev}_{\cD^*},\qquad \,^{**}\operatorname{c/ev}_\cD\Rightarrow \overline{\operatorname{c/ev}}_{\,^*\cD},\label{dualinside}
\end{equation}
where "$\operatorname{c/ev}$" means either $\operatorname{ev}$ or $\operatorname{cev}$.
\begin{proposition}\label{2drinfeld}
     There are \textbf{2-Drinfel'd modifications}
        \begin{equation*}
        \omega_\cD: \bar\vartheta_\cD^*\Rightarrow \vartheta_{\cD^*},\qquad \bar\omega_\cD: \,^*\vartheta_{\cD}\Rightarrow \bar\vartheta_{\,^*\cD},
    \end{equation*}
    which witness the homotopy between left-over and right-under balancings upon a \textit{reversal} of the framing.
\end{proposition}
\begin{proof}
Recall the (horizontal) antipode $\tilde S:\tilde C\rightarrow\tilde C^{\text{m-op},\text{c-op}}$ is a (strict) op-monoidal functor $\tilde S\circ (-\cdot -) = (-\cdot -)^\text{op}\circ (\tilde S\otimes\tilde S)$. Its adjunctions $\tilde S^{-1}\circ \tilde S\cong 1_{\tilde C} \cong \tilde S\circ\tilde S^{-1}$ lead to the following identifications
\begin{equation*}
    \tilde{S}\mathfrak{t}\cong (-\cdot -)(\tilde S^2\otimes 1)\tilde R,\qquad S^{-1}\mathfrak{t}\cong (-\cdot-)(1\otimes\tilde S^{-2})\tilde R
\end{equation*}
on the tortile objects. This in turn induces the invertible (unitary) 2-morphisms
\begin{equation}
    c_{-^*,-^*}: c_{\cD^*,\,^*\cD}^*\Rightarrow c_{(\cD^*)^*,\cD},\qquad c_{\,^*-,\,^*-}:\,^*c_{\cD^*,\,^*\cD}\Rightarrow c_{\cD,\,^*(\,^*\cD)},\label{dualbraid}
\end{equation}
corresponding to the braiding on the the duals $-^*$ or the pre-duals $\,^*-$. 

By horizontally composing $c_{-^*,-^*}$ with the 2-morphisms in {\bf Proposition \ref{dualevals}}, we have with \eqref{dualinside},
\begin{align*}
    \bar\vartheta_\cD^* &= \big(\overline{\operatorname{cev}}^*_\cD\boxtimes (\cD^*)^*\big)\circ (\cD^*\boxtimes c_{\cD^*,\,^*\cD}^*) \circ (\operatorname{ev}_\cD^*\boxtimes\cD) \\
    &\Rightarrow \big(\operatorname{ev}_\cD\boxtimes(\cD^*)^*\big)\circ(\cD^*\boxtimes c_{(\cD^*)^*,\cD})\circ (\overline{\operatorname{cev}}_\cD^{**}\boxtimes\cD) \\
    &\cong \big(\overline{\operatorname{ev}}_{\cD^*}\boxtimes(\cD^*)^*\big)\circ (\cD^*\boxtimes c_{(\cD^*)^*,\cD})\circ (\operatorname{cev}_{\cD^*}\boxtimes\cD) = \vartheta_{\cD^*},
\end{align*}
where we have used the equivalences mentioned in {\it Remark \ref{preduals}}. This defines $\omega_\cD$. The other 2-morphism $\bar\omega_\cD$ can be constructed in an analogous way. 
\end{proof}
\noindent Moreover, the balancing \eqref{balancing} and \eqref{dualbraid} allow us to achieve the identifications
\begin{equation*}
    \omega_{\cD\boxtimes\cA} = c_{-^*,-^*}\circ (\omega_\cD\boxtimes\omega_\cA)\circ c_{-^*,-^*},\qquad \bar\omega_{\cD\boxtimes\cA} = c_{\,^*-,\,^*-}\circ (\bar\omega_\cD\boxtimes\bar\omega_\cA)\circ c_{\,^*-,\,^*-}
\end{equation*}
for each $\cD,\cA$, hence the 2-Drinfel'd modifications are compatible with the tensor product.

If we suppose for the moment that we have a monoidal pseudonatural isomorphism $-^{****}\simeq \id$ on $\operatorname{2Rep}(\tilde C;\tilde R)$, then the invertibility of \eqref{dualinside} gives rise to 
    \begin{equation*}
         \overline{\operatorname{c/ev}}_\cD \cong  (\overline{\operatorname{c/ev}})^{****}_\cD\cong \operatorname{c/ev}_{\cD^*}^{**},
    \end{equation*}
which in essence identifies barred quantities as a "half-way" to the quadruple dual. We will prove the triviality of the quadruple dual in a separate paper.

\begin{rmk}\label{categoricalS4}
   An alternative proof of $-^{****}\simeq \id$ is to leverage a "categorical" Radford $S^4$-formula\footnote{Recall that the usual Radford $S^4$-formula \cite{Delvaux2006ANO} states that the action of $S^4$ in a finite-dimensional Hopf algebra is an inner automorphism by grouplike elements.} proven in \cite{Etingof:2004}: there exists a \textit{distinguished invertible object} $c\in C$ in a finite tensor category $C$ with a natural isomorphism $\delta: -^{**} \Rightarrow c\otimes \,^{**}-\otimes c^{-1}$ --- the idea is that the rigid duality $-^*$ gives rise to an antipode functor on the dual $C^*$. Thus it is reasonable to expect a categorical Radford $S^4$-formula for the Hopf category $\tilde{C}=\tilde C$ to hold. One subtlety is that $\tilde C$ itself is not finite, but it can be made unimodular with the use of a \textit{Haar cointegral} (see \S 5.2.3 in \cite{chen2:2025}).
\end{rmk}   

Due to the naturality of the ribbon balancing functors, an immediate corollary is the following.
\begin{corollary}\label{generalC8}
    When $\omega_\cD,\bar\omega_\cD$ are invertible, then for all $\tilde C$-module functors $F:\cD\rightarrow\cD'$ we have
    \begin{equation}
    \vartheta_{F^*} = \bar\vartheta_F^*,\qquad \vartheta_{\,^* F} = \,^*\bar\vartheta_F.\label{rigidC8}
\end{equation}
As such, $\operatorname{End}(\cI)$ is a ribbon tensor category.
\end{corollary}
\noindent Notice these are rigid dagger generalizations of the C8 condition in \cite{Douglas:2018}, Definition 2.2.4.

\begin{rmk}
    The reason we call $\omega_\cD,\bar\omega_\cD$ the "2-Drinfel'd modifications" is the following. In the scenario where $\vartheta,\bar\vartheta$ are genuine "twists" --- ie. (pseudo)natural transformations on the identity on $\operatorname{2Rep}(\tilde C;\tilde R)$ (which cannot happen unless the duality is involutive) --- then  $\omega=\bar\omega^{-1}$ is an invertible modification between these pseudonatural twists. As such, they serve to "change" the pivotal structures on rigid 2-categories, and hence appears in \textit{Remark \ref{pivottwist}} as part of the notion of "maximal imbalancing".
\end{rmk}


\subsubsection{Reidemeister II: double-twist cancellations and the belt-buckle move}\label{beltbuckle}
In the rigid dagger setting, each of the composites in the definition of the ribbon balancings $\vartheta_\cD,\bar\vartheta_\cD$ are planar-unitary, and hence are themselves planar-unitary. As such they admit folds against their adjoints. The folds
\begin{equation*}
    e_{\vartheta_\cD}:\vartheta_\cD^\dagger\circ\vartheta_\cD\Rightarrow 1_{\,^*\cD},\qquad e_{\bar\vartheta_\cD}: \bar\vartheta_\cD^\dagger\circ\bar\vartheta_\cD\Rightarrow 1_{\cD^*},
\end{equation*}
in particular, are known as the \textbf{Kauffman double twist cancellations}; see fig. 55 (d) in \cite{Barrett_2024}. These adjunctions witness the cancellation of twists that live on the {\it same side} of the tangle. We call such twist cancellations "of the first type". 

As the name suggests, there is another type of double twist cancellation, in which the twists lie on different sides of the tangle. 
\begin{proposition}
    Suppose $\tilde R$ is quasi-Hemritian \eqref{quasiherm}. There are 2-morphisms
    \begin{equation*}
    \kappa_\cD:\big(\cD\boxtimes(\vartheta_\cD\circ\bar\vartheta_\cD)\big)\circ\operatorname{cev}_\cD\Rightarrow \operatorname{cev}_\cD,\qquad \bar\kappa_\cD:\big((\bar\vartheta_\cD\circ \vartheta_\cD)\boxtimes\cD\big)\circ\overline{\operatorname{cev}}_\cD\Rightarrow\overline{\operatorname{cev}}_\cD
    \end{equation*}
    which witness the null-homotopy of distinct twist types on a fold.
\end{proposition}
\begin{proof}
    To prove this, we need to introduce the so-called \textbf{belt-buckle moves}. Geometrically, these are isotopies which "drag" twists on the \textit{outside} (and only the outside!) of a fold to the top. They are implemented by the following 2-morphisms
\begin{equation*}
    \mathrm{P}_\cD: (\cD\boxtimes\bar\vartheta_\cD)\circ\operatorname{cev}_\cD\Rightarrow c_{\,^*\cD,\cD}\circ\overline{\operatorname{cev}}_{\cD},\qquad \mathrm{Q}_\cD: (\vartheta_{\cD}\boxtimes \cD)\circ \overline{\operatorname{cev}}_\cD\Rightarrow c_{\cD,\cD^*}\circ \operatorname{cev}_\cD.
\end{equation*}
We shall describe the construction of $\mathrm{P}_\cD$ in our current context; $\mathrm{Q}_\cD$ can be obtained in a similar manner. From the quasi-Hermitian hypothesis $c_{\,^*\cD,\cD}^\dagger\cong c_{\cD,\,^*\cD}$, we have an adjunction
\[e_{c_{\,^*\cD,\cD}}\cD^*\cD: \begin{tikzcd}
	{(\,^*\cD)\cD\cD^*\cD} && {\cD(\,^*\cD)\cD^*\cD}
	\arrow[""{name=0, anchor=center, inner sep=0}, "{c_{\,^*\cD,\cD}\cD^*\cD}"', curve={height=12pt}, from=1-1, to=1-3]
	\arrow[""{name=1, anchor=center, inner sep=0}, "{c_{\cD, \,^*\cD}\cD^*\cD}"', curve={height=12pt}, from=1-3, to=1-1]
	\arrow["\vdash"{description}, draw=none, shorten <=3pt, shorten >=3pt, Rightarrow, from=0, to=1]
\end{tikzcd}: \iota_{c_{\,^*\cD,\cD}}\cD^*\cD.\]
The 2-morphism $\mathrm{P}_\cD$ is given by the diagram
\[\mathrm{P}_\cD= \begin{tikzcd}
	& \cI \\
	{\cD\cD^*} && {(\,^*\cD)\cD} & {\cD(\,^*\cD)} \\
	& {\cD\cD^*(\,^*\cD)\cD} & {(\,^*\cD)\cD\cD^*\cD} & {\cD(\,^*\cD)\cD^*\cD}
	\arrow[from=1-2, to=2-1]
	\arrow[from=1-2, to=2-3]
	\arrow["{\upsilon_{\overline{\operatorname{cev}}_\cD,\operatorname{cev}_\cD}}", shorten <=6pt, shorten >=6pt, Rightarrow, from=1-2, to=3-2]
	\arrow[from=2-1, to=3-2]
	\arrow[from=2-3, to=2-4]
	\arrow[from=2-3, to=3-2]
	\arrow[""{name=0, anchor=center, inner sep=0}, from=3-2, to=3-3]
	\arrow[""{name=1, anchor=center, inner sep=0}, curve={height=30pt}, from=3-2, to=3-4]
	\arrow[from=3-3, to=2-3]
	\arrow[shorten <=8pt, shorten >=8pt, Rightarrow, from=3-3, to=2-4]
	\arrow[""{name=2, anchor=center, inner sep=0}, draw=none, curve={height=-6pt}, from=3-3, to=3-4]
	\arrow[from=3-4, to=2-4]
	\arrow[""{name=3, anchor=center, inner sep=0}, draw=none,curve={height=-6pt}, from=3-4, to=3-3]
	\arrow[shorten <=6pt, shorten >=6pt, Rightarrow, from=2-3, to=0]
	\arrow["{\Omega_{\,^*\cD|\cD\cD^*}}"', shorten <=3pt, shorten >=1pt, Rightarrow, from=1, to=3-3]
	\arrow["\vdash"{description}, draw=none, from=2, to=3]
\end{tikzcd},\]
where the central triangle is filled by $c_{\,^*\cD\boxtimes\cD,\operatorname{cev}_\cD}\bullet \bar K_\cD'$, and the square to its right is ${\upsilon_{\operatorname{ev}_\cD,c_{\,^*\cD,\cD}}}$. 

Now starting with the functor 
\begin{equation*}
    \big(\cD\boxtimes(\vartheta_\cD\circ\bar\vartheta_\cD)\big)\circ\operatorname{cev}_\cD: \cI \rightarrow \cD\boxtimes\,^*\cD,
\end{equation*}
we can apply a series of 2-morphisms
\begin{align*}
    \big(\cD\boxtimes(\vartheta_\cD\circ\bar\vartheta_\cD)\big)\circ\operatorname{cev}_\cD &\xRightarrow{(\cD\boxtimes\vartheta_\cD)\circ \mathrm{P}_\cD} (\cD\boxtimes\vartheta_\cD)\circ c_{\,^*\cD,\cD}\circ\overline{\operatorname{cev}}_\cD \\
    & \xRightarrow{c_{\cD,\vartheta_\cD}\circ\overline{\operatorname{cev}}_\cD} c_{\cD^*,\cD}\circ (\vartheta_\cD\boxtimes\cD)\circ\overline{\operatorname{cev}}_\cD\xRightarrow{c_{\cD^*,\cD}\circ \mathrm{Q}_\cD}c_{\cD^*,\cD}\circ c_{\cD,\cD^*}\circ\operatorname{cev}_\cD\\
    & \cong c_{\cD,\cD^*}^\dagger\circ c_{\cD,\cD^*}\circ\operatorname{cev}_\cD \xRightarrow{e_{c_{\cD,\cD^*}}\circ\overline{\operatorname{cev}}_\cD} \overline{\operatorname{cev}}_\cD,
\end{align*}
where in the last line we have used the quasi-Hermimticity property $c_{\cD^*,\cD} \cong c_{\cD,\cD^*}^\dagger$. This defines $\kappa_\cD$; the other one $\bar\kappa_\cD$ can be obtained similarly. 
\end{proof}
\noindent The belt-buckle moves $\mathrm{P},\mathrm{Q}$ are also useful for transporting the rigidity conditions introduced in \S \ref{rigidwrithes}, \S \ref{swallowkingeqs} to the ribbon balancings.

There is a priori no reason for the different twist types to be related. Indeed, in the geometric picture, the $2\pi$-rotations of the framing introduced by the left-under $\vartheta_\cD^\dagger$ and right-under $\bar\vartheta_\cD$ balancings are related by a reflection about the vertical line. However, even if the two ribbon balancings do coincide, it still does not mean that they can be trivialized individually. This brings us to the first Reidemseister move mentioned in \textit{Remark \ref{firstreide}}.

\subsubsection{Reidemeister I: half-framed, unframed and self-dual objects}\label{reidemeisterI}
We now prove in this section that, under certain conditions, we can find invertible 2-morphisms which witness the \textbf{first Reidemeister moves}. However, due to the intricate duality structures involved, there are subtleties associated to the constructions --- specifically, we will construct \textit{two} different Reidemeister I witnesses in the following. 

For this section, we will assume $\tilde R$ is quasi-Hermitian \eqref{quasiherm}, and that  the fold-crossing 2-morphisms $H,G$ and the writhes $K,\bar K$, as defined in \S \ref{writhe}, are invertible. We now seek to construct two related, but a priori distinct, versions of the Reidemeister I move. These are invertible 2-morphisms
        \begin{equation*}
    \mathrm{R}_\cD: \bar\vartheta_\cD\Rightarrow (\operatorname{ev}_\cD\boxtimes\,^*\cD)\circ(\cD^*\boxtimes\overline{\operatorname{ev}}_\cD^\dagger),\qquad \mathrm{L}_\cD: (\overline{\operatorname{cev}}_\cD^\dagger\boxtimes\cD^*) \circ (\,^*\cD\boxtimes\operatorname{cev}_\cD)\Rightarrow\vartheta_\cD.
\end{equation*}
Note crucially that quasi-Hermiticity does \textit{not} imply $c_{\cD^*,\,^*\cD}^\dagger\cong c_{\,^*\cD,\cD^*}$!
\begin{enumerate}
    \item {\bf Definition of $\mathrm{R}_\cD$:} by definition,
    \begin{align*}
         \bar\vartheta_\cD &=   (\,^*\cD\boxtimes{\operatorname{ev}}_\cD)\circ (c_{\cD^*,\,^*\cD}\boxtimes\cD)\circ(\cD^*\boxtimes\overline{\operatorname{cev}}_\cD) \\
        & \qquad\xRightarrow{G_{\cD,\,^*\cD}^{-1}\circ (\cD^*\boxtimes\overline{\operatorname{cev}}_\cD)} (\operatorname{ev}_\cD\boxtimes\,^*\cD)\circ (\cD^*\boxtimes c_{\cD,\,^*\cD}^\dagger)\circ (\cD^*\boxtimes\overline{\operatorname{cev}}_\cD) \\
        &\qquad\cong (\operatorname{ev}_\cD\boxtimes\,^*\cD)\circ \big(\cD^*\boxtimes(c_{\,^*\cD,\cD}\circ \operatorname{cev}_{\,^*\cD})\big) \\
        &\qquad\xRightarrow{(\operatorname{ev}_\cD\boxtimes\,^*\cD)\circ (\cD^*\boxtimes K_{\,^*\cD})} (\operatorname{ev}_\cD\boxtimes \,^*\cD)\circ (\cD^*\boxtimes \overline{\operatorname{ev}}_\cD^\dagger),
    \end{align*}
    where in the third line we have used quasi-Hermiticity $c_{\cD,\,^*\cD}^\dagger\cong c_{\,^*\cD,\cD}$, and the equivalences described in \textit{Remark \ref{preduals}}. 
    \item \textbf{Definition of $\mathrm{L}_\cD$:} we begin instead with the adjoint $\vartheta_\cD^\dagger$. Then,
    \begin{align*}
          \vartheta_\cD^\dagger &= (\operatorname{cev}_\cD^\dagger\boxtimes\,^*\cD)\circ(\cD\boxtimes c_{\cD^*,\,^*\cD}^\dagger) \circ (\overline{\operatorname{ev}}_\cD^\dagger\boxtimes\cD^*) \\
        &\qquad\xRightarrow{H_{\cD,\,^*\cD}\circ (\overline{\operatorname{ev}}_\cD^\dagger\boxtimes\cD^*)} (\,^*\cD\boxtimes\operatorname{cev}_\cD^\dagger)\circ(c_{\cD,\,^*\cD}\boxtimes\cD^*)\circ (\overline{\operatorname{ev}}_\cD^\dagger\boxtimes\cD^*) \\
        &\qquad\cong (\,^*\cD\boxtimes\operatorname{cev}_\cD^\dagger)\circ \big((c_{\cD,\,^*\cD}\circ \operatorname{ev}_{\,^*\cD}^\dagger)\boxtimes\cD^*\big) \\
        &\qquad \xRightarrow{(\,^*\cD\boxtimes\operatorname{cev}_\cD^\dagger)\circ (\bar K_{\,^*\cD}^{-1}\boxtimes \cD^*)}(\,^*\cD\boxtimes\operatorname{cev}_\cD^\dagger)\circ (\overline{\operatorname{cev}}_\cD\boxtimes\cD^*),
    \end{align*}
    whence taking the adjunction-mate yields the desired 2-morphism $\mathrm{L}_\cD$.
\end{enumerate}

\begin{proposition}\label{1streidemeister1}
There exist \textbf{double-twist cancellations of the second type} $$m_\cD:\bar\vartheta_\cD\circ\vartheta_\cD\Rightarrow 1_{\,^*\cD},\qquad n_\cD: 1_{\cD^{*}}\Rightarrow \vartheta_\cD\circ \bar\vartheta_\cD$$ as mentioned previously in \S \ref{beltbuckle}.
\end{proposition}
\noindent These are "off-the-fold" versions of $\kappa_\cD,\bar\kappa_\cD$. To see how they are related, see \S \ref{baezlangfordassumption}.
\begin{proof}
   Consider the composite $\mathrm{R}\circ\mathrm{L}^{-1}$. Under the hypothesis of the previous proposition, $\mathrm{R}\circ\mathrm{L}^{-1}$ is a(n invertible) 2-morphism
\begin{equation*}
    \bar\vartheta_\cD\circ\vartheta_\cD \Rightarrow (\operatorname{ev}_\cD\boxtimes\,^*\cD)\circ(\cD^*\boxtimes\overline{\operatorname{ev}}_\cD^\dagger)\circ (\overline{\operatorname{cev}}_\cD^\dagger\boxtimes\cD^*) \circ (\,^*\cD\boxtimes\operatorname{cev}_\cD).
\end{equation*}
By \textit{Remark \ref{preduals}}, we can write $\overline{\operatorname{ev}}_\cD\cong \operatorname{ev}_{\,^*\cD} \cong \,^*\operatorname{cev}_\cD$, such that the functor on the right-hand side reads
\begin{equation}
(\operatorname{ev}_\cD\boxtimes\,^*\cD)\circ(\cD^*\boxtimes \,^*\operatorname{cev}_\cD^\dagger)\circ (\,^*\operatorname{ev}_\cD^\dagger\boxtimes\cD^*) \circ (\,^*\cD\boxtimes\operatorname{cev}_\cD),
    \label{doublesnake}
\end{equation}
Now by a series of interchangers, we can bring the barred-folds to the outside such that
\begin{equation*}
    (\,^*{\operatorname{ev}}_\cD^\dagger\boxtimes\,^*\cD)\circ \Big(\,^*\cD\boxtimes\big((\cD\boxtimes\operatorname{ev}_\cD)\circ( \operatorname{cev}_\cD\boxtimes \cD)\big)\boxtimes\,^*\cD\Big)\circ(\,^*\cD \boxtimes \,^*{\operatorname{cev}}_\cD^\dagger),
\end{equation*}
whence it becomes clear that it admits a tuple of snakerators/cusps 
\begin{align*}
    & (\,^*{\operatorname{ev}}_\cD^\dagger\boxtimes\,^*\cD)\circ \Big(\,^*\cD\boxtimes\big((\cD\boxtimes\operatorname{ev}_\cD)\circ( \operatorname{cev}_\cD\boxtimes \cD)\big)\boxtimes\,^*\cD\Big)\circ(\,^*\cD \boxtimes \,^*{\operatorname{cev}}_\cD^\dagger)\\
    &\qquad \quad \xRightarrow{    (\,^*{\operatorname{ev}}_\cD^\dagger\boxtimes\,^*\cD)\circ \Big(\,^*\cD\boxtimes\varphi_\cD\boxtimes\,^*\cD\Big)\circ(\,^*\cD \boxtimes \,^*{\operatorname{cev}}_\cD^\dagger)}     (\,^*{\operatorname{ev}}_\cD^\dagger\boxtimes\,^*\cD)\circ(\,^*\cD \boxtimes \,^*{\operatorname{cev}}_\cD^\dagger) \xRightarrow{\,^*\varphi_\cD^\dagger} 1_{\,^*\cD}.
\end{align*}
This allows us to directly trivialize the double twist; similar argument applies to $\mathrm{L}\circ \mathrm{R}^{-1}$, whence
\begin{align}
    &m_\cD = \,^*\varphi_\cD^\dagger \bullet (\,^*\cD\boxtimes\varphi_\cD\boxtimes\,^*\cD)\bullet (\mathrm{R}\circ\mathrm{L}^{-1}): \bar\vartheta_\cD\circ\vartheta_\cD\Rightarrow 1_{\,^*\cD}.\nonumber\\
    &n_\cD:(\mathrm{L}\circ \mathrm{R}^{-1})\bullet (\cD^{*}\boxtimes \,^*\varrho_{\cD}^\dagger\boxtimes \cD^*)\bullet\varrho_\cD: 1_{\cD^{*}}\Rightarrow \vartheta_\cD\circ \bar\vartheta_\cD.\label{canceldoubletwist}
\end{align}
\end{proof}

More refined notions of "framing" (cf. \cite{BAEZ2003705,Douglas:2018,Barrett_2024}) can thus be achieved.
\begin{definition}\label{framing}
Take $\cD\in\operatorname{2Rep}(\tilde C;\tilde R)$ and suppose the 2-Drinfel'd modifications $\omega_\cD,\bar\omega_\cD$ are invertible.
    \begin{enumerate}
        \item $\cD$ is \textbf{half-framed} if we have (i) an identification $\vartheta_\cD^\dagger\cong \bar\vartheta_\cD$, and (ii) the two types of Kaufmann double twist cancellations coincide,
    \begin{equation}
        m_\cD = e_{\vartheta_\cD},\qquad n_\cD = \iota_{\vartheta_\cD}.\label{half-framed}
    \end{equation}

\item a half-framed $\cD$ is \textbf{unframed} iff (i) its left- and right-duals coincide $\,^*\cD\cong\cD^*$ such that \eqref{adjdual} holds, and (ii) we have an identification $\vartheta_\cD\cong \bar\vartheta_\cD$ such that $\mathrm{L}^{-1}_\cD\bullet\mathrm{R}_\cD $ is a 2-isomorphism on $\vartheta_\cD$.

    \item a unframed $\cD$ is \textbf{self-dual} iff (i) $\cD^*\cong \cD$ and we have the identifications
    \begin{equation*}
        \operatorname{ev}_\cD \cong \operatorname{cev}_\cD^\dagger,\qquad \varphi_\cD =\varrho_\cD^\dagger,
    \end{equation*}
    and (ii) $\vartheta_\cD=\vartheta_\cD^*$ and $\mathrm{R}_\cD=\mathrm{L}_\cD^{-1}$.
    \end{enumerate}
\end{definition}

The notion of "unframed" defined above is justified as follows. From \textit{Remark \ref{warning}}, we see that the barred-folds coincide with the regular folds $\overline{\operatorname{ev}}_\cD\cong\operatorname{ev}_{\cD^*}$ whenever the dual is involutive $\cD^{**}\cong \cD$. If \eqref{adjdual} further holds, then the functor $(\operatorname{ev}_\cD\boxtimes\,^*\cD)\circ(\cD^*\boxtimes\overline{\operatorname{ev}}_\cD^\dagger)\cong (\operatorname{ev}_\cD\boxtimes \cD^*)\circ(\cD^*\boxtimes {\operatorname{cev}}_{\cD^{**}})$ admits a snakerator $\varrho_{\cD}$ from $1_{\cD^*}$. Composing its with $\mathrm{R}_\cD$ then kills a single twist
\begin{equation}
    \varrho_{\cD}^\dagger\bullet \mathrm{R}_\cD:\bar\vartheta_\cD\cong\vartheta_\cD\Rightarrow 1_{\cD^*};\label{reide1twist}
\end{equation}
similar arguments gives a trivialization of $\vartheta_\cD^\dagger$ from $\mathrm{L}_\cD^\dagger$. Having these witnesses for a single Reidemeister I move is what motivates the above definition of "unframed object". 

\begin{rmk}\label{strangesnakes}
    Let us take a closer look at the "double strange snakerators/cusps" \eqref{doublesnake}. Diagrammatically, they bear a striking resemblance to fig. 30 in \cite{Barrett_2024}, and the reason is as follows. Suppose the double-dual $-^{**}$ defines an involution $-^\#$ in accordance with \textit{Remark \ref{categoricalS4}}, then the barred folds can intuitively be understood as the double-duality datum. The condition \eqref{daggercommute} then gives a \textit{3}-endomorphism
    \begin{equation*}
        \Gamma_\text{pt}:  \text{pt}^{\#\dagger\#\dagger}\Rrightarrow \text{pt}
    \end{equation*}
    in the monoidal 3-category $B\operatorname{2Rep}(\tilde C;\tilde R)$. Together with the involutive-ness of the adjunction $\Theta_\text{pt}:\text{pt}^{\dagger\dagger}\Rrightarrow \text{pt}$, we achieve precisely \textbf{Theorem 4.5} in \cite{Barrett_2024}. Moreover, the "half-framing" condition \eqref{half-framed} identifies the two types of Kaufmann double twist cancellations, such that a correspondence with the structures in a $\mathsf{Gray}$-category with duals can be made. This substantiates the contents of Table \ref{tab:framing}.
\end{rmk}

The notions of framing introduced in {\bf Definition \ref{framing}} detect the geometry of a surface labelled by $\cD\in\operatorname{2Rep}(\tilde C;\tilde R)$ through its conditions on the ribbon twists. The situation can be summarized in the following table \ref{tab:framed}. 

\begin{table}[h]
    \centering
    \begin{tabular}{c|c|c|c|c}
       $\cD$ is\dots  & fully-framed & half-framed & unframed & self-dual \\
       \hline
        distinct twists & \makecell{$\vartheta_\cD,\vartheta_\cD^*,\vartheta_\cD^\dagger$ \\ $\bar\vartheta_\cD,\,^*\bar\vartheta_\cD,\bar\vartheta_\cD^\dagger$} & \makecell{$\vartheta_\cD,\bar\vartheta_\cD$ \\ $\vartheta_\cD^*,\,^*\bar\vartheta_\cD$} & $\vartheta_\cD,\vartheta_\cD^*$ & $\vartheta_\cD$ 
    \end{tabular}
    \caption{A table describing the number of twists on an object $\cD$ depending on how framed it is. Notice twists such as $\,^*\vartheta_\cD,\bar\vartheta_\cD^*$ did not appear due to the invertibility of the 2-Drinfel'd modifications, which relate these back to $\bar\vartheta,\vartheta$ respectively.}
    \label{tab:framed}
\end{table}

\noindent The following relations $\vartheta_\cD^{**}\cong \bar\vartheta_\cD, \,^{**}\bar\vartheta \cong \vartheta_\cD$ between the ribbon balancings can be seen from the triviality of the quadruple dual. Indeed, according to \textit{Remark \ref{categoricalS4}}, double-duals can be morally replaced by "barred" quantities.

\subsubsection{Unframed self-dual objects}\label{baezlangfordassumption}
We now turn to the notion of "self-dual" objects. If $\cD$ were self-dual, then its Reidemeister I moves in {\bf Proposition \ref{1streidemeister1}} coincide. Moreover, under quasi-Hermiticity, one also sees that the coherent writhings $K_\cD=\bar K_\cD^\dagger$ on $\cD$ coincide (see \S \ref{rigidwrithes}). A self-dual object is thus one which is equipped with (i) a single coherent writhing, and (ii) a {\it single} notion of twist and a Reidemsiter I move trivializing it. This is very close to what an "unframed self-dual object" $\cD$ means in \cite{BAEZ2003705}, but we require one more reduction.

Recall from \S \ref{beltbuckle} that, through a belt-buckle move, we can apply a writhing to remove a single twist on a fold. As such, the 2-morphisms $$K_{\,^*\cD}\bullet \mathrm{P}_\cD: (\cD\boxtimes\bar\vartheta_\cD)\circ\operatorname{cev}_\cD\Rightarrow \overline{\operatorname{ev}}_\cD^\dagger,\qquad K_\cD\bullet \mathrm{Q}_\cD: (\vartheta_{\cD}\boxtimes \cD)\circ \overline{\operatorname{cev}}_\cD \Rightarrow\operatorname{ev}_\cD^\dagger$$ should be closely related to the first Reidemeister moves \eqref{1streidemeister1}. Particularly in the \textit{unframed} case, a 2-morphism parallel with $K_{\,^*\cD}\bullet \mathrm{P}_\cD$ can be constructed from \eqref{reide1twist},
\begin{equation*}
    (\cD\boxtimes \varrho_{\cD}^*\bullet \mathrm{R}_\cD)\circ \operatorname{cev}_\cD: (\cD\boxtimes\bar\vartheta_\cD)\circ\operatorname{cev}_\cD\Rightarrow \operatorname{cev}_\cD\cong \widetilde{\operatorname{cev}}_\cD \cong \overline{\operatorname{ev}}^\dagger_\cD,
\end{equation*}
whence the condition stating that these moves are equivalent,
\begin{equation}
    (\cD\boxtimes \varrho_{\cD}^*\bullet \mathrm{R}_\cD)\circ \operatorname{cev}_\cD = K_{\,^*\cD}\bullet \mathrm{P}_\cD,\label{baez-langfordreideI}
\end{equation}
is the precise way to phrase the "writhes = first Reidemeister move" perspective taken in \cite{BAEZ2003705} (see also {\it Remark \ref{firstreide}}).


By applying \eqref{baez-langfordreideI} twice, one can immediately deduce that $\kappa_\cD$ coincides with $m_\cD$ in \eqref{canceldoubletwist}. The half-framing condition \eqref{half-framed} then tells us that the two types of twist cancellations coincide on a fold,
    \begin{equation*}
        \kappa_\cD = (\cD\boxtimes e_{\vartheta_\cD})\circ\operatorname{cev}_\cD.
    \end{equation*}
This finally allows us to recover the notion of a "unframed self-dual generator" in \cite{BAEZ2003705}.



\subsection{Double braiding; the Hopf links}\label{hopflinks}
Let us now consider the \textit{double braiding} map $C_{\cD,\cA}=c_{\cA,\cD}\circ c_{\cD,\cA}: \cD\boxtimes\cA\rightarrow\cD\boxtimes\cA$ mentioned in \textit{Remark \ref{syllepsis}}. It is an endofunctor, hence we are able to form the duality folds on them, $$\operatorname{ev}_{C_{\cD,\cA}}:C_{\cD,\cA}^*\boxtimes C_{\cD,\cA}\Rightarrow \id_{1_\cI},\qquad \operatorname{cev}_{C_{\cD,\cA}}:\id_{1_\cI}\Rightarrow C_{\cD,\cA}\boxtimes C_{\cD,\cA}^*.$$ They witness the null-homotopy of adjacent double braidings --- and hence Hopf links; see the following --- with the opposite framing. 

\subsubsection{Hopf links on distinct objects; the 2-Hopf modifications}
Let us first begin with a brief survey on the Hopf link 1-morphisms we can construct on distinct objects $\cD,\cA\in\operatorname{2Rep}(\tilde C;\tilde R)$. For each pair of objects $\cD,\cA$ and endomorphisms $F\in\operatorname{End}(\cD),G\in\operatorname{End}(\cA)$, we construct the following 1- and 2-morphisms
\begin{align*}
     &h_{\cD,\cA} = (\operatorname{cev}_\cD^\dagger\boxtimes\operatorname{cev}_\cA^\dagger) \circ  (\cD\boxtimes C_{\cD^*,\cA}\boxtimes \cA^*) \circ(\operatorname{cev}_\cD\boxtimes\operatorname{cev}_\cA): \cI\rightarrow\cI,\\
     &h_{F,\cA} = (\widetilde{\operatorname{ev}}_{F^*}\boxtimes\operatorname{cev}_\cA^\dagger) \circ  (F\boxtimes C_{F^*,\cA}\boxtimes \cA^*) \circ(\operatorname{cev}_F\boxtimes\operatorname{cev}_\cA): h_{\cD,\cA}\Rightarrow h_{\cD,\cA}, \\
     & h_{\cD,G} = (\operatorname{cev}_\cD^\dagger\boxtimes\widetilde{\operatorname{ev}}_{G*}) \circ  (\cD\boxtimes C_{\cD^*,G}\boxtimes G^*) \circ(\operatorname{cev}_\cD\boxtimes\operatorname{cev}_G): h_{\cD,\cA}\Rightarrow h_{\cD,\cA}.
\end{align*}
As the mixed double braiding 2-morphisms is compatible with nudging, we have
\begin{equation*}
    \upsilon_{F\boxtimes F^*,G\boxtimes G^*}\bullet (h_{F,\cA}\bullet h_{\cD,G})\bullet \upsilon_{F\boxtimes F^*,G\boxtimes G^*}^{-1} = h_{\cD,G}\bullet h _{F,\cA},
\end{equation*}
which means that the matrix elements of the "mixed Hopf links" $h_{F,\cA},h_{\cD,G}$ commute up to the invertible interchanger $\upsilon$. Hence, if we fix a basis for which the interchangers are diagonal, the matrices corresponding to the mixed Hopf links can be simultaneously diagonalized. 



We now consider the relationship between the Hopf links $h$ and its close relative.
\begin{theorem}
Define the $\mathrm{pivotal~opposite}$ of $h_{\cD,\cA}$,
\begin{equation*}
\tilde h_{\cD,\cA} = ({\operatorname{ev}}_{\cD}\boxtimes {\operatorname{ev}}_{\cA}) \circ  (\cD^*\boxtimes C_{\cD,\cA^*}\boxtimes \cA) \circ(\operatorname{ev}_\cD^\dagger\boxtimes\operatorname{ev}_\cA^\dagger): \cI\rightarrow\cI.
\end{equation*}
    Assume the 2-morphisms in \textbf{Proposition \ref{dualevals}} are invertible, then there are 2-morphisms such that
    \begin{equation*}
        h_{\cD,\cA}\Rightarrow \tilde h_{\cD,\cA},\qquad \,^*h_{\cD,\cA}\Rightarrow\tilde h_{\,^*\cA,\,^*\cD}.
    \end{equation*}
    Hence, there is a 2-morphism $\tilde h_{\,^*\cD,\,^*\cA}^* \Rightarrow \tilde h_{\cA,\cD}$ which swaps the arguments of the Hopf links under an orientation reversal.
\end{theorem}
\begin{proof}
    By using $c_{\cD^*,\cA}$, its unitary hexagonators and the braid exchange \eqref{braidexch}, we can form the following 2-morphism\footnote{Note that this 2-morphism acts like an interchanger, but it is \textit{not} one: it interchanges two functors which act on overlapping objects.}
\begin{equation}
    C_{c_\cD,c_\cA}=\begin{tikzcd}
	{\cD\cD^*\cA\cA^*} & {\cD\cA\cD^*\cA^*} & {\cD\cD^*\cA\cA^*} \\
	{\cD\cD^*\cA^*\cA} && {\cD\cD^*\cA^*\cA} \\
	{\cD^*\cD\cA^*\cA} & {\cD^*\cA^*\cD\cD} & {\cD^*\cD\cA^*\cA}
	\arrow[from=1-1, to=1-2]
	\arrow[from=1-1, to=2-1]
	\arrow[from=1-2, to=1-3]
	\arrow[""{name=0, anchor=center, inner sep=0}, from=1-2, to=2-3]
	\arrow[from=1-3, to=2-3]
	\arrow[""{name=1, anchor=center, inner sep=0}, from=2-1, to=1-2]
	\arrow[from=2-1, to=3-1]
	\arrow[""{name=2, anchor=center, inner sep=0}, from=2-1, to=3-2]
	\arrow[from=2-3, to=3-3]
	\arrow[from=3-1, to=3-2]
	\arrow[""{name=3, anchor=center, inner sep=0}, from=3-2, to=2-3]
	\arrow[from=3-2, to=3-3]
	\arrow["{\cD\Omega_{\dots}}", shorten >=4pt, Rightarrow, from=1-1, to=1]
	\arrow["{\cD\Omega_{\dots}}"', shorten <=4pt, shorten >=2pt, Rightarrow, from=0, to=1-3]
	\arrow["{c_{c_{\cD^*\cA^*,\cD},\cA}}"', shorten <=15pt, shorten >=15pt, Rightarrow, from=2, to=0]
	\arrow["{\Omega_{\dots}\cA}", shorten >=4pt, Rightarrow, from=3-1, to=2]
	\arrow["{\Omega_{\dots}\cA}"', shorten <=4pt, shorten >=2pt, Rightarrow, from=3, to=3-3]
\end{tikzcd},\label{doublebraidexchange}\end{equation}
where we have neglected the obvious arguments in $\Omega$ to save space. Notice if $\cA=\cD$, then the diamond in the middle of \eqref{doublebraidexchange} is related to the interchangers that have appeared in \eqref{diamond}, and therefore plays a part in the swallowking equations \eqref{swallowking} in \S \ref{swallowkingeqs}. This observation will be crucial later in \S \ref{falsswallow}.  

Together with the following adjunction described in \S \ref{2catdimension},
\[e_{\operatorname{cev}_\cD}:\begin{tikzcd}
	{\cD\cD^*} && \cI
	\arrow[""{name=0, anchor=center, inner sep=0}, "{\operatorname{cev}_\cD^\dagger}"', curve={height=12pt}, from=1-1, to=1-3]
	\arrow[""{name=1, anchor=center, inner sep=0}, "{\operatorname{cev}_\cD}"', curve={height=12pt}, from=1-3, to=1-1]
	\arrow["\vdash"{description}, draw=none, from=1, to=0]
\end{tikzcd}:\iota_{\operatorname{cev}_\cD},\]
as well as the writhings, this 2-morphism \eqref{doublebraidexchange} fits into the central square of the following diagram,
\[\Gamma_{\cD,\cA}=\begin{tikzcd}
	\cI && {\cD\cD^*\cA\cA^*} & {\cD\cD^*\cA\cA^*} && \cI \\
	\cI && {\cD^*\cD\cA^*\cA} & {\cD^*\cD\cA^*\cA} && \cI
	\arrow[from=1-1, to=1-3]
	\arrow[""{name=0, anchor=center, inner sep=0}, Rightarrow, no head, from=1-1, to=2-1]
	\arrow[from=1-3, to=1-4]
	\arrow[from=1-4, to=1-6]
	\arrow[""{name=1, anchor=center, inner sep=0}, Rightarrow, no head, from=1-6, to=2-6]
	\arrow[from=2-1, to=2-3]
	\arrow[""{name=2, anchor=center, inner sep=0}, from=2-3, to=1-3]
	\arrow[from=2-3, to=2-4]
	\arrow[""{name=3, anchor=center, inner sep=0}, from=2-4, to=1-4]
	\arrow[from=2-4, to=2-6]
	\arrow["{C_{c_\cD,c_\cA}}", shorten <=15pt, shorten >=15pt, Rightarrow, from=2, to=3]
	\arrow["{K'_\cD K_\cA'}"', shorten <=17pt, shorten >=17pt, Rightarrow, from=2, to=0]
	\arrow["{K'^\dagger_\cD K'^\dagger_\cA}"', shorten <=17pt, shorten >=17pt, Rightarrow, from=3, to=1]
\end{tikzcd}: h_{\cD,\cA}\Rightarrow \tilde h_{\cD,\cA}\]
    which proves the first part. 

    The second part will use the invertibility of the 2-morphisms $a_\cD,b_\cD$ in \textbf{Proposition \ref{dualevals}}, as well as the duality transforms \eqref{dualbraid} $\,^* C_{\cD^*,\cA} \cong C_{\,^*\cA,\cD}.$ We have that
    \begin{align*}
        \,^* h_{\cD,\cA} &\cong (\,^*\operatorname{cev}_\cA\boxtimes\,^*\operatorname{cev}_\cD)\circ  (\cA\boxtimes \,^* C_{\cD^*,\cA}\boxtimes \,^*\cD) \circ 
        (\,^*\operatorname{cev}_\cA^\dagger\boxtimes\,^*\operatorname{cev}_\cD^\dagger) \\
        &\xRightarrow{\big(b_\cA^{-1}\boxtimes b_\cD^{-1}\big) \circ (\cA\boxtimes \cong\boxtimes \,^*\cD) \circ \big(b_\cA^\dagger\boxtimes b_\cD^\dagger\big)} (\overline{\operatorname{ev}}_\cA\boxtimes\overline{\operatorname{ev}}_\cD)\circ  (\cA\boxtimes C_{\,^*\cA,\cD}\boxtimes \,^*\cD) \circ 
        (\overline{\operatorname{ev}}_\cA^\dagger\boxtimes \overline{\operatorname{ev}}_\cD^\dagger) \\ 
        &\cong (\operatorname{ev}_{\,^*\cA}\boxtimes \operatorname{ev}_{\,^*\cD})\circ  (\cA\boxtimes C_{\,^*\cA,\cD}\boxtimes \,^*\cD) \circ 
        (\operatorname{ev}_{\,^*\cA}^\dagger\boxtimes \operatorname{ev}_{\,^*\cD}^\dagger) = \tilde h_{\,^*\cA,\,^*\cD},
    \end{align*}
    where we have used  \textit{Remark \ref{preduals}} in the final line. Composing the dual of this 2-morphisms with $\Gamma_{\cD,\cA}$ then proves the theorem.
\end{proof}

We note here that it is in general not possible to swap the arguments of the Hopf links without an orientation reversal being performed --- that is, unless the double braiding on $\cD\boxtimes\cA$ happens to be trivializable. We shall see that this occurs precisely in the special case $\cA=\cD$; let us explore this in the following.

\subsubsection{Trivializable false Hopf links}
Letting now $\cA=\cD$, consider the following \textbf{"false" $(\cD^*,\cD)$-Hopf link}
    \begin{equation*}
        h_{\cD}=h_{\cD,\cD} = (\operatorname{cev}_\cD^\dagger\boxtimes\operatorname{cev}_\cD^\dagger) \circ  (\cD\boxtimes C_{\cD^*,\cD}\boxtimes \cD^*) \circ(\operatorname{cev}_\cD\boxtimes\operatorname{cev}_\cD): \cI\rightarrow\cI.
    \end{equation*}
By quasi-Hermiticity \eqref{quasiherm}, the double-braiding is self-dual $C_{\cD^*,\cD}^\dagger\cong C_{\cD^*,\cD}$. The adjoint-inverse of \eqref{eq1}, $(k_\cD^{\dagger})^{-1}:\operatorname{ev}_\cD\circ C_{\cD^*,\cD}\Rightarrow \operatorname{ev}_\cD$, 
allows us to construct a 2-morphism $\mathcal{H}_{\cD,\cD^*}$ trivializing the $(\cD^*,\cD)$-Hopf link, given by the following diagram 
\[\begin{tikzcd}
	\cI & {\cD\cD^*} & {\cD\cD^*} & \cI \\
	& {\cD\cD^*\cD\cD^*} & {\cD\cD^*\cD\cD^*}
	\arrow[from=1-1, to=1-2]
	\arrow[""{name=0, anchor=center, inner sep=0}, "{1_\cI}", curve={height=-30pt}, from=1-1, to=1-4]
	\arrow[from=1-1, to=2-2]
	\arrow[""{name=1, anchor=center, inner sep=0}, from=1-2, to=1-3]
	\arrow[from=1-3, to=1-4]
	\arrow[""{name=2, anchor=center, inner sep=0}, from=2-2, to=1-2]
	\arrow[from=2-2, to=2-3]
	\arrow[""{name=3, anchor=center, inner sep=0}, from=2-3, to=1-3]
	\arrow[from=2-3, to=1-4]
	\arrow[shorten <=9pt, shorten >=9pt, Rightarrow, from=1-1, to=2]
	\arrow["{e_{\operatorname{cev}_\cD}}", shorten <=4pt, shorten >=4pt, Rightarrow, from=1, to=0]
	\arrow["{(k_\cD^\dagger)^{-1}}", shorten <=15pt, shorten >=15pt, Rightarrow, from=2, to=3]
	\arrow[shorten <=9pt, shorten >=14pt, Rightarrow, from=3, to=1-4]
\end{tikzcd}\]
where the triangles on the sides are given by the tensoring of appropriate folds with the snakerators/cusps $\varphi_\cD,{\varphi}^\dagger$. This 2-morphism trivializes a Hopf linking involving $\cD,\cD^*$; similarly for the left-dual $\,^*\cD$. We call these "false" Hopf links, as they can be disentangled.


Focusing on the $(\cD,\cD^*)$-Hopf link and its trivialization $\mathcal{H}_{\cD,\cD^*}$, there is a another trivialization $\mathcal{H}_{\cD,\cD^*}':1_\cI\Rightarrow h_{\cD}$ given by the following diagram
\[\begin{tikzcd}
	\cI & {\cD\cD^*} & {\cD\cD^*} & \cI \\
	& {\cD\cD^*\cD\cD^*} & {\cD\cD^*\cD\cD^*}
	\arrow[from=1-2, to=1-1]
	\arrow[""{name=0, anchor=center, inner sep=0}, from=1-2, to=2-2]
	\arrow[""{name=1, anchor=center, inner sep=0}, from=1-3, to=1-2]
	\arrow[""{name=2, anchor=center, inner sep=0}, from=1-3, to=2-3]
	\arrow[""{name=3, anchor=center, inner sep=0}, "{1_\cI}"', curve={height=30pt}, from=1-4, to=1-1]
	\arrow[from=1-4, to=1-3]
	\arrow[from=1-4, to=2-3]
	\arrow[from=2-2, to=1-1]
	\arrow[from=2-3, to=2-2]
	\arrow[shorten <=9pt, shorten >=9pt, Rightarrow, from=0, to=1-1]
	\arrow["{k_\cD^{-1}}"', shorten <=15pt, shorten >=15pt, Rightarrow, from=2, to=0]
	\arrow["{\iota_{\operatorname{cev}_\cD^\dagger}}"', shorten <=4pt, shorten >=4pt, Rightarrow, from=3, to=1]
	\arrow[shorten <=9pt, shorten >=14pt, Rightarrow, from=1-4, to=2]
\end{tikzcd}\]
constructed from the snakerators/cusps $\varrho^\dagger,\varrho^\dagger$, instead of $\varphi$. The same arguments hold for the left-dual versions of the Hopf links, by making use of the barred-versions of the folds and snakerators/cusps. 

\subsubsection{False Hopf links and the swallowking equation}\label{falsswallow}
In light of quasi-Hermiticity $c_{\cD^*,\cD}^\dagger\cong c_{\cD,\cD^*}$ and the relation $\iota_F^\dagger =e_{F^\dagger}$ for the adjunction-folds, the middle square in $\mathcal{H}_{\cD,\cD^*}'$ is adjoint to that in $\mathcal{H}_{\cD,\cD^*}$. However, the two triangles at the sides are different: $\mathcal{H}_{\cD,\cD^*}$ has $\varphi$'s while $\mathcal{H}_{\cD,\cD^*}'$ has $\varrho$'s. They are, of course, related through the {swallowtail 2-morphisms} $\mathcal{S}_\cD,\mathcal{S}_\cD'$ studied in \S \ref{swallotail}; concisely, we write
\begin{align}
     \mathcal{H}'^\dagger_{\cD^*,\cD} = \mathcal{S}_\cD\circ \mathcal{H}_{\cD,\cD^*}\circ\mathcal{S}_\cD^\dagger,\qquad  \mathcal{H}'^\dagger_{\cD,\cD^*} =                                                                  \mathcal{S}'_\cD\circ\mathcal{H}_{\cD^*,\cD}\circ\mathcal{S}_\cD'^\dagger.\label{hopfswallow}
\end{align}
Here we have kept the 2-morphisms $a,b$ implicit, but they can be easily written in. 

Now on the other hand, \eqref{braidexch} and the ensuing diagram $\Gamma_{\cD,\cD}$ there allow us to exhibit a 2-morphism
\begin{equation*}
    \Lambda_\cD= \mathcal{H}_{\cD^*,\cD}\bullet \Gamma_{\cD,\cD}\bullet \mathcal{H}_{\cD,\cD^*}: 1_\cI\Rightarrow 1_\cI
\end{equation*}
which relates $\mathcal{H}_{\cD,\cD^*}$ with $\mathcal{H}_{\cD^*,\cD}$,
\[\Lambda_\cD=\begin{tikzcd}
	{\mathcal{H}_{\cD,\cD^*}} \\
	{\Gamma_{\cD,\cD}} \\
	{\mathcal{H}_{\cD^*,\cD}}
	\arrow["\bullet"{description}, draw=none, from=1-1, to=2-1]
	\arrow["\bullet"{description}, draw=none, from=2-1, to=3-1]
\end{tikzcd}=\begin{tikzcd}
	\cI & {\cD\cD^*} & {\cD\cD^*} & \cI \\
	& {\cD\cD^*\cD\cD^*} & {\cD\cD^*\cD\cD^*} \\
	& {\cD^*\cD\cD^*\cD} & {\cD^*\cD\cD^*\cD} \\
	\cI & {\cD^*\cD} & {\cD^*\cD} & \cI
	\arrow[from=1-1, to=1-2]
	\arrow[""{name=0, anchor=center, inner sep=0}, "{1_\cI}", curve={height=-30pt}, from=1-1, to=1-4]
	\arrow[""{name=1, anchor=center, inner sep=0}, from=1-1, to=2-2]
	\arrow[""{name=2, anchor=center, inner sep=0}, from=1-2, to=1-3]
	\arrow[""{name=3, anchor=center, inner sep=0}, from=1-3, to=1-4]
	\arrow[""{name=4, anchor=center, inner sep=0}, from=1-3, to=2-3]
	\arrow[""{name=5, anchor=center, inner sep=0}, Rightarrow, no head, from=1-4, to=4-4]
	\arrow[""{name=6, anchor=center, inner sep=0}, from=2-2, to=1-2]
	\arrow[from=2-2, to=2-3]
	\arrow[from=2-3, to=1-4]
	\arrow[""{name=7, anchor=center, inner sep=0}, from=3-2, to=2-2]
	\arrow[from=3-2, to=3-3]
	\arrow[from=3-2, to=4-1]
	\arrow[""{name=8, anchor=center, inner sep=0}, from=3-3, to=2-3]
	\arrow[""{name=9, anchor=center, inner sep=0}, from=3-3, to=4-4]
	\arrow[""{name=10, anchor=center, inner sep=0}, Rightarrow, no head, from=4-1, to=1-1]
	\arrow[""{name=11, anchor=center, inner sep=0}, from=4-1, to=4-2]
	\arrow[""{name=12, anchor=center, inner sep=0}, from=4-2, to=3-2]
	\arrow[""{name=13, anchor=center, inner sep=0}, from=4-2, to=4-3]
	\arrow[""{name=14, anchor=center, inner sep=0}, from=4-3, to=3-3]
	\arrow[from=4-3, to=4-4]
	\arrow[""{name=15, anchor=center, inner sep=0}, "{1_\cI}", curve={height=-30pt}, from=4-4, to=4-1]
	\arrow[shorten <=3pt, shorten >=1pt, Rightarrow, from=1, to=1-2]
	\arrow["{e_{\operatorname{cev}_\cD}}", shorten <=4pt, shorten >=4pt, Rightarrow, from=2, to=0]
	\arrow[shorten <=6pt, shorten >=9pt, Rightarrow, from=4, to=3]
	\arrow["{(k_\cD^\dagger)^{-1}}", shorten <=15pt, shorten >=15pt, Rightarrow, from=6, to=4]
	\arrow[shorten <=11pt, shorten >=11pt, Rightarrow, from=7, to=10]
	\arrow[shorten <=15pt, shorten >=15pt, Rightarrow, from=7, to=8]
	\arrow[shorten <=11pt, shorten >=11pt, Rightarrow, from=8, to=5]
	\arrow[shorten <=6pt, shorten >=9pt, Rightarrow, from=11, to=12]
	\arrow["{\bar k_\cD^{-1}}", shorten <=15pt, shorten >=15pt, Rightarrow, from=12, to=14]
	\arrow[shorten >=3pt, Rightarrow, from=4-3, to=9]
	\arrow["{\iota_{\operatorname{ev}_\cD}}"', shorten <=4pt, shorten >=4pt, Rightarrow, from=15, to=13]
\end{tikzcd}\]
where $\bar k_\cD^{-1}: C_{\cD,\cD^*}\circ\operatorname{cev}_\cD\Rightarrow \operatorname{cev}_\cD$ is the inverse of \eqref{eq2}. 

\begin{proposition}
    The conditions \eqref{hopfswallow} and the swallowking equations \eqref{swallowking} imply that $\Lambda_\cD$ is self-adjoint
    \begin{equation*}
        \Lambda_\cD  = \Lambda_\cD^\dagger.
    \end{equation*}
\end{proposition}
\noindent Note quasi-Hermiticity $c_{\cD^*,\cD}^\dagger\cong c_{\cD,\cD^*}$ implies that we have the contraction
    \begin{equation*}
        \id_{c_{\cD^*,\cD}}\circ \id^\dagger_{c_{\cD^*,\cD}} = \id_{c_{\cD^*,\cD}}\circ \id_{c_{\cD,\cD^*}} = \id_{C_{\cD^*,\cD}}.
    \end{equation*}
\begin{proof}
Recall the observation made below \eqref{doublebraidexchange}. If we paste the above expression for $\Lambda_\cD$ with the swallowtails $\mathcal{S}_\cD,\mathcal{S}'_\cD$ and their adjoints on either side, this observation braid exchange square in the middle of $\Gamma_{\cD,\cD^*}$ fits precisely in between the 2-morphism $\mathcal{K}_\cD$ obtained in \S \ref{swallowkingeqs},
\begin{equation*}
    (\mathcal{S}_\cD\circ\mathcal{K}_\cD\circ \mathcal{S}_\cD'),\qquad (\mathcal{S}_\cD^\dagger\circ\mathcal{K}_\cD^\dagger\circ \mathcal{S}_\cD'^\dagger),
\end{equation*}
where we recall $\mathcal{K}_\cD$ is obtained by gluing the big diamond \eqref{diamond} with \eqref{writheswap}. This allows us to form a 2-morphism expressed schematically by
\[\begin{tikzcd}
	{\mathcal{S}_\cD} & {\mathcal{H}_{\cD,\cD^*}} & {\mathcal{S}_\cD^\dagger} \\
	{\mathcal{K}_\cD} & {\Gamma_{\cD,\cD}} & {\mathcal{K}^\dagger_\cD} \\
	{\mathcal{S}'_\cD} & {\mathcal{H}_{\cD^*,\cD}} & {\mathcal{S}'^\dagger_\cD}
	\arrow["\circ"{description}, draw=none, from=1-1, to=1-2]
	\arrow["\circ"{description}, draw=none, from=1-1, to=2-1]
	\arrow["\circ"{description}, draw=none, from=1-2, to=1-3]
	\arrow["\bullet"{description}, draw=none, from=1-2, to=2-2]
	\arrow["\circ"{description}, draw=none, from=1-3, to=2-3]
	\arrow["\bullet"{description}, draw=none, from=2-1, to=2-2]
	\arrow["\circ"{description}, draw=none, from=2-1, to=3-1]
	\arrow["\bullet"{description}, draw=none, from=2-2, to=2-3]
	\arrow["\bullet"{description}, draw=none, from=2-2, to=3-2]
	\arrow["\circ"{description}, draw=none, from=2-3, to=3-3]
	\arrow["\circ"{description}, draw=none, from=3-1, to=3-2]
	\arrow[draw=none, from=3-2, to=3-3]
	\arrow["\circ"{description}, draw=none, from=3-2, to=3-3]
\end{tikzcd}\]
However, \eqref{hopfswallow} states that this should be nothing but $\Lambda^\dagger_\cD$. The statement then follows from the swallowking \eqref{swallowking}. 
\end{proof}
\noindent If $\tilde R$ were Hermitian instead of just quasi-Hermitian, then by \textit{Remark \ref{syllepsis}} it would mean the entire 2-category is sylleptic, hence all Hopf links are false. Thus, one should read the above proposition as a \textit{coherence condition} imposed on sylleptic 2-categories with duals.

\subsubsection{Order of the Hopf links}\label{truehopf}
In light of the above discussion, we now turn to studying the non-trivializable \textbf{"true" Hopf links} on $\cD$. They are given by the following endofunctors on $\cI$,
\begin{align*}
    &\alpha_\cD= (\operatorname{cev}_\cD^\dagger\boxtimes \overline{\operatorname{cev}}_\cD^\dagger)\circ(\cD\boxtimes C_{\cD^*,\,^*\cD}\boxtimes\cD)\circ (\operatorname{cev}_\cD\boxtimes\overline{\operatorname{cev}}_\cD),\\
    &\bar \alpha_\cD= (\overline{\operatorname{ev}}_\cD\boxtimes {\operatorname{ev}}_\cD)\circ(\cD\boxtimes C_{\,^*\cD,\cD^*}\boxtimes\cD)\circ (\overline{\operatorname{ev}}_\cD^\dagger\boxtimes\operatorname{ev}_\cD^\dagger)\\
    & \beta_\cD = (\operatorname{cev}_\cD^\dagger\boxtimes \operatorname{ev}_{\cD})\circ (\cD\boxtimes C_{\cD^*,\cD^*}\boxtimes\cD)\circ (\operatorname{cev}_\cD\boxtimes\operatorname{ev}_\cD^\dagger),\\ 
    &\bar \beta_\cD =  (\overline{\operatorname{ev}}_\cD\boxtimes\overline{\operatorname{cev}}_\cD^\dagger)\circ (\cD\boxtimes C_{\,^*\cD,\,^*\cD}\boxtimes\cD)\circ (\overline{\operatorname{ev}}_\cD^\dagger\boxtimes\overline{\operatorname{cev}}_\cD).
\end{align*}
Other Hopf links can be obtained by applying interchangers to swap the order of the duality folds. Since each of their composites respects the tensor product, so do they.


Recall form \textbf{Definition \ref{framing}} that being "half-framed" $\vartheta^\dagger\cong\bar\vartheta$ relates the ribbon balancings to their adjoints. For the Hopf links, on the other hand, a similar condition between the $\alpha,\beta$'s has less to do with the framing, but rather the order of its double braiding.
\begin{proposition}\label{selfadjointhopf}
    Consider the following objects in $\tilde C$:
    \begin{gather*}
        \mathfrak{r}_{\alpha}= (-\cdot -)\big((\tilde S^4\otimes 1)(\tilde R~\hat\cdot~\tilde R)\big),\qquad \mathfrak{r}_{\bar\alpha} =(-\cdot -)\big((1\otimes \tilde S^4)(\tilde R~\hat\cdot~\tilde R)\big),\\
        \mathfrak{r}_\beta = \tilde S^4(-\cdot-)(\tilde R~\hat\cdot~\tilde R),\qquad \mathfrak{r}_{\bar\beta} = \tilde S^{-4}(-\cdot-)(\tilde R~\hat\cdot~\tilde R).
    \end{gather*}
    For each $\chi=\alpha,\bar\alpha,\beta,\bar\beta$, if $\overline{\mathfrak{r}}_\chi = \mathfrak{r}_\chi^T$ is Hermitian (cf. \eqref{quasiherm}) then the true Hopf link $\chi_\cD$ is self-adjoint.
\end{proposition}
\begin{proof}
    The key observation is that the double braidings in the definition of $\alpha,\bar\alpha$ involves objects that are two (pre-)duals apart, while those in $\beta,\bar\beta$ involves objects with the same number of (pre-)duals. Since $\tilde S^2$ is a monoidal functor, the quantities $\mathfrak{r}_\chi$ can be seen to implement these double braiding maps in the true Hopf link $\chi_\cD$ for each $\chi=\alpha,\bar\alpha,\beta,\bar\beta$. Thus, under \eqref{dualbraid} and \textit{Remark \ref{syllepsis}}, the Hermiticity of $\mathfrak{r}_\chi$ then implies that these double braids are self-adjoint. The proposition then follows directly.

\end{proof}
\noindent The self-adojointedness of the Hopf links gives a trivialization of its square,
    \begin{equation*}
        e_{\alpha_\cD}:  \alpha_\cD\circ\alpha_\cD\Rightarrow 1_\cI;
    \end{equation*}
    similarly for the other Hopf links. 
    
    In general, the number $n<\infty$ at which the $n$-fold Hopf links\footnote{Note these are the Hopf links constructed from $n$-hold compositions of the double braiding $C_{\cD,\cD^*}^{\circ n}$, not the $n$-fold composition of the Hopf links themselves.} $\chi_{n,\cD}^\dagger \cong \chi_{n,\cD}$ become self-adjoint can be deduced from the order $n$ at which the objects $\mathfrak{r}_\chi^n$ become Hermitian. We posit that $n$ should also be related to the order of the quantum deformation parameter $q$, as a primitive root of unity.

\begin{rmk}\label{braidingorder}
Given the categorical $S^4$-formula for $\tilde C$, as argued for in \textit{Remark \ref{categoricalS4}}, we see that in order for for the Hopf links to be self-adjoint, the quantity $\overline{\mathfrak{r}}$ must be isomorphic to $\mathfrak{r}^T$ through a conjugation by invertible objects. Conversely, suppose we assume $\frak{r}$ is itself quasi-Hermitian, then the order $n$ at which the Hopf links $\chi_\cD$ become trivializable can be deduced from \textit{four times} the order of the antipode functor $\tilde S^{\circ 4n}\cong 1_{\tilde C}$.
\end{rmk}

We emphasize that the hypotheses of \textbf{Proposition \ref{selfadjointhopf}} are sufficient conditions for the self-adjointedness of the Hopf links, but they may not be necessary. We now examine how these true Hopf links acquire additional relations among each other in the next section, depending on how framed $\cD$ is.

\subsubsection{True Hopf links and the ribbon framing}
Suppose $\cD$ is unframed as given in in \textbf{Definition \ref{framing}}, which means in particular that $\cD^*\cong\,^*\cD$ and \eqref{adjdual} holds. From this, we immediately have the following.
\begin{proposition}
    When $\cD$ is unframed, then
    \begin{equation*}
        \alpha_\cD\cong \bar\beta_\cD\cong \bar\alpha_\cD\cong \beta_\cD.
    \end{equation*}
    Further, if $\cD$ were self-dual then there is only one true Hopf link on $\cD$ up to isomorphism.
\end{proposition}
\begin{proof}
    By \textit{Remark \ref{preduals}} and the pivotality condition \eqref{adjdual}, we have
    \begin{equation*}
        \operatorname{ev}^\dagger_\cD = \operatorname{cev}_{\cD^*} \cong \overline{\operatorname{cev}}_\cD,\qquad \operatorname{cev}^\dagger_\cD = \operatorname{ev}_{\cD^*} \cong \overline{\operatorname{ev}}_\cD.
    \end{equation*}
    Further, we have the following 2-isomorphisms
    \begin{equation}
        C_{\cD^*,\,^*\cD} \cong C_{\cD^*,\cD^*}\cong C_{\,^*\cD,\,^*\cD}\cong C_{\,^*\cD,\cD^*},\label{unframedbraid}
    \end{equation}
    whence the first part of the proposition follows. 
    
    Now if $\cD$ were self-dual, then not only do all of the true Hopf links on $\cD$ coincide, but they are also isomorphic to their duals,
    \begin{equation*}
        \alpha_\cD^* \cong \alpha_{\cD^*} =\alpha_\cD
    \end{equation*}
    through the 2-Hopf modifications.
\end{proof}
\noindent Similar to table \ref{tab:framed}, we can also construct a table listing the distinct true Hopf links depending on how framed the object $\cD$ is.

\begin{table}[h]
    \centering
    \begin{tabular}{c|c|c|c}
        $\cD$ is\dots & fully-framed & unframed & self-dual \\
         \hline
     Hopf links  & \makecell{$\chi_\cD, \bar\chi_\cD$ \\
    $\chi^*_\cD, \,^*\bar\chi_\cD$} + adjs. &$\alpha_\cD,\alpha^*_\cD$ + adjs. & $\alpha_\cD$ + adj.
    \end{tabular}
    \caption{A table listing the number of distinct true Hopf links up to 2-isomorphism depending on how framed $\cD$ is. Here, $\chi=\alpha,\beta$ and by "+ adj." we mean that the adjoints of the Hopf links are also in general distinct, unless the hypothesis of \textbf{Proposition \ref{selfadjointhopf}} holds.}
    \label{tab:hopf}
\end{table}

Keep in mind that the self-adjointedness of the Hopf links are a priori independent of the framing. In particular, the tensor unit $\cI$ is \textit{always} self-dual. Moreover, from \textbf{Proposition \ref{symmetric}} we know that $c_{\cI,\cI}\cong 1_\cI$, hence the ribbon balancings and the Hopf links
\begin{equation*}
    \alpha_\cI \cong 2\langle\mathfrak{Dim}(\cI)\rangle 1_{\cI},\qquad \vartheta_\cI \cong \langle\mathfrak{Dim}(\cI)\rangle\cdot 1_\cI
\end{equation*} 
are proportional to the trace of the 2-categorical dimension $\mathfrak{Dim}(\cI)$ defined in \textit{Remark \ref{warning}}. 

\medskip

\section{Pivotality in the classical limit} \label{classicalpivotal}
Throughout the above, we have remarked repeatedly that full pivotality should be recovered in the classical limit. In the final section of this paper, we make this statement precise. In doing so, we will also see that being unframed is closely related to pivotality.

\begin{theorem}
     In the classical limit, $\operatorname{2Rep}(\mathbb{U}_{q=1}\G;\id\otimes\id)$ is pivotal in the sense of \cite{Douglas:2018}.
\end{theorem}
\begin{proof}
    By \textbf{Proposition \ref{planaruni}}, we just need to produce all of the C1-C8 conditions in Definition 2.2.4 of \cite{Douglas:2018}. The conditions C2-C4 have already been demonstrated in \S \ref{duals}, hence it suffices to check C1, C5-C8. We know from \S \ref{2gthopf} and \textit{Remark \ref{classicalpivotal}} that, in the classical limit, the antipode $\tilde S$ is unipotent, and hence gives an identification $\cD\cong (\cD^*)^*$ as $\tilde C$-module categories. This settles C7. 
    
    Note the 2-gauge transformations satisfy $\Lambda_{\gamma_{\bar e}}=\Lambda_{\gamma_e}^{-1}$ under under orientation reversal in the classical limit, whence
    \begin{equation*}
        \Lambda_{\tilde S_v\gamma_e} = \Lambda_{\gamma_{\bar e}} = \Lambda_{\gamma_e}^{-1} = \Lambda_{\tilde S\gamma_e},
    \end{equation*}
    which tells us that the vertical and horizontal antipodes $\tilde {S}_v,\tilde S$ coincide on the edge parameters/1-morphisms $\cE$. Since $\tilde S_v$ determines the adjunction and $\tilde S$ determines duality, the planar-pivotality of the folds imply that we have
    \begin{equation*}
        \overline{\operatorname{cev}}_\cD\cong \operatorname{ev}^\dagger_\cD,\qquad \overline{\operatorname{ev}}_\cD\cong \operatorname{cev}_\cD^\dagger;
    \end{equation*}
    however, $\overline{\operatorname{c/ev}}_\cD \cong \operatorname{c/ev}_{\cD^*}$ whenever the dual is involutive, whence C5-C6 follow. From the fact that the braiding is trivial (ie. merely given by the flip functor), the 2-Drinfel'd transformations $\omega_\cD,\bar\omega_\cD$ are trivial as well, and in particular invertible. C8 then follows from {\bf Proposition \ref{2drinfeld}}.

    Now it remains to recover C1, the swallowtail equations. We shall do this from the swallowking equations described in \S \ref{swallowkingeqs}. Since the braiding is trivial, so is the writhing. The 2-morphism \eqref{writheswap} simply takes the form $\upsilon_{\operatorname{ev}_\cD,\operatorname{cev}_\cD}^\dagger \circ \operatorname{id}_\text{flip}\circ\upsilon_{\operatorname{cev}_\cD,\operatorname{ev}_\cD}$. The unitarity of the interchangers then reduces the 2-morphism $\mathcal{K}_\cD$, which we recall is obtained by attaching the left- and right-sides of the diamond \eqref{diamond} with \eqref{writheswap}, to $\operatorname{id}_\text{flip}$. The swallowking equations \eqref{swallowking} then take the form
    \begin{equation*}
        \mathcal{S}_\cD\circ \operatorname{id}_\text{flip}\circ\mathcal{S}_\cD' = \id_\text{flip}.
    \end{equation*}
    This forces each of the swallowtails $\mathcal{S}_\cD,\mathcal{S}'_\cD$ to be trivial.
\end{proof}

It can be seen from the proof that every object in $\operatorname{2Rep}(\mathbb{U}_{q=1}\G;\id\otimes\id)$ is \textit{unframed} in the sense of \textbf{Definition \ref{framing}}, whence \eqref{adjdual} tells us that the object-level pairing convention is given by the pivotal one chosen in \cite{Douglas:2018}.  This result is consistent with the fact that $\operatorname{2Rep}(\mathbb{G})$ for {\it finite} 2-groups $\mathbb{G}$ --- which, like ordinary finite 1-groups, have no non-trivial deformation --- are known to be pivotal \cite{Douglas:2018,Bartsch:2023wvv}. In fact, such 2-representation 2-categories are expected to be spherical, hence we also expect $\operatorname{2Rep}(\tilde C;\tilde R)$ to become spherical in the classical limit.

\medskip

We conclude this paper by making the following remark. From \textit{Remark \ref{warning}} and table \ref{tab:framed}, one can interpret the cause of the issue raised in Warning 2.2.5 of \cite{Douglas:2018} as the fact that unframed objects are unable to tell barred and unbarred functors apart. As such, the data of being "unframed", namely the fixed chosen invertible 2-Drinfel'd modification which identifies $\vartheta$ and $\bar\vartheta$, gives rise to a notion of "unframed equivalence" (ie. "pivotal adjoint equivalence" mentioned in Remark 2.3.9 of \cite{Douglas:2018}) of objects preserving this data. In general, this notion relaxes to the centralizer subcategory $C_{\operatorname{End}(\cD)}(\vartheta)$ of the twists mentioned in \textit{Remark \ref{pivottwist}}.

\section{Conclusions}
We have studied in this paper the notion of a "ribbon tensor 2-category" arising naturally out of the categorical quantum symmetries of the 4d 2-Chern-Simons theory on a lattice. Much of the properties and the coherence conditions were derived, but the author would like to emphasize that "ribbon tensor 2-categories" should exist in contexts much more general than just categorical quantum groups. In particular, such structures should at least have framing properties that leads to a "$SO(3)$-volutive" refinement of strict pivotality studied in \cite{Douglas:2018}. It is the hope of the author to continue and use the structures examined here in order to construct interesting 4d TQFTs and invariants in the future.

Though, that is not to say that the only novelty of this paper is to give an example of such a ribbon 2-category. It also leads to very interesting implications in the study of 4d and 3d TQFTs. For instance, as we have noted in \textit{Remark \ref{half-framed}} and Table \ref{tab:framing}, ribbon tensor 2-categories give rise to $\mathsf{Gray}$-categories with duals. As mentioned also in the introduction, such structures were shown to describe (locally) non-extended 3d defect TQFTs \cite{Carqueville:2016kdq,Carqueville:2018sld,Koppen:2021kry,Carqueville:2023aak}, and serves as the natural foundation for their orbifold defect data. As such, $\operatorname{2Rep}(\tilde C;\tilde R)$ provides a way in which the 3d defect TQFTs can be parameterized, through different choices of Lie 2-groups, by the categorical quantum symmetries arising from the 4d 2-Chern-Simons TQFT. 

This observation can also be understood as the manifestation of a 4d-3d topological bulk-boundary relation, as formulated in the setting of topological orders and quantum liquids in the series \cite{Kong:2020wmn,Kong:2020iek,Kong2024-vr} of papers. These papers relied heavily on the theory of the so-called \textit{separable $n$-categories} and the condensation completion functor \cite{Gaiotto:2019xmp}. The prevailing philosophy in the literature (see the above references, as well as eg. \cite{KitaevKong_2012,Kong:2020}) is that, given a modular tensor category describing a 3d topological order, its condensation completion describes a 4d topological bulk admitting it as boundary condition. We see that this perspective leads to a very interesting prospect for understanding the 4d-3d correspondence for the Reshetikhin-Turaev TQFT. 

\paragraph{The algebraic Baez conjecture.} Take the modular tensor category $\operatorname{Rep}U_q\frak{sl}_2$ underlying the Reshetikhin-Turaev TQFT \cite{Reshetikhin:1990pr,Reshetikhin:1991tc}. Its condensation completion, or simply the pre-modular tensor 2-category $\operatorname{Mod}(\operatorname{Rep}U_q\frak{sl}_2)$ of its module categories, is believed to underlie the 4d Crane-Yetter TQFT \cite{Douglas:2018}. Recent results \cite{haioun2025nonsemisimplewrtboundarycraneyetter} have shown that the Reshetikhin-Turaev TQFT does indeed exist as a boundary condition of the Crane-Yetter TQFT.

On the other hand, starting from the classical theory, it can be shown \cite{Chen:2024axr} that the 3d Chern-Simons action with gauge group $G$ lives on the boundary of the 4d 2-Chern-Simons action on the \textit{inner automorphism 2-group} $\operatorname{Inn}G=G\xrightarrow{\id}G$. This is in fact also true \textit{semi}classically in the BFV formalism \cite{Cattaneo:2016zsq}: the Chern-Simons pre-symplectic form also appears at the boundary of that of the 2-Chern-Simons theory. 

Given the skein-theoretic quantization of Chern-Simons theory as the (Witten-)Reshetikhin-Turaev TQFT \cite{WITTEN1990285}, it is therefore natural to ask the following question.
\begin{quote}
\centering
 {\em At the quantum level, how is the 2-Chern-Simons TQFT is related to the Crane-Yetter TQFT?} 
\end{quote}
\noindent As mentioned in \cite{Chen1:2025?}, this question relates to a statement/conjecture made by Baez in \cite{Baez:1995ph}, which presupposes that the 2-Chern-Simons theory on $\operatorname{Inn}SU(2)$ is equivalent to the $SO(4)$ Crane-Yetter TQFT. 

The algebraic version of this statement is then a (pre-modular) equivalence
\begin{equation*}
    \operatorname{2Rep}(\mathbb{U}_q\operatorname{inn}\mathfrak{sl}(2)) \simeq \operatorname{Mod}(\operatorname{Rep}(U_q\mathfrak{sl}_2)),
\end{equation*}
which given the results of this paper can be directly checked. This will be the subject of a future work.


\newpage

\newpage

\printbibliography

\end{document}